\documentclass[12pt]{article}
\usepackage{amsmath}
\usepackage{graphicx,psfrag,epsf}
\usepackage{enumerate}
\usepackage{natbib}
\allowdisplaybreaks
\usepackage{url} 

\newcommand{\blind}{0}

\usepackage[margin=1in, a4paper]{geometry}
\usepackage{adjustbox}
\usepackage{enumerate}
\usepackage{bm}
\usepackage{amsfonts}
\usepackage[colorlinks,citecolor=blue,urlcolor=blue,hyperfootnotes=false]%
{hyperref}
\usepackage{amsthm}
\usepackage{amssymb}
\usepackage{graphicx}
\usepackage{float}
\usepackage{booktabs}
\usepackage{threeparttable}
\usepackage{amsmath}
\numberwithin{equation}{section}
\usepackage{lscape}
\usepackage{scrextend}
\usepackage{relsize}
\usepackage{multicol}
\usepackage{chngcntr}
\usepackage{etoolbox}
\usepackage{amssymb}
\usepackage{hyperref}
\usepackage{tabularx}
\usepackage{mathrsfs}
\usepackage[english]{babel}
\usepackage{multirow}
\usepackage{booktabs}
\usepackage[labelfont=bf]{caption}
\usepackage{float}
\usepackage{titlesec}
\usepackage{array}
\usepackage{color}
\usepackage{afterpage}
\usepackage{framed}
\usepackage{pifont} 
\usepackage{float}
\usepackage{afterpage}%
\usepackage{comment}
\usepackage{multirow}
\usepackage{accents}
\usepackage{amssymb}
\usepackage{xcolor}
\usepackage{extpfeil}
\usepackage{natbib}
\setcitestyle{open={(},close={)}}

\newtheorem{theorem}{Theorem}[section]
\newtheorem{corollary}{Corollary}[section]

\newtheorem{lemma}{Lemma}[section]

\theoremstyle{definition} 
\newtheorem{assumption}{Assumption}[section]

\newtheorem{case}{Case}[section]

\addto\extrasenglish{%
   
}
\addto\extrasenglish{%
   
}
\addto\extrasenglish{%
  
}
\addto\extrasenglish{%
  
}
\addto\extrasenglish{%
  
}
\addto\extrasenglish{%
  
}
\addto\extrasenglish{%
  
}
\addto\extrasenglish{%
  
}

\hypersetup{
  linkcolor=blue,
  urlcolor=blue}

\begin{document}

\def\spacingset#1{\renewcommand{\baselinestretch}%
{#1}\small\normalsize} \spacingset{1}


\if0\blind
{
  \title{\bf Paired Sample Tests for High-dimensional Uncorrelatedness via Random Integration}
  \author{Shiyao Huang\thanks{Department of Business Statistics and Econometrics, Guanghua School of Management, Peking University, Beijing, 100871, China. Email: \texttt{2401111054@gsm.pku.edu.cn}.}\hspace{.2cm}\\
    Peking University\\
    and \\
    Xiaojun Song\thanks{Corresponding author: Department of Business Statistics and Econometrics, Guanghua School of Management, Peking University, Beijing, 100871, China. Email: \texttt{sxj@gsm.pku.edu.cn}. 
This work was supported by the National Natural Science Foundation of China [Grant Numbers 72373007 and 72333001]. The author also gratefully acknowledges the research support from the Center for Statistical Science of Peking University, and the Key Laboratory of Mathematical Economics and Quantitative Finance (Peking University) of the Ministry of Education, China.} \\
    Peking University}
  \maketitle
} \fi

\if1\blind
{
  \bigskip
  \bigskip
  \bigskip
  \begin{center}
    {\LARGE\bf Paired Sample Tests for High-dimensional Covariance Matrices via Random Integration}
\end{center}
  \medskip
} \fi

\bigskip
\begin{abstract}
This paper proposes a novel nonparametric test to assess the uncorrelatedness between two high-dimensional random vectors. We develop our test by generalizing the random integration proposed by \citet{jiang2023use,jiang2024nonparametric}, and the resulting test statistic estimates a weighted squared $\mathscr{L}_2$ norm of the covariance matrix. Asymptotic properties of the test statistic are derived by letting both the sample size $n$ and the dimension $p$ diverge to infinity. Under the null hypothesis of uncorrelatedness, our proposed test statistic is asymptotically normal with zero mean and unit variance, without requiring any specification of the relative magnitude regarding $n$ and $p$. 
Monte Carlo simulations demonstrate the good finite-sample performance of our proposed methods. Compared with many existing tests, our test statistic is more powerful at detecting ``weak but pervasive'' dependence while maintaining a comparable empirical size. The advantages of the proposed methods are further illustrated by an empirical analysis that assesses the correlation between DNA methylation and gene expression.
\end{abstract}

\noindent%
\vfill

\newpage
\spacingset{1.8} 

\section{Introduction}
\label{Introduction}
Measuring the level of association between a pair of random vectors is a fundamental problem in statistical analysis, with wide applications in economics \citep{gorecki2020independence}, biological studies \citep{zhang2024fast}, machine learning \citep{li2023deep}, and many other fields. To quantify dependence structures, a large body of literature has proposed various measures, including the Pearson correlation \citep{pearson1895vii}, rank correlations \citep{kendall1938new,spearman1961proof}, correlations based on distribution functions, density functions, and characteristic functions \citep{blum1961distribution,hoeffding1994non,szekely2007measuring}, kernel-based dependence measure \citep{gretton2005kernel}, and sign covariances \citep{bergsma2014consistent}. For recent development, see also \citet{berrett2020conditional} and \citet{shah2020hardness}.

Among the various approaches to measuring dependence, testing for independence is perhaps one of the most fundamental problems and has attracted considerable attention from researchers. For the fixed-dimensional case, methods for independence assessment include tests based on distance covariances \citep{szekely2007measuring,szekely2009brownian}, sign covariances \citep{bergsma2014consistent}, ball covariances \citep{pan2020ball}, angle covariances \citep{jiang2026robust}, the Hilbert-Schmidt Independence Criterion \citep{gretton2005measuring,gao2025studentized}, ranks \citep{weihs2018symmetric,deb2023multivariate}, and copulas \citep{kojadinovic2009tests,genest2019testing}. In recent decades, many efforts have been devoted to testing independence in the high-dimensional regime. For example, \citet{szekely2013distance} modified the distance correlation and established asymptotic theories by fixing the sample size while allowing the dimension to diverge to infinity. \citet{zhu2020distance} investigated distance- and reproducing-kernel-based dependence metrics and proposed aggregating marginal sample distances to detect nonlinear dependence. Later on, \citet{gao2021asymptotic} extended the framework of \citet{szekely2013distance} by allowing both the sample size and the dimension to diverge arbitrarily and further established convergence rates for the corresponding normal approximations. \citet{zhou2024rank} proposed tests for high-dimensional independence based on rank-based indices derived from Hoeffding's $D$, Blum--Kiefer--Rosenblatt's $R$, and Bergsma--Dassios--Yanagimoto's $\tau^*$. More recently, \citet{wang2026testing} developed rank-based max-type and sum-type tests, whereas \citet{jiang2026distribution} introduced the semi-Grothendieck's covariance for independence testing.

Despite their popularity, high-dimensional independence tests typically degenerate to tests for uncorrelatedness when the dimension $p$ grows much faster than the sample size $n$ \citep{zhu2020distance,gao2021asymptotic}. Moreover, \citet{gao2021asymptotic} showed that preserving the ability of distance correlation to detect nonlinear dependence requires the dimension to grow at a sufficiently slow rate relative to the sample size ($p^2/n\rightarrow 0$). Thus, in the high-dimensional regime, it is natural and reasonable to focus directly on testing uncorrelatedness when both $p$ and $n$ diverge in an arbitrary manner. However, the existing literature on high-dimensional uncorrelatedness remains rather limited. In this paper, we aim to bridge this gap by proposing a family of test statistics based on the random integration \citep{jiang2023use,jiang2024nonparametric}. Among existing works, the proposal most closely related to ours is that of \citet{yang2015independence}, who employed canonical correlation to test uncorrelatedness. Nevertheless, their analysis is limited to the regime where the dimension is at most comparable to the sample size ($p/n\rightarrow c\in(0,\infty)$), thereby excluding the high-dimensional, low-sample-size setting. In addition, their procedure assumes that the mean vectors of the underlying distributions are known, which is generally unrealistic in practice. In sharp contrast, our proposed test allows the dimensionality to grow much faster than the sample size ($p/n\rightarrow \infty$) and is self-centralized, requiring no prior knowledge of the mean vectors.

We construct our tests based on the Random Integration of Differences (RID). Originally introduced by \citet{jiang2023use,jiang2024nonparametric}, RID has been successfully applied to two-sample testing problems, including tests for the equality of mean vectors and variance matrices, yet its applicability to paired data remains unexplored. RID can be viewed as a generalization of traditional $\mathscr{L}_2$-norm-based tests, including those proposed by \citet{chen2010two}, Chen et al. (2010), and \citet{li2012two}. However, the standard $\mathscr{L}_2$ norm assigns equal weights to all entries, which may lead to a loss of power against ``weak but pervasive'' alternatives, where pairwise correlations are individually weak but collectively exhibit a globally pervasive dependence structure. To address this issue, RID introduces a weighting matrix to enhance empirical power and offer empirical researchers greater flexibility in constructing test statistics tailored to specific applications. Our proposed RID-based test for uncorrelatedness has the following advantages. First, it does not require a direct relationship between $p$ and $n$. Second, its limiting distribution is standard normal under the null hypothesis of uncorrelatedness, with critical values readily available. Third, compared with existing tests, it delivers comparable empirical size and is more powerful in detecting ``weak but pervasive'' dependence, which is commonly encountered in gene studies \citep{yang2010common}.

The rest of this paper is organized as follows. In \autoref{Testing procedures}, we introduce the RID and construct the main building block of our test. The asymptotic properties of the proposed test statistic under the null and alternative hypotheses are systematically analyzed in \autoref{Asymptotic theory}. \autoref{Simulations} examines the finite-sample performance of the proposed methods through numerical simulations. An empirical application examining the correlation between DNA methylation and gene expression levels is presented in \autoref{Empirical analysis}. \autoref{Concluding comments} concludes the paper. Proofs of theoretical results are relegated to the online supplementary material. Throughout this article, we use $\stackrel{d}{\longrightarrow}$ to denote convergence in distribution, $\stackrel{p}{\longrightarrow}$ to denote convergence in probability, and $\longrightarrow$ to denote convergence of sequences.

\section{Testing procedures}
\label{Testing procedures}

Suppose ${X}$ and ${Y}$ are two $p$-dimensional random vectors, with means ${\mu}_1 = \mathbb{E}({X})$, ${\mu}_2 = \mathbb{E}({Y})$ and variance matrices ${\Sigma}_1 = \mathbb{E}[({X} - {\mu}_1)({X} - {\mu}_1)']$, ${\Sigma}_2 = \mathbb{E}[({Y} - {\mu}_2)({Y} - {\mu}_2)']$, respectively. Let ${\Sigma}$ denote the covariance matrix ${\Sigma} = \mathbb{E}[({X} - {\mu}_1)({Y} - {\mu}_2)']$. In this paper, we aim to test whether ${X}$ and ${Y}$ are uncorrelated, i.e., 
\begin{align}
\mathcal{H}_0:\ {\Sigma}={O}_{p} \quad \text{against} \quad \mathcal{H}_1:\ {\Sigma}\neq {O}_{p},
\label{testing problem}
\end{align}
where ${O}_{p}$ is the $p\times p$ zero matrix.

We motivate our test based on an important equivalence. On one hand, for any ${\xi},{\eta}\in\mathbb{R}^p$, ${\Sigma}= {O}_p$ implies that ${\xi}^\prime {\Sigma}{\eta} = 0$. On the other hand, denote $\sigma_{ij}$ the $(i,j)$th element of ${\Sigma}$ and ${e}_i$ the $p$-dimensional vector with $i$th element equal to 1 and others equal to 0. Note that $\sigma_{ij} = {e}_i'{\Sigma}{e}_j$. Then, $ e_i^\prime{\Sigma} e_j=0\,\,\,\forall i,j=1,2,\ldots,p$ yields that ${\Sigma} = {O}_p$. As such, we derive the equivalence
\begin{align}
{\Sigma}={O}_{p}\iff {\xi}^\prime {\Sigma}{\eta}=0,\ \forall {\xi},{\eta}\in\mathbb{R}^p.
\label{equivalence}
\end{align}
Following the RID framework of \citet{jiang2023use,jiang2024nonparametric} and the equivalence in \eqref{equivalence}, testing $\mathcal{H}_0$ can be reformulated as assessing whether
\begin{align*}
Q := \int_{\mathbb{R}^{2p}} ({\xi}^\prime {\Sigma}{\eta})^2 w({\xi}) w({\eta}) d{\xi} d{\eta} = 0,
\end{align*}
where $w: \mathbb{R}^p \to [0,\infty)$ is a weight function satisfying $0 < \int_{\mathbb{R}^p} w({x}) d{x} < \infty$. $Q$ can be regarded as a Cram\'{e}r--von Mises type statistic \citep{anderson1952asymptotic} as it integrates over all possible directions. Suppose we have $n$ independent and identically distributed (i.i.d.) observations $\{({X}_i,{Y}_i)\}_{i=1}^n$ from $({X},{Y})$, a natural idea to build the test statistic may be first estimating ${\Sigma}$ via
\begin{align*}
\breve{{\Sigma}}_n = \frac{1}{n}\sum_{i=1}^n({X}_i - \bar{{X}})({Y}_i - \bar{{Y}}),\,\,\,\text{where}\,\,\,\bar{X} = \frac{1}{n}\sum_{i=1}^n{X}_i\,\,\,\text{and}\,\,\,\bar{Y} = \frac{1}{n}\sum_{i=1}^nY_i,
\end{align*}
and then constructing $\breve{Q}_n = \int_{\mathbb{R}^{2p}} ({\xi}^\prime \breve{{\Sigma}}_n{\eta})^2 w({\xi}) w({\eta}) d{\xi} d{\eta}$. However, without specifying any function structures of $w({x})$, the calculation of $\breve{Q}_n$ may be computationally burdensome, especially for large $p$. Luckily, the following \autoref{Th2.1} shows that, by restricting $w({x})$ as a product of $p$ univariate density functions, $Q$ admits a closed-form expression, thereby avoiding the need to calculate high-dimensional numerical integrations.

\begin{theorem}
\label{Th2.1}
Suppose $w({x})=\prod_{i=1}^pw_i(x_i)$, where $w_i(x_i)$ is a density function with mean $a_i=\int_{-\infty}^\infty tw_i(t)dt$ and variance $b_i^2=\int_{-\infty}^\infty (t-a_i)^2w_i(t)dt$ for $i=1,2,\ldots,p$. Then, 
\begin{align}
Q=\textnormal{tr}({W}{\Sigma}{W}{\Sigma}^\prime),
\label{simplified expression}
\end{align}
where ${W}={B}+{a}{a}^\prime$ with ${a}=(a_1,a_2,\ldots,a_p)^\prime$ and ${B}=\textnormal{diag}(b_1^2,b_2^2,\ldots,b_p^2)$.
\end{theorem}
\eqref{simplified expression} can be viewed as a weighted version of the squared $\mathscr{L}_2$ norm $\mathrm{tr}({\Sigma}{\Sigma}^\prime)$. However, both the subsequent theoretical analysis and numerical experiments demonstrate that the test based on $\mathrm{tr}({\Sigma}{\Sigma}^\prime)$ suffers from a severe loss of power against ``weak but pervasive'' alternatives. Under such scenarios, the individual elements of the covariance matrix are vanishingly small, yet the matrix itself is highly dense with few zero entries. Thus, incorporating a weighting matrix ${W}$ into the discrepancy measure allows different components of the covariance structure to contribute unequally to the test statistic, thereby enhancing the test's flexibility in adapting to various dependence structures and improving its ability to capture signals with different patterns. Following (\ref{simplified expression}), we decompose $Q$ as the linear combination of three expectations:
\begin{align*}
Q=&\textnormal{tr}\left({W}\mathbb{E}[({X}_i-{\mu}_1)({Y}_i-{\mu}_2)']{W}\mathbb{E}[({Y}_j-{\mu}_2)({X}_j-{\mu}_1)']\right)\\
=&\textnormal{tr}({W}\mathbb{E}({X}_i{Y}_i^\prime){W}\mathbb{E}({Y}_j{X}_j^\prime))-2\textnormal{tr}({W}\mathbb{E}({X}_i{Y}_i^\prime){W}{\mu}_2{\mu}_1^\prime)+\textnormal{tr}({W}{\mu}_1{\mu}_2^\prime{W}{\mu}_2{\mu}_1^\prime)\\
=&\mathbb{E}({X}_i^\prime{W}{X}_j{Y}_j^\prime{W}{Y}_i)-2\mathbb{E}({X}_i^\prime{W}{X}_j{Y}_j^\prime{W}{Y}_k)+\mathbb{E}({X}_i^\prime{W}{X}_j{Y}_k^\prime{W}{Y}_l)\\
:=& Q_1-2Q_2+Q_3,
\end{align*}
where $({X}_i,{Y}_i)$, $({X}_j,{Y}_j)$, $({X}_k,{Y}_k)$, $({X}_l,{Y}_l)$ are independent copies of $({X},{Y})$. Therefore, with the observed paired data $\{({X}_i,{Y}_i)\}_{i=1}^n$, we construct the estimator for $Q$ as $\hat{Q}_n=\hat{Q}_{n1}-2\hat{Q}_{n2}+\hat{Q}_{n3}$, where 
\begin{align*}
&\hat{Q}_{n1}=\frac{1}{n(n-1)}\sum_{i\neq j}^n {X}_i^\prime{W}{X}_j{Y}_j^\prime{W}{Y}_i,\ \ \hat{Q}_{n2}=\frac{1}{n(n-1)(n-2)}\sum_{i\neq j\neq k}^n {X}_i^\prime{W}{X}_j{Y}_j^\prime{W}{Y}_k,\\
&\hat{Q}_{n3}=\frac{1}{n(n-1)(n-2)(n-3)}\sum_{i\neq j\neq k\neq l}^n {X}_i^\prime{W}{X}_j{Y}_k^\prime{W}{Y}_l.
\end{align*}
Obviously, $\hat{Q}_n$ is an unbiased estimator for $Q$. Furthermore, it is worthwhile to note that $\hat{Q}_n$ is invariant to location shifts. In other words, with ${A}_i={X}_i-{\mu}_1$, ${B}_i={Y}_i-{\mu}_2$, $\hat{Q}_n$ can be automatically centralized as $\hat{Q}_n=\tilde{Q}_{n1}-2\tilde{Q}_{n2}+\tilde{Q}_{n3}$, where 
\begin{align*}
&\tilde{Q}_{n1}=\frac{1}{n(n-1)}\sum_{i\neq j}^n {A}_i^\prime{W}{A}_j{B}_j^\prime{W}{B}_i,\ \ \tilde{Q}_{n2}=\frac{1}{n(n-1)(n-2)}\sum_{i\neq j\neq k}^n {A}_i^\prime{W}{A}_j{B}_j^\prime{W}{B}_k,\\
&\tilde{Q}_{n3}=\frac{1}{n(n-1)(n-2)(n-3)}\sum_{i\neq j\neq k\neq l}^n {A}_i^\prime{W}{A}_j{B}_k^\prime{W}{B}_l.
\end{align*}
Therefore, our proposed methods require no prior information regarding the mean vectors, and thus overcome a key limitation of \citet{yang2015independence}. In the next section, we systematically investigate the asymptotic properties of $\hat{Q}_n$ and develop the corresponding test statistic.

\section{Asymptotic theory}
\label{Asymptotic theory}

In this section, we establish the asymptotic properties of $\hat{Q}_n$ under the null hypothesis $\mathcal{H}_0$ and the alternative hypothesis $\mathcal{H}_1$. We construct our final test statistic based on $\hat{Q}_n$ and demonstrate that it is asymptotically standard normal under the null hypothesis of uncorrelatedness.

\subsection{Asymptotic null distribution}
\label{Asymptotic null distribution}

To investigate the asymptotic properties of $\hat{Q}_n$, we make the following assumptions.

\begin{assumption}
\label{A1}
There exist two  $p \times d$ matrices $\Gamma_1$ and $\Gamma_2$, and $d$-dimensional random vectors $\{{Z}_{i}\}_{i=1}^n$, such that $ X_i = \mu_1 + \Gamma_1 {Z}_{i}$, $ Y_i = \mu_2 + \Gamma_2 {Z}_{i}$ for $i = 1, 2,\ldots, n$. $\Gamma_1$, $\Gamma_2$, and $\{{Z}_{i}\}_{i=1}^n$ satisfy:
\begin{enumerate}[(1)]
    \item $\Gamma_1 \Gamma_1^\prime = \Sigma_1$, $\Gamma_2 \Gamma_2^\prime = \Sigma_2$, and ${\Gamma}_1{\Gamma}_2^\prime={\Sigma}$, with $d \geqslant 2p$ and $\text{rank}({\Gamma}_1) = \text{rank}({\Gamma}_2) = p$.
    \item $\{{Z}_{i}\}_{i=1}^n$ are independent and identically distributed, with $\mathbb{E}({Z}_{i}) = 0_{d}$, $\text{Var}({Z}_{i}) =  I_{d}$, where $0_{d}$ is the $d$-dimensional zeros and $ I_{d}$ is the $d$-dimensional identity matrix.
    \item $\mathbb{E}(Z_{ij}^4) = 3 + \Delta$ for some constant $\Delta$, where $Z_{ij}$ is the $j$th component of ${Z}_{i}$. Each $Z_{ij}$ has a finite 8th moment. Also, for a positive integer $q$ and nonnegative integers $\zeta_j$'s such that $\sum_{j=1}^q \zeta_j \leq 8$,
    \begin{align}
    \mathbb{E}(Z_{ij_1}^{\zeta_1} Z_{ij_2}^{\zeta_2}\cdots Z_{ij_q}^{\zeta_q}) = \mathbb{E}(Z_{ij_1}^{\zeta_1})\mathbb{E}(Z_{ij_2}^{\zeta_2}) \cdots \mathbb{E}(Z_{ij_q}^{\zeta_q}),
    \label{A2(3)}
    \end{align}
    where $j_1, j_2,\ldots, j_q$ are mutually distinct indices.
\end{enumerate}
\end{assumption}

\begin{assumption}
\label{A2}
As $n\rightarrow \infty$, $p=p(n)\rightarrow\infty$, $\textnormal{tr}(({W}{\Sigma}_i)^2)\rightarrow \infty$, $i=1,2$, and 
\begin{align}
\textnormal{tr}(({W}{\Sigma}_i)^4)=o\big(\textnormal{tr}^2(({W}{\Sigma}_i)^2)\big), \ i=1,2.
\label{equation A2}
\end{align}
\end{assumption}

\autoref{A1} provides a highly general framework for generating the paired population $({X},{Y})$, since the dimension of ${Z}_i$ is allowed to be arbitrary provided that $d\geqslant 2p$. It is worth noting that \autoref{A1}(1) is stronger than the assumptions commonly imposed in the two-sample testing literature \citep{chen2010two,chen2010tests,li2012two,jiang2023use,jiang2024nonparametric}, which typically only require $d\geqslant p$. Here, $d\geqslant 2p$ arises mainly due to the paired-sample structure. To provide some intuition, recall that the positive definiteness of ${\Sigma}_1$ and ${\Sigma}_2$ requires the row vectors of each ${\Gamma}_j$, $j=1,2$, to be linearly independent. Moreover, under the null hypothesis, we have ${\Sigma}={O}_p$, implying that the row vectors of ${\Gamma}_1$ and ${\Gamma}_2$ must also be mutually linearly independent. Consequently, under the null hypothesis, there exist $2p$ linearly independent vectors in a $d$-dimensional space, which necessarily requires $d\geqslant 2p$. \autoref{A1}(2)--(3) only require some moment restrictions about ${Z}_i$, and the distribution of ${Z}_i$ is unspecified. (\ref{A2(3)}) provides a kind of pseudo-independence for each ${Z}_{ij}$ and reveals that the components of ${Z}_i$ are largely white noises. Note that the model provided in \autoref{A1} is a popular setup in high-dimensional tests; see, for instance, \citet{bai1996effect}, \citet{chen2010two}, Chen et al. (2010), \citet{li2012two}, and \citet{jiang2023use,jiang2024nonparametric}.

\autoref{A2} generalizes Assumption A2 in \citet{li2012two}, and does not require any relative magnitude regrading $p$ and $n$. To gain some insights into \autoref{A2}, consider a special case where $a_1=a_2=\ldots=a_p=a\neq 0$, $b_1=b_2=\ldots=b_p=b>0$, $r=a/b$, and $\lambda_1\leqslant \lambda_2\leqslant\ldots\leqslant \lambda_p$ and $\gamma_{11}\leqslant \gamma_{12}\leqslant\ldots\leqslant \gamma_{1p}$ are eigenvalues of ${W}$ and ${\Sigma}_1$, respectively. By some simple calculations, we obtain $\lambda_1=\ldots=\lambda_{p-1}=b^2$, $\lambda_p=b^2+pa^2$. Then, as shown in \citet{jiang2024nonparametric}, under such parametric setups for ${W}$, 
\begin{align*}
\frac{\text{tr}(( W \Sigma_1)^4)}{\text{tr}^2(( W \Sigma_1)^2)} 
\leq 
\frac{\lambda_1^4 \text{tr}( \Sigma_1^4) + (\lambda_p^4 - \lambda_1^4) \gamma_{1p}^4}{\lambda_1^4 \text{tr}^2( \Sigma_1^2)}
=
\frac{\text{tr}( \Sigma_1^4)}{\text{tr}^2( \Sigma_1^2)} 
+ 
\frac{[(1 + p r^2)^4 - 1] \gamma_{1p}^4}{\text{tr}^2( \Sigma_1^2)}.
\end{align*}
Therefore, as long as $\textnormal{tr}({\Sigma}_1^4)=o\big(\textnormal{tr}^2({\Sigma}_1^2)\big)$, $r^2=O(p^{-3/4})$, and $p^{1/2}\gamma_{1p}^2=o(\textnormal{tr}\big({\Sigma}_1^2)\big)$, \autoref{A2} holds for ${\Sigma}_1$. As pointed out in Chen et al. (2010), $\textnormal{tr}({\Sigma}_1^4)=o\big(\textnormal{tr}^2({\Sigma}_1^2)\big)$ is in fact a mild condition. As long as all eigenvalues of ${\Sigma}_1$ are bounded away from zero and infinity, $\textnormal{tr}({\Sigma}_1^4)=o\big(\textnormal{tr}^2({\Sigma}_1^2)\big)$ is trivially true when $p\rightarrow \infty$. Meanwhile, some commonly encountered variance matrices satisfy $\textnormal{tr}({\Sigma}_1^4)=o\big(\textnormal{tr}^2({\Sigma}_1^2)\big)$. For example, ${\Sigma}_1=(\rho^{|i-j|})_{p\times p}$, ${\Sigma}_1=(\sigma_i\sigma_j\rho^{|i-j|})_{p\times p}$, and ${\Sigma}_1=(\sigma_i\sigma_j\rho^{|i-j|}\mathbb{I}(|i-j|\leqslant d_0))_{p\times p}$, where $\rho\in(-1,1)$, $d_0$ is a predetermined positive integer, and $\mathbb{I}(\mathscr{A})$ is the indicator function for event $\mathscr{A}$. As shown in Chen et al. (2010), as long as all $\sigma_i$'s are bounded away from zero and infinity, all these three types of variance matrices satisfy $\textnormal{tr}({\Sigma}_1^4)=o\big(\textnormal{tr}^2({\Sigma}_1^2)\big)$. In addition, $p^{1/2}\gamma_{1p}^2=o\big(\textnormal{tr}^2({\Sigma}_1^2)\big)$ is also a commonly used assumption in the literature; see also \citet{wang2015high} and \citet{jiang2024nonparametric}.

We first establish some important results for the variance order of $\tilde{Q}_{nj}$, $j=1,2,3$, which hold regardless of the null hypothesis and facilitate us in deriving the leading term of $\textnormal{Var}(\hat{Q}_n)$. 

\begin{lemma}
\label{Lm3.1}
Suppose \autoref{A1} holds. Then, as $n,p\rightarrow\infty$,
\begin{align*}
\textnormal{Var}(\tilde{Q}_{n1})=&\frac{2}{n(n-1)}\textnormal{tr}(({W}{\Sigma}_1)^2)\textnormal{tr}(({W}{\Sigma}_2)^2)\\
&+O\Big(n^{-1}\textnormal{tr}\big(({W}{\Sigma}{W}{\Sigma}^\prime)^2+{W}{\Sigma}_1{W}{\Sigma}{W}{\Sigma}_2{W}{\Sigma}^\prime+{\Lambda}_3\circ {\Lambda}_3\big)\Big)\\
&+O\Big(n^{-2}\big[\textnormal{tr}^2({W}{\Sigma}{W}{\Sigma}^\prime)+\textnormal{tr}\big({\Lambda}_1^2\circ {\Lambda}_2^2+({\Lambda}_1\circ {\Lambda}_1)({\Lambda}_2\circ {\Lambda}_2)\big)\big]\Big),\\
\textnormal{Var}(\tilde{Q}_{n2})=&O\Big(n^{-2}\textnormal{tr}({W}{\Sigma}_1{W}{\Sigma}{W}{\Sigma}_2{W}{\Sigma}^\prime)+n^{-3}\big[\textnormal{tr}(({W}{\Sigma}_1)^2)\textnormal{tr}(({W}{\Sigma}_2)^2)+\textnormal{tr}({\Lambda}_1^2\circ {\Lambda}_2^2)\big]\Big),\\
\textnormal{Var}(\tilde{Q}_{n3})=&O\Big(n^{-4}\textnormal{tr}(({W}{\Sigma}_1)^2)\textnormal{tr}(({W}{\Sigma}_2)^2)\Big),
\end{align*}    
where ${\Lambda}_1={\Gamma}_1^\prime{W}{\Gamma}_1$, ${\Lambda}_2={\Gamma}_2^\prime{W}{\Gamma}_2$, $\Lambda_3={\Lambda}_1{\Lambda}_2$, and $\circ$ is the Hadamard Product such that for $k_1\times k_2$ matrices ${M}_1=(m_{1,ij})_{k_1\times k_2}$ and ${M}_2=(m_{2,ij})_{k_1\times k_2}$, ${M}_1\circ {M}_2=(m_{1,ij}m_{2,ij})_{k_1\times k_2}$.
\end{lemma}

\autoref{Lm3.1} is a very general result, as it does not rely on the null hypothesis ${\Sigma}={O}_{p}$ and the order condition given in \autoref{A2}. 
Note that under $\mathcal{H}_0$, ${\Sigma}={O}_{p}$ and ${\Lambda}_3={O}_{d}$. In addition, by Lemma 2(2) in the supplementary material, we have
\begin{align*}
\textnormal{tr}(({\Lambda}_1\circ {\Lambda}_1)({\Lambda}_2\circ {\Lambda}_2))\leqslant\textnormal{tr}({\Lambda}_1^2\circ {\Lambda}_2^2)\leqslant \sqrt{\textnormal{tr}({\Lambda}_1^4)\textnormal{tr}({\Lambda}_2^4)}=\sqrt{\textnormal{tr}(({W}{\Sigma}_1)^4)\textnormal{tr}(({W}{\Sigma}_2)^4)}.
\end{align*}
Thus, under $\mathcal{H}_0$ and \autoref{A2}, we get 
\begin{align}
\begin{split}
\textnormal{Var}(\tilde{Q}_{n1})=&\frac{2}{n(n-1)}\textnormal{tr}(({W}{\Sigma}_1)^2)\textnormal{tr}(({W}{\Sigma}_2)^2)+o\big(n^{-2}\textnormal{tr}(({W}{\Sigma}_1)^2)\textnormal{tr}(({W}{\Sigma}_2)^2)\big),\\
\textnormal{Var}(\tilde{Q}_{n2})=&o\big(n^{-2}\textnormal{tr}(({W}{\Sigma}_1)^2)\textnormal{tr}(({W}{\Sigma}_2)^2)\big),\\
\textnormal{Var}(\tilde{Q}_{n3})=&o\big(n^{-2}\textnormal{tr}(({W}{\Sigma}_1)^2)\textnormal{tr}(({W}{\Sigma}_2)^2)\big).
\end{split}
\label{variance_order}
\end{align}
Furthermore, by the Cauchy--Schwarz inequality, it holds that 
\begin{align*}
|\textnormal{Cov}(\tilde{Q}_{ni},\tilde{Q}_{nj})|\leqslant \sqrt{\textnormal{Var}(\tilde{Q}_{ni})\textnormal{Var}(\tilde{Q}_{nj})} \text{ for } i,j=1,2,3 \text{ and } i\neq j.
\end{align*}
Therefore, under $\mathcal{H}_0$, the variance of $\hat{Q}_n$ is given by
\begin{align*}
\textnormal{Var}(\hat{Q}_n)=\sigma_n^2+o(\sigma_n^2),\textnormal{ where } \sigma_n^2=\frac{2}{n(n-1)}\textnormal{tr}(({W}{\Sigma}_1)^2)\textnormal{tr}(({W}{\Sigma}_2)^2).
\end{align*}
Following similar results established in \citet{chen2010two}, Chen et al. (2010), \citet{li2012two}, \citet{jiang2023use}, and \citet{jiang2024nonparametric}, after being rescaled by $\sigma_n$, $\hat{Q}_n$ is asymptotically standard normal according to the Central Limit Theorem for martingale difference sequences \citep[p.58]{hall2014martingale}. We formally state this property in the following theorem.

\begin{theorem}
\label{Th3.1}
Suppose Assumptions \ref{A1}--\ref{A2} hold. Then, under $\mathcal{H}_0:{\Sigma}={O}_{p}$, as $n,p\rightarrow\infty$, 
\begin{align}
\frac{\hat{Q}_n}{\sigma_n}\stackrel{d}\longrightarrow\mathcal{N}(0,1).
\label{equation Th3.1}
\end{align}
\end{theorem}

Inspired by the asymptotic normality in (\ref{equation Th3.1}), we estimate $\sigma_n^2$ to derive a feasible test. For notational brevity, let $V_x=\textnormal{tr}(({W{\Sigma}}_1)^2)$ and $V_y=\textnormal{tr}(({W{\Sigma}}_2)^2)$. Then, estimating $\sigma_n^2$ amounts to estimating $V_x$ and $V_y$. Similar to the construction of $\hat{Q}_n$, note that $V_x$ can be decomposed as
\begin{align*}
V_x=&\textnormal{tr}({W}\mathbb{E}({X}_i{X}_i^\prime){W}\mathbb{E}({X}_j{X}_j^\prime))-2\textnormal{tr}({W}\mathbb{E}({X}_i{X}_i^\prime){W}{\mu}_1{\mu}_1^\prime)+\textnormal{tr}(({W}{\mu}_1{\mu}_1^\prime)^2)\\
=&\mathbb{E}\big[({X}_i^\prime{W}{X}_j)^2\big]-2\mathbb{E}({X}_i^\prime{W}{X}_j{X}_j^\prime{W}{X}_k)+\mathbb{E}({X}_i^\prime{W}{X}_j{X}_k^\prime{W}{X}_l),
\end{align*}
where ${X}_i$, ${X}_j$, ${X}_k$, ${X}_l$ are independent copies of ${X}$. Thus, an unbiased estimator for $V_x$ is given by $\hat{V}_x=\hat{V}_{x1}-2\hat{V}_{x2}+\hat{V}_{x3}$, where 
\begin{align*}
&\hat{V}_{x1}=\frac{1}{n(n-1)}\sum_{i\neq j}^n ({X}_i^\prime{W}{X}_j)^2,\ \ \hat{V}_{x2}=\frac{1}{n(n-1)(n-2)}\sum_{i\neq j\neq k}^n {X}_i^\prime{W}{X}_j{X}_j^\prime{W}{X}_k,\\
&\hat{V}_{x3}=\frac{1}{n(n-1)(n-2)(n-3)}\sum_{i\neq j\neq k\neq l}^n {X}_i^\prime{W}{X}_j{X}_k^\prime{W}{X}_l.
\end{align*}
Similarly, we estimate $V_y$ by $\hat{V}_y=\hat{V}_{y1}-2\hat{V}_{y2}+\hat{V}_{y3}$, where 
\begin{align*}
&\hat{V}_{y1}=\frac{1}{n(n-1)}\sum_{i\neq j}^n ({Y}_i^\prime{W}{Y}_j)^2,\ \ \hat{V}_{y2}=\frac{1}{n(n-1)(n-2)}\sum_{i\neq j\neq k}^n {Y}_i^\prime{W}{Y}_j{Y}_j^\prime{W}{Y}_k,\\
&\hat{V}_{y3}=\frac{1}{n(n-1)(n-2)(n-3)}\sum_{i\neq j\neq k\neq l}^n {Y}_i^\prime{W}{Y}_j{Y}_k^\prime{W}{Y}_l.
\end{align*}
After obtaining $\hat{V}_x$ and $\hat{V}_y$, the estimation of $\sigma_n^2$ is then constructed as $$\hat{\sigma}_n^2=\frac{2}{n(n-1)}\hat{V}_x\hat{V}_y.$$ 
The ratio consistency of $\hat{\sigma}_n^2$ is formally established in the next lemma.

\begin{lemma}
\label{Lm3.2}
Suppose Assumptions \ref{A1}--\ref{A2} hold. Then, as $n,p\rightarrow\infty$, 
\begin{align*}
\frac{\hat{\sigma}_n^2}{\sigma_n^2}\stackrel{p}\longrightarrow 1.
\end{align*}
\end{lemma}

\autoref{Lm3.2} shows that $\hat{\sigma}_n^2$ is a consistent estimator for $\sigma_n^2$ regardless of the null hypothesis $\mathcal{H}_0$. Then, with the above preparations in hand, a feasible test statistic is obtained by gathering \autoref{Th3.1}, \autoref{Lm3.2}, and Slutsky's theorem.

\begin{theorem}
\label{Th3.2}
Suppose Assumptions \ref{A1}--\ref{A2} hold. Then, under $\mathcal{H}_0:{\Sigma}={O}_{p}$, as $n,p\rightarrow\infty$, 
\begin{align*}
\hat{T}_n:=\frac{\hat{Q}_n}{\hat\sigma_n}\stackrel{d}\longrightarrow\mathcal{N}(0,1).
\end{align*}
\end{theorem}

\autoref{Th3.2} establishes that the standardized statistic $\hat{T}_n$ converges in distribution to the standard normal distribution under the null hypothesis of uncorrelatedness. It should be emphasized, however, that our proposed procedure is inherently one-sided, since $Q\geqslant 0$ and thus departures from the null hypothesis are reflected through large values of $\hat{T}_n$. Accordingly, for a given significance level $\alpha$, we reject $\mathcal{H}_0$ whenever $\hat{T}_n > z_{1-\alpha}$, where $z_{1-\alpha}$ denotes the $(1-\alpha)$-quantile of the standard normal distribution. Such a one-sided rejection rule is commonly adopted in the literature for test statistics constructed from degenerate $U$-statistics; see also \citet{delgado2006consistent} and \citet{chen2010two}.

\subsection{Asymptotic power}
\label{Asymptotic power}
In this section, we investigate the asymptotic power of our proposed test. Before presenting theoretical results, we first introduce the assumptions we need under alternative hypotheses.

\begin{assumption}
\label{A3}
As $n\rightarrow \infty$, $p=p(n)\rightarrow\infty$, $\textnormal{tr}(({W}{\Sigma}_i)^2)\rightarrow \infty$, $i=1,2$, and 
\begin{flalign}
&(1)\ \textnormal{tr}^2({W}{\Sigma}{W}{\Sigma}^\prime)=o\big(\textnormal{tr}(({W}{\Sigma}_1)^2)\textnormal{tr}(({W}{\Sigma}_2)^2)\big).\label{A3 1}&&\\
&(2)\ \textnormal{tr}({W}{\Sigma}_1{W}{\Sigma}{W}{\Sigma}_2{W}{\Sigma}^\prime)=o\big(n^{-1}\textnormal{tr}(({W}{\Sigma}_1)^2)\textnormal{tr}(({W}{\Sigma}_2)^2)\big).
\label{A3 2}
\end{flalign}

\end{assumption}

\autoref{A3} is trivially true under the null hypothesis and guarantees that the leading term of $\textnormal{Var}(\hat{Q}_n)$ would not be affected under $\mathcal{H}_1$. To see the intuition behind this, note that $\textnormal{tr}({\Lambda}_3^2)=\textnormal{tr}(({W}{\Sigma}{W}{\Sigma}^\prime)^2)$, $\textnormal{tr}({\Lambda}_3{\Lambda}_3^\prime)=\textnormal{tr}({W}{\Sigma}_1{W}{\Sigma}{W}{\Sigma}_2{W}{\Sigma}^\prime)$, and $\max\{\textnormal{tr}({\Lambda}_3^2),\textnormal{tr}({\Lambda}_3\circ{\Lambda}_3)\}\leqslant \textnormal{tr}({\Lambda}_3{\Lambda}_3^\prime)$
by Lemma 2(1) in the supplementary material. Therefore, to control the overall variance magnitude  under $\mathcal{H}_1$, it suffices to focus only on $\textnormal{tr}^2({W}{\Sigma}{W}{\Sigma}^\prime)$ and $\textnormal{tr}({W}{\Sigma}_1{W}{\Sigma}{W}{\Sigma}_2{W}{\Sigma}^\prime)$. Furthermore, by \autoref{A2} and (S1.5) given in the supplementary material, it holds that
\begin{align*}
&\textnormal{tr}^2({W}{\Sigma}{W}{\Sigma}^\prime)=\textnormal{tr}^2({\Lambda}_1{\Lambda}_2)\leqslant \textnormal{tr}({\Lambda}_1^2)\textnormal{tr}({\Lambda}_2^2)=\textnormal{tr}(({W}{\Sigma}_1)^2)\textnormal{tr}(({W}{\Sigma}_2)^2),\\
&\textnormal{tr}({W}{\Sigma}_1{W}{\Sigma}{W}{\Sigma}_2{W}{\Sigma}^\prime)=\textnormal{tr}({\Lambda}_1^2{\Lambda}_2^2)\leqslant \sqrt{\textnormal{tr}({\Lambda}_1^4)\textnormal{tr}({\Lambda}_2^4)}=o\Big(\textnormal{tr}(({W}{\Sigma}_1)^2)\textnormal{tr}(({W}{\Sigma}_2)^2)\Big).
\end{align*}
Thus, the order assumptions (\ref{A3 1})--(\ref{A3 2}) are reasonable in the high-dimensional settings. 
Define the quantity
\begin{align*}
G_n({W};{\Gamma}_1,{\Gamma}_2):= \frac{Q}{\sigma_n} = \sqrt{\frac{n(n-1)}{2}\cdot\frac{\textnormal{tr}^2({W}{\Sigma}{W}{\Sigma}')}{\textnormal{tr}(({W}{\Sigma}_1)^2)\textnormal{tr}(({W}{\Sigma}_2)^2)}}\geqslant 0.
\end{align*}
The following theorem establishes the asymptotic normality after properly centralizing $\hat{Q}_n$ and provides a general framework to analyze the asymptotic power of our test.

\begin{theorem}
\label{Th3.3}
Suppose Assumptions \ref{A1}--\ref{A3} hold. Then, under $\mathcal{H}_1:{\Sigma}\neq {O}_{p}$, as $n,p\rightarrow \infty$,
\begin{align}
\frac{\hat{Q}_n-Q}{{\sigma}_n}\stackrel{d}\longrightarrow \mathcal{N}(0,1).
\label{euation Th3.3}
\end{align}
Furthermore,
\begin{align}
\lim_{n,p\rightarrow \infty} \Big[\mathbb{P}\left\{\hat{T}_n>z_{1-\alpha}\right\} - \Phi\big( G_n({W};{\Gamma}_1,{\Gamma}_2) -z_{1-\alpha}\big)\Big]=0.
\label{power function}
\end{align}
\end{theorem}

(\ref{power function}) characterizes the asymptotic power of our proposed methods. The test statistic $\hat{T}_n$ possesses non-trivial power against the alternative hypothesis $\mathcal{H}_1$ provided that the limit of $G_n({W};{\Gamma}_1,{\Gamma}_2)$ exists and satisfies $\lim_{n,p\rightarrow\infty} G_n({W};{\Gamma}_1,{\Gamma}_2) > 0$. In addition, if $\lim_{n,p\rightarrow\infty} G_n({W};{\Gamma}_1,{\Gamma}_2) = \infty$, the proposed test will be asymptotically consistent. Furthermore, \autoref{Th3.3} provides a useful framework for analyzing the power of the proposed test under local alternatives based on the quantity
\begin{align*}
D({W};{\Gamma}_1,{\Gamma}_2)
=
\frac{
\textnormal{tr}({W\Sigma W\Sigma'})
}{
\sqrt{
\textnormal{tr}\bigl(({W}{\Sigma}_1)^2\bigr)
\textnormal{tr}\bigl(({W}{\Sigma}_2)^2\bigr)
}
}.
\end{align*}
To better interpret this, note that condition (\ref{A3 1}) requires $D({W};{\Gamma}_1,{\Gamma}_2)\rightarrow 0$. On the other hand, to guarantee nontrivial asymptotic power of $\hat{T}_n$, the condition $\lim_{n,p\rightarrow\infty} G_n({W};{\Gamma}_1,{\Gamma}_2)>0$
requires at least $nD({W};{\Gamma}_1,{\Gamma}_2)\rightarrow c>0$. Therefore, in the high-dimensional regime, the critical order of detectable local alternatives is given by $D({W};{\Gamma}_1,{\Gamma}_2)=c/n$
for some positive constant $c$, implying that the proposed test for high-dimensional uncorrelatedness is capable of detecting local alternatives converging to the null at a rate faster than the classical parametric rate $n^{-1/2}$ in fixed-dimensional settings; see also \citet{chen2010two} and \citet{jiang2024nonparametric} for a similar analysis.

Finally, we end this section by highlighting the role of the weighting matrix ${W}$. According to \autoref{Th3.3}, the asymptotic power of the proposed test is determined by the quantity $G_n({W};{\Gamma}_1,{\Gamma}_2)$, which naturally motivates a comparison between two weighting matrices ${W}_1$ and ${W}_2$ through the asymptotic relative efficiency (ARE), defined as
\begin{align*}
\textnormal{ARE}({W}_1,{W}_2)
:=
\frac{
G_n^2({W}_1;{\Gamma}_1,{\Gamma}_2)
}{
G_n^2({W}_2;{\Gamma}_1,{\Gamma}_2)
}.
\end{align*}
In particular, we are interested in whether a properly chosen weighting matrix ${W}$ can achieve higher asymptotic power than the identity matrix ${I}_p$, that is, whether $\lim_{n,p\to\infty}\textnormal{ARE}({W},{I}_p)>1$. Let $1_p$ denote the $p$-dimensional vector with all elements equal to 1. The following theorem shows that, under some mild conditions, a correctly specified weighting matrix ${W}$ with ${W}\neq {I}_p$ always yields power no worse than that obtained using ${I}_p$.
\begin{theorem}
\label{Th3.4}
Suppose Assumptions \ref{A1}--\ref{A3} hold with $a_1=a_2=\ldots=a_p=a\neq0$, $b_1=b_2=\ldots=b_p=b>0$, and $r=a/b$. Under the alternative hypothesis $\mathcal{H}_1:{\Sigma}\neq {O}_{p}$, as $n,p\rightarrow\infty$:

(1) If $pr^2\rightarrow 0$, then 
\begin{align*}
\lim_{n,p\rightarrow\infty} \textnormal{ARE}( W,{I}_p)=1.
\end{align*}

(2) If $p^{-1}\max\{r^2,pr^4\}(1_p'\Sigma1_p)^2/\textnormal{tr}({\Sigma\Sigma'})\rightarrow c\in(0,\infty]$ and for $i=1,2$, $\max\{r^2,pr^4\}1_p'{\Sigma}_i^21_p/\textnormal{tr}({\Sigma}_i^2)\rightarrow 0$, then
\begin{align*}
\lim_{n,p\rightarrow\infty} \textnormal{ARE}( W,{I}_p)>1.
\end{align*}
\end{theorem}

\autoref{Th3.4} provides important practical guidance for the construction of the weighting matrix ${W}$. In particular, \autoref{Th3.4}(1) indicates that ${W}$ should at least satisfy $pr^2\to c_1\in(0,\infty]$ to achieve asymptotic power improvement. Moreover, the conditions imposed in \autoref{Th3.4}(2) are generally mild and can be satisfied in a wide range of settings. For illustration, let $0<\gamma_{i1}\leqslant \gamma_{i2}\leqslant \cdots \leqslant \gamma_{ip}$ denote the eigenvalues of ${\Sigma}_i$ for $i=1,2$. By straightforward calculations, we have $1_p'{\Sigma}_i^21_p/\textnormal{tr}({\Sigma}_i^2)
\leqslant
\gamma_{ip}^2/\gamma_{i1}^2.$
Hence, if the condition number $\gamma_{ip}/\gamma_{i1}$ converges to a positive constant, together with $r\to0$ and $pr^4\to0$, it follows immediately that $\max\{r^2,pr^4\}1_p'{\Sigma}_i^21_p/\textnormal{tr}({\Sigma}_i^2)\rightarrow 0$. Furthermore, note that $p^{-1}\max\{r^2,pr^4\}(1_p'\Sigma1_p)^2/\textnormal{tr}({\Sigma\Sigma'})\rightarrow c\in(0,\infty]$
is equivalent to
\begin{align}
\max\{pr^2,p^2r^4\}
\frac{
\textnormal{tr}^2({\Sigma}1_p1_p')
}{
\textnormal{tr}(1_p1_p'1_p1_p')
\textnormal{tr}({\Sigma}{\Sigma}')
}
\to c\in(0,\infty].
\label{simplified convergence}
\end{align}
Condition \eqref{simplified convergence} essentially requires ${\Sigma}$ and $1_p1_p'$ to be ``highly correlated'' in the $\mathscr{L}_2$-norm sense. In particular, if ${\Sigma}=c_2 1_p1_p'$ for some arbitrarily small constant $c_2>0$, then \eqref{simplified convergence} holds automatically whenever $pr^2\to c_1\in(0,\infty]$. 
That is, when the dependence is weak ($c_2>0$ can be arbitrarily small and can even converge to 0) but pervasive ($\Sigma$ contains no zero entries), employing $W$ is expected to be more powerful if $pr^2\rightarrow c_1\in(0,\infty]$. The importance of the weighting matrix $W$ is further illustrated in the next section.

\section{Simulations}
\label{Simulations}

In this section, we conduct Monte Carlo simulations to investigate the finite-sample performance of the proposed test. For the weighting matrix $W$, we set
$a_1 = a_2 = \cdots = a_p = a$ and $b_1 = b_2 = \cdots = b_p = b > 0$. 
Under this specification, the weighting matrix can be written as $W = b^2 I_p + a^2 1_p 1_p'$, where $1_p$ denotes a $p$-dimensional vector of ones. We consider the following four configurations for $W$:
\begin{align*}
&W_1:  a = 0.1 p^{-1/2}, \; b = 1; \,\,\,
W_2: a = 0.1 p^{-1}, \; b = 1;\\
&W_3:  a = 0.1 p^{-2}, \; b = 1; \,\,\,
W_4: a = 0, \; b = 1.
\end{align*}
Let $\hat{T}_n(W_i)$ denote the test statistic constructed using $W_i$, $i = 1, \ldots, 4$. 
According to \autoref{Th3.4}(1), the asymptotic power of $\hat{T}_n(W_2)$ and $\hat{T}_n(W_3)$ is expected to be similar to that of $\hat{T}_n(W_4)$ since $pr^2\rightarrow 0$, where $r=a/b$. In contrast, $\hat{T}_n(W_1)$ is designed to better capture ``weak but pervasive'' alternatives and therefore ought to deliver higher empirical power in such settings. Following \autoref{A1}(1), we set $d = 2p$ to ensure the existence of $X$ and $Y$. The data are generated via $X_i = \mu_1 + \Gamma_1 Z_i$ and $Y_i = \mu_2 + \Gamma_2 Z_i$, $i = 1,2,\ldots,n$, where $\{Z_i\}_{i=1}^n$ are independent $2p$-dimensional random vectors with i.i.d. components $Z_{ij}$. We consider the following four data-generating processes (DGPs) for each $Z_{ij}$:
\begin{align*}
&\text{DGP1: } \mathcal{N}(0,1);\,\,\,\text{DGP2: } \sqrt{\frac{3}{5}}t_5;\,\,\,\text{DGP3: } \frac{1}{\sqrt{8}}(\chi^2_4 - 4);\,\,\,\text{DGP4: } \Gamma(4,0.5) - 2,
\end{align*}
where $t_v$ is the $t$-distribution with $v$ degrees of freedom, $\chi_s^2$ is the $\chi^2$ distribution with $s$ degrees of freedom, and $\Gamma(u,v)$ is the gamma distribution with a shape parameter $u$ and a rate parameter $v$. All DGPs are properly standardized such that $\mathbb{E}(Z_i) = 0_{2p}$ and $\text{Var}(Z_i) = I_{2p}$. Moreover, note that the first two DGPs are symmetric around 0, while the latter two are skewed to the right. Therefore, our DGPs include both symmetric and asymmetric distributions. The sample size is set to $n=50$, and the empirical rejection rate is computed from 1000 replications. The significance level is fixed at $\alpha=0.05$, and the corresponding one-sided critical value is $z_{1-\alpha}=1.6449$. The dimension is chosen as $p\in\{100,200,500,1000\}$ to reflect the high-dimensional regime.

To make a comparison with existing tests, we consider comparing the numerical performance of our proposed tests with (1) the fixed-$n$ distance correlation $t$-test \citep[SR]{szekely2013distance}, (2) the distance-based and RKHS-based tests of \citet[ZZYS]{zhu2020distance}, (3) the diverging-$n$ distance correlation test \citep[GFLS]{gao2021asymptotic}, and (4) the rank-based max-sum tests proposed by \citet[WLF]{wang2026testing}. For notational convenience, we use $\text{ZZYS}_{\mathrm{d}}$ and $\text{ZZYS}_{\mathrm{h}}$ to denote the ZZYS tests based on distance covariance and Hilbert--Schmidt covariance, respectively. For the WLF procedure, we employ $\text{WLF}_{\mathrm{sum}}$ and $\text{WLF}_{\mathrm{max}}$ to denote the corresponding sum-type and max-type test statistics. Under the null hypothesis, the critical values of these competing tests are readily available, without the need for bootstrap methods. In the following, we first consider the empirical size in \autoref{Simulation 1}, and then study the empirical power in \autoref{Simulation 2}. 

\subsection{Simulation 1}
\label{Simulation 1}
We consider the following two scenarios:
\begin{align*}
&\text{S1.1: } \mu_1= 0_p,\mu_2 = 1_p;\,\,\Gamma_1 = (I_p,O_p),\Gamma_2 = (O_p,(e^{-|i-j|})_{p\times p});\\
&\text{S1.2: } \mu_1= 1_p,\mu_2 = 0.5\times 1_p;\,\,\Gamma_1 = ((e^{-|i-j|})_{p\times p},O_p),\Gamma_2 = (O_p,2\times (e^{-|i-j|})_{p\times p}).
\end{align*}
In both scenarios, we have $\Sigma=\Gamma_1\Gamma_2'=O_p$, implying that $X$ and $Y$ are uncorrelated. Moreover, both $\Sigma_1=\Gamma_1\Gamma_1'$ and $\Sigma_2=\Gamma_2\Gamma_2'$ are positive definite, ensuring the validity of the simulation setups. The numerical results under scenarios S1.1--S1.2 are summarized in \autoref{S1.1--S1.2}.

Overall, the proposed tests exhibit satisfactory size control under both symmetric and asymmetric distributions, indicating that their finite-sample performance is relatively stable across different distributional settings. In addition, under the null hypothesis of uncorrelatedness, the choice of the weighting matrix $W$ has little impact on the empirical size. Compared with existing tests, the statistics $\hat{T}_n(W_i)$ maintain the empirical size close to the nominal level and perform comparably well. In sharp contrast, $\mathrm{WLF}_{\max}$ tends to be conservative, with empirical size substantially below the nominal level in both scenarios. 
Taken together, these findings suggest that our proposed methods are particularly suitable for high-dimensional settings with relatively small sample sizes.

\begin{table}[H]
\vspace{1em}
\renewcommand\arraystretch{0.85}
\centering
\caption{Empirical rejection rates under scenarios S1.1--S1.2.} 
\resizebox{\linewidth}{!}{
\begin{tabular}{ccccccccccccc}
\hline
Scenario               & DGP                   & $p$  & $\hat{T}_n(W_1)$ & $\hat{T}_n(W_2)$ & $\hat{T}_n(W_3)$ & $\hat{T}_n(W_4)$ & $\text{SR}$ & $\text{ZZYS}_{\text{d}}$ & $\text{ZZYS}_{\text{h}}$ & $\text{GFLS}$ & $\text{WLF}_{\text{sum}}$ & $\text{WLF}_{\text{max}}$ \\ \hline
\multirow[t]{16}{*}{S1.1} & \multirow[t]{4}{*}{DGP1} & 100  & 0.0610           & 0.0590           & 0.0590           & 0.0590           & 0.0620      & 0.0540                   & 0.0520                   & 0.0650        & 0.0610                    & 0.0040                    \\
                       &                       & 200  & 0.0500           & 0.0500           & 0.0500           & 0.0500           & 0.0430      & 0.0540                   & 0.0460                   & 0.0470        & 0.0410                    & 0.0020                    \\
                       &                       & 500  & 0.0500           & 0.0540           & 0.0540           & 0.0540           & 0.0450      & 0.0430                   & 0.0470                   & 0.0500        & 0.0510                    & 0.0010                    \\
                       &                       & 1000 & 0.0490           & 0.0530           & 0.0530           & 0.0530           & 0.0490      & 0.0410                   & 0.0360                   & 0.0560        & 0.0520                    & 0.0000                    \\ \cline{2-13} 
                       & \multirow[t]{4}{*}{DGP2} & 100  & 0.0600           & 0.0570           & 0.0570           & 0.0570           & 0.0520      & 0.0480                   & 0.0440                   & 0.0600        & 0.0550                    & 0.0100                    \\
                       &                       & 200  & 0.0600           & 0.0530           & 0.0530           & 0.0530           & 0.0570      & 0.0580                   & 0.0540                   & 0.0590        & 0.0590                    & 0.0030                    \\
                       &                       & 500  & 0.0660           & 0.0660           & 0.0660           & 0.0660           & 0.0630      & 0.0460                   & 0.0400                   & 0.0640        & 0.0540                    & 0.0000                    \\
                       &                       & 1000 & 0.0530           & 0.0500           & 0.0500           & 0.0500           & 0.0450      & 0.0620                   & 0.0610                   & 0.0490        & 0.0530                    & 0.0000                    \\ \cline{2-13} 
                       & \multirow[t]{4}{*}{DGP3} & 100  & 0.0550           & 0.0520           & 0.0520           & 0.0520           & 0.0510      & 0.0550                   & 0.0510                   & 0.0560        & 0.0490                    & 0.0080                    \\
                       &                       & 200  & 0.0550           & 0.0590           & 0.0600           & 0.0600           & 0.0590      & 0.0550                   & 0.0420                   & 0.0620        & 0.0540                    & 0.0000                    \\
                       &                       & 500  & 0.0530           & 0.0360           & 0.0360           & 0.0360           & 0.0360      & 0.0380                   & 0.0410                   & 0.0380        & 0.0390                    & 0.0010                    \\
                       &                       & 1000 & 0.0530           & 0.0550           & 0.0550           & 0.0550           & 0.0490      & 0.0560                   & 0.0540                   & 0.0560        & 0.0560                    & 0.0000                    \\ \cline{2-13} 
                       & \multirow[t]{4}{*}{DGP4} & 100  & 0.0540           & 0.0500           & 0.0500           & 0.0500           & 0.0510      & 0.0470                   & 0.0460                   & 0.0550        & 0.0440                    & 0.0050                    \\
                       &                       & 200  & 0.0520           & 0.0500           & 0.0500           & 0.0500           & 0.0460      & 0.0530                   & 0.0480                   & 0.0520        & 0.0480                    & 0.0040                    \\
                       &                       & 500  & 0.0690           & 0.0530           & 0.0530           & 0.0530           & 0.0520      & 0.0560                   & 0.0490                   & 0.0540        & 0.0510                    & 0.0010                    \\
                       &                       & 1000 & 0.0580           & 0.0580           & 0.0580           & 0.0580           & 0.0570      & 0.0500                   & 0.0540                   & 0.0600        & 0.0560                    & 0.0000                    \\ \hline
\multirow[t]{16}{*}{S1.2} & \multirow[t]{4}{*}{DGP1} & 100  & 0.0660           & 0.0650           & 0.0640           & 0.0640           & 0.0630      & 0.0470                   & 0.0460                   & 0.0670        & 0.0540                    & 0.0020                    \\
                       &                       & 200  & 0.0520           & 0.0510           & 0.0510           & 0.0510           & 0.0450      & 0.0410                   & 0.0460                   & 0.0500        & 0.0400                    & 0.0010                    \\
                       &                       & 500  & 0.0520           & 0.0540           & 0.0540           & 0.0540           & 0.0430      & 0.0510                   & 0.0550                   & 0.0490        & 0.0460                    & 0.0020                    \\
                       &                       & 1000 & 0.0460           & 0.0500           & 0.0500           & 0.0500           & 0.0510      & 0.0430                   & 0.0480                   & 0.0530        & 0.0460                    & 0.0010                    \\ \cline{2-13} 
                       & \multirow[t]{4}{*}{DGP2} & 100  & 0.0500           & 0.0510           & 0.0510           & 0.0510           & 0.0490      & 0.0490                   & 0.0620                   & 0.0510        & 0.0470                    & 0.0030                    \\
                       &                       & 200  & 0.0410           & 0.0430           & 0.0430           & 0.0430           & 0.0420      & 0.0530                   & 0.0450                   & 0.0470        & 0.0460                    & 0.0020                    \\
                       &                       & 500  & 0.0620           & 0.0550           & 0.0550           & 0.0550           & 0.0530      & 0.0550                   & 0.0440                   & 0.0560        & 0.0570                    & 0.0000                    \\
                       &                       & 1000 & 0.0650           & 0.0670           & 0.0670           & 0.0670           & 0.0640      & 0.0600                   & 0.0650                   & 0.0700        & 0.0600                    & 0.0000                    \\ \cline{2-13} 
                       & \multirow[t]{4}{*}{DGP3} & 100  & 0.0540           & 0.0550           & 0.0550           & 0.0550           & 0.0550      & 0.0510                   & 0.0500                   & 0.0590        & 0.0450                    & 0.0070                    \\
                       &                       & 200  & 0.0510           & 0.0450           & 0.0440           & 0.0440           & 0.0380      & 0.0450                   & 0.0420                   & 0.0430        & 0.0390                    & 0.0010                    \\
                       &                       & 500  & 0.0580           & 0.0540           & 0.0540           & 0.0540           & 0.0520      & 0.0530                   & 0.0440                   & 0.0550        & 0.0570                    & 0.0000                    \\
                       &                       & 1000 & 0.0480           & 0.0560           & 0.0560           & 0.0560           & 0.0530      & 0.0490                   & 0.0560                   & 0.0550        & 0.0490                    & 0.0000                    \\ \cline{2-13} 
                       & \multirow[t]{4}{*}{DGP4} & 100  & 0.0480           & 0.0530           & 0.0530           & 0.0530           & 0.0510      & 0.0430                   & 0.0440                   & 0.0520        & 0.0520                    & 0.0010                    \\
                       &                       & 200  & 0.0640           & 0.0680           & 0.0670           & 0.0670           & 0.0640      & 0.0570                   & 0.0390                   & 0.0680        & 0.0520                    & 0.0010                    \\
                       &                       & 500  & 0.0550           & 0.0530           & 0.0530           & 0.0530           & 0.0490      & 0.0490                   & 0.0520                   & 0.0530        & 0.0540                    & 0.0010                    \\
                       &                       & 1000 & 0.0590           & 0.0620           & 0.0620           & 0.0620           & 0.0570      & 0.0610                   & 0.0550                   & 0.0600        & 0.0600                    & 0.0000                    \\ \hline
\end{tabular}}
\label{S1.1--S1.2}
\end{table}

\subsection{Simulation 2}
\label{Simulation 2}

Next, we consider the empirical power of our methods. For the rest of this section, we fix $\mu_1 = 1_p$ and $\mu_2 = 0.5\times 1_p$. Consider the following four scenarios:
\begin{align*}
&\text{S2.1: } \Gamma_1 = (I_p,O_p),\Gamma_2 = (M_1,I_p);\\
&\text{S2.2: } \Gamma_1 = (I_p,0.5M_2),\Gamma_2 = (0.5M_2,I_p);\\
&\text{S2.3: } \Gamma_1 = (I_p,O_p),\Gamma_2 = (M_3,I_p - M_3);\\
&\text{S2.4: } \Gamma_1 = (I_p,O_p),\Gamma_2 = (M_4,I_p),
\end{align*}
where $M_1 = (p^{-3/5}\mathbb{I}(|i-j|\leqslant \sqrt{p}))_{p\times p}$, $M_2 = (e^{-|i-j|}\mathbb{I}(i\leqslant 20)\mathbb{I}(j\leqslant 20))_{p\times p}$, $M_3 = p^{-1}1_p1_p'$, and
\begin{align*}
M_4 = \left(\begin{matrix}
p^{-3/4}1_{p/2}1_{p/2}' & O_{p/2}\\
O_{p/2} & p^{-3/4}1_{p/2}1_{p/2}'
\end{matrix}\right).
\end{align*}
Specifically, in S2.1, the covariance matrix $\Sigma$ has a banded structure, with the band width increasing at the rate of $\sqrt{p}$. In S2.2, dependence is present only in the first 20 components. S2.3 corresponds to a ``weak but pervasive'' alternative, and S2.4 represents a ``block-pervasive'' structure. The simulation results are summarized in Tables \ref{S2.1--S2.2}--\ref{S2.3--S2.4}.

We first analyze scenario S2.1. Comparing the empirical power of $\hat{T}_n$ under different weighting matrices, we observe that $\hat{T}_n(W_1)$ consistently achieves the highest power. In contrast, the empirical power of $\hat{T}_n(W_i)$, $i=2,3,4$, is substantially lower and decreases rapidly as the dimension $p$ increases. Meanwhile, the performance of $\hat{T}_n(W_1)$ remains relatively stable across different dimensions, highlighting the importance of the condition $pr^2 \to c>0$. Furthermore, the empirical powers of $\hat{T}_n(W_2)$ and $\hat{T}_n(W_3)$ are almost identical to that of $\hat{T}_n(W_4)$, which validates the conclusion of \autoref{Th3.4}(1): once $pr^2 \to 0$, modifying the weighting matrix $W$ no longer lead to any asymptotic power improvement. Compared with existing tests, our proposed test statistic $\hat{T}_n(W_1)$ still delivers the highest power, while SR, $\text{ZZYS}_{\text{d}}$, $\text{ZZYS}_{\text{h}}$, GFLS, and $\text{WLF}_{\textnormal{sum}}$ suffer from severe power loss as $p$ increases. Furthermore, the empirical power of $\text{WLF}_{\textnormal{max}}$ is even below the nominal size, which suggests that a larger sample size may be necessary for $\text{WLF}_{\textnormal{max}}$ to achieve nontrivial power under such simulation settings.

Next, we turn to scenario S2.2. In this setting, the proposed tests exhibit performance comparable to SR and GFLS, while clearly outperforming ZZYS and $\text{WLF}_{\mathrm{sum}}$. The only competitor that dominates our methods is $\text{WLF}_{\mathrm{max}}$, whose empirical power remains close to 1 across different values of $p$. This phenomenon can be explained by the design of $\text{WLF}_{\mathrm{max}}$. As the dimension $p$ increases, the overall signal strength dilutes, leading to a loss of power for most tests. In contrast, $\text{WLF}_{\mathrm{max}}$ focuses on the strongest signal. Since the covariance structure among the first 20 components of $X$ and $Y$ is fixed, the largest signal remains unchanged. Consequently, the power of $\text{WLF}_{\mathrm{max}}$ is unaffected by increasing dimensionality and stays close to 1. Nevertheless, when considering both size control and power performance across all scenarios, our proposed procedures remain a more reliable and competitive alternative to $\text{WLF}_{\mathrm{max}}$.

Finally, we consider scenarios S2.3--S2.4. Under the ``weak but pervasive'' and ``block-pervasive'' alternatives, the advantage of $\hat{T}_n(W_1)$ becomes more evident. In particular, $\hat{T}_n(W_1)$ achieves nontrivial empirical power across all DGPs, whereas all other tests exhibit substantially lower power, with empirical power in S2.3 remaining close to the nominal size. Moreover, as the dimension $p$ increases, the empirical power of $\hat{T}_n(W_1)$ improves markedly, while other tests suffer from a rapid decline in power. These findings suggest that, when applying our proposed methods, the condition $pr^2 \to 0$ should be satisfied to ensure power gains across a broad class of alternatives. Overall, the above simulation results highlight the robustness and effectiveness of our proposed methods across a wide range of scenarios, indicating that the proposed test $\hat{T}_n$ is a strong competitor to existing methods for testing the uncorrelatedness of a pair of high-dimensional random vectors.

\begin{table}[H]
\vspace{1em}
\renewcommand\arraystretch{0.85}
\centering
\caption{Empirical rejection rates under scenarios S2.1--S2.2.} 
\resizebox{\linewidth}{!}{
\begin{tabular}{ccccccccccccc}
\hline
Scenario               & DGP                   & $p$  & $\hat{T}_n(W_1)$ & $\hat{T}_n(W_2)$ & $\hat{T}_n(W_3)$ & $\hat{T}_n(W_4)$ & $\text{SR}$ & $\text{ZZYS}_{\text{d}}$ & $\text{ZZYS}_{\text{h}}$ & $\text{GFLS}$ & $\text{WLF}_{\text{sum}}$ & $\text{WLF}_{\text{max}}$ \\ \hline
\multirow[t]{16}{*}{S2.1} & \multirow[t]{4}{*}{DGP1} & 100  & 0.9390           & 0.7590           & 0.7560           & 0.7560           & 0.7340      & 0.5900                   & 0.3120                   & 0.7430        & 0.6660                    & 0.0070                    \\
                       &                       & 200  & 0.8770           & 0.4860           & 0.4860           & 0.4860           & 0.4680      & 0.3440                   & 0.1890                   & 0.4830        & 0.3970                    & 0.0000                    \\
                       &                       & 500  & 0.8200           & 0.2180           & 0.2180           & 0.2180           & 0.2070      & 0.1730                   & 0.1030                   & 0.2200        & 0.1920                    & 0.0010                    \\
                       &                       & 1000 & 0.8610           & 0.1280           & 0.1280           & 0.1280           & 0.1220      & 0.1060                   & 0.0810                   & 0.1310        & 0.1000                    & 0.0010                    \\ \cline{2-13} 
                       & \multirow[t]{4}{*}{DGP2} & 100  & 0.9490           & 0.7700           & 0.7690           & 0.7690           & 0.7470      & 0.6780                   & 0.2870                   & 0.7630        & 0.7420                    & 0.0080                    \\
                       &                       & 200  & 0.8780           & 0.4780           & 0.4750           & 0.4750           & 0.4650      & 0.4180                   & 0.1830                   & 0.4750        & 0.4410                    & 0.0030                    \\
                       &                       & 500  & 0.8320           & 0.2210           & 0.2210           & 0.2210           & 0.2170      & 0.1780                   & 0.0870                   & 0.2220        & 0.1810                    & 0.0010                    \\
                       &                       & 1000 & 0.8550           & 0.1420           & 0.1420           & 0.1420           & 0.1320      & 0.1070                   & 0.0720                   & 0.1440        & 0.1360                    & 0.0000                    \\ \cline{2-13} 
                       & \multirow[t]{4}{*}{DGP3} & 100  & 0.9610           & 0.8160           & 0.8160           & 0.8160           & 0.8330      & 0.7730                   & 0.4750                   & 0.8420        & 0.8070                    & 0.0080                    \\
                       &                       & 200  & 0.8610           & 0.4760           & 0.4750           & 0.4750           & 0.4690      & 0.4370                   & 0.2670                   & 0.4880        & 0.4800                    & 0.0010                    \\
                       &                       & 500  & 0.8220           & 0.2340           & 0.2340           & 0.2340           & 0.2330      & 0.2310                   & 0.1500                   & 0.2450        & 0.2170                    & 0.0010                    \\
                       &                       & 1000 & 0.8570           & 0.1530           & 0.1530           & 0.1530           & 0.1430      & 0.1440                   & 0.1060                   & 0.1500        & 0.1490                    & 0.0010                    \\ \cline{2-13} 
                       & \multirow[t]{4}{*}{DGP4} & 100  & 0.9360           & 0.7740           & 0.7700           & 0.7700           & 0.7680      & 0.6490                   & 0.3830                   & 0.7760        & 0.7010                    & 0.0040                    \\
                       &                       & 200  & 0.8460           & 0.4680           & 0.4680           & 0.4680           & 0.4670      & 0.3830                   & 0.2120                   & 0.4770        & 0.4180                    & 0.0020                    \\
                       &                       & 500  & 0.8280           & 0.2170           & 0.2170           & 0.2170           & 0.2150      & 0.1850                   & 0.1140                   & 0.2210        & 0.1920                    & 0.0000                    \\
                       &                       & 1000 & 0.8860           & 0.1250           & 0.1250           & 0.1250           & 0.1210      & 0.1260                   & 0.1050                   & 0.1250        & 0.1300                    & 0.0000                    \\ \hline
\multirow[t]{16}{*}{S2.2} & \multirow[t]{4}{*}{DGP1} & 100  & 1.0000           & 1.0000           & 1.0000           & 1.0000           & 1.0000      & 1.0000                   & 0.9440                   & 1.0000        & 0.9980                    & 1.0000                    \\
                       &                       & 200  & 0.9860           & 0.9850           & 0.9850           & 0.9850           & 0.9830      & 0.8740                   & 0.5240                   & 0.9830        & 0.7600                    & 1.0000                    \\
                       &                       & 500  & 0.5200           & 0.5230           & 0.5230           & 0.5230           & 0.5130      & 0.2810                   & 0.1450                   & 0.5240        & 0.2110                    & 1.0000                    \\
                       &                       & 1000 & 0.2290           & 0.2350           & 0.2350           & 0.2350           & 0.2230      & 0.1340                   & 0.0890                   & 0.2380        & 0.1290                    & 0.9940                    \\ \cline{2-13} 
                       & \multirow[t]{4}{*}{DGP2} & 100  & 1.0000           & 1.0000           & 1.0000           & 1.0000           & 1.0000      & 1.0000                   & 0.9450                   & 1.0000        & 0.9990                    & 1.0000                    \\
                       &                       & 200  & 0.9860           & 0.9800           & 0.9800           & 0.9800           & 0.9810      & 0.9240                   & 0.5550                   & 0.9830        & 0.7660                    & 1.0000                    \\
                       &                       & 500  & 0.5600           & 0.5550           & 0.5550           & 0.5550           & 0.5470      & 0.3510                   & 0.1580                   & 0.5620        & 0.2530                    & 1.0000                    \\
                       &                       & 1000 & 0.2290           & 0.2420           & 0.2410           & 0.2410           & 0.2320      & 0.1420                   & 0.0990                   & 0.2410        & 0.1120                    & 1.0000                    \\ \cline{2-13} 
                       & \multirow[t]{4}{*}{DGP3} & 100  & 1.0000           & 1.0000           & 1.0000           & 1.0000           & 1.0000      & 1.0000                   & 0.9580                   & 1.0000        & 1.0000                    & 1.0000                    \\
                       &                       & 200  & 0.9820           & 0.9820           & 0.9820           & 0.9820           & 0.9810      & 0.9030                   & 0.5090                   & 0.9810        & 0.7880                    & 1.0000                    \\
                       &                       & 500  & 0.5250           & 0.5340           & 0.5340           & 0.5340           & 0.5160      & 0.3300                   & 0.1570                   & 0.5260        & 0.2620                    & 1.0000                    \\
                       &                       & 1000 & 0.2360           & 0.2530           & 0.2530           & 0.2530           & 0.2340      & 0.1540                   & 0.1050                   & 0.2440        & 0.1350                    & 1.0000                    \\ \cline{2-13} 
                       & \multirow[t]{4}{*}{DGP4} & 100  & 1.0000           & 1.0000           & 1.0000           & 1.0000           & 1.0000      & 1.0000                   & 0.9430                   & 1.0000        & 0.9990                    & 1.0000                    \\
                       &                       & 200  & 0.9880           & 0.9850           & 0.9850           & 0.9850           & 0.9830      & 0.8830                   & 0.5140                   & 0.9840        & 0.7990                    & 1.0000                    \\
                       &                       & 500  & 0.5250           & 0.5340           & 0.5340           & 0.5340           & 0.5230      & 0.3060                   & 0.1700                   & 0.5310        & 0.2500                    & 1.0000                    \\
                       &                       & 1000 & 0.2200           & 0.2280           & 0.2280           & 0.2280           & 0.2150      & 0.1560                   & 0.1170                   & 0.2290        & 0.1220                    & 1.0000                    \\ \hline
\end{tabular}}
\label{S2.1--S2.2}
\end{table}

\begin{table}[H]
\vspace{1em}
\renewcommand\arraystretch{0.85}
\centering
\caption{Empirical rejection rates under scenarios S2.3--S2.4.} 
\resizebox{\linewidth}{!}{
\begin{tabular}{ccccccccccccc}
\hline
Scenario               & DGP                   & $p$  & $\hat{T}_n(W_1)$ & $\hat{T}_n(W_2)$ & $\hat{T}_n(W_3)$ & $\hat{T}_n(W_4)$ & $\text{SR}$ & $\text{ZZYS}_{\text{d}}$ & $\text{ZZYS}_{\text{h}}$ & $\text{GFLS}$ & $\text{WLF}_{\text{sum}}$ & $\text{WLF}_{\text{max}}$ \\ \hline
\multirow[t]{16}{*}{S2.3} & \multirow[t]{4}{*}{DGP1} & 100  & 0.3820           & 0.1160           & 0.1130           & 0.1130           & 0.1050      & 0.0840                   & 0.0610                   & 0.1100        & 0.1010                    & 0.0060                    \\
                       &                       & 200  & 0.4480           & 0.0700           & 0.0690           & 0.0690           & 0.0630      & 0.0520                   & 0.0530                   & 0.0710        & 0.0600                    & 0.0010                    \\
                       &                       & 500  & 0.6830           & 0.0490           & 0.0490           & 0.0490           & 0.0460      & 0.0470                   & 0.0530                   & 0.0490        & 0.0480                    & 0.0000                    \\
                       &                       & 1000 & 0.9050           & 0.0570           & 0.0570           & 0.0570           & 0.0560      & 0.0590                   & 0.0590                   & 0.0600        & 0.0590                    & 0.0000                    \\ \cline{2-13} 
                       & \multirow[t]{4}{*}{DGP2} & 100  & 0.4030           & 0.1250           & 0.1240           & 0.1240           & 0.1220      & 0.1140                   & 0.0800                   & 0.1260        & 0.1220                    & 0.0030                    \\
                       &                       & 200  & 0.4460           & 0.0770           & 0.0770           & 0.0770           & 0.0630      & 0.0730                   & 0.0620                   & 0.0730        & 0.0700                    & 0.0020                    \\
                       &                       & 500  & 0.6750           & 0.0700           & 0.0700           & 0.0700           & 0.0640      & 0.0440                   & 0.0510                   & 0.0680        & 0.0560                    & 0.0020                    \\
                       &                       & 1000 & 0.9000           & 0.0620           & 0.0620           & 0.0620           & 0.0560      & 0.0460                   & 0.0480                   & 0.0630        & 0.0530                    & 0.0000                    \\ \cline{2-13} 
                       & \multirow[t]{4}{*}{DGP3} & 100  & 0.3950           & 0.1280           & 0.1260           & 0.1260           & 0.1270      & 0.0960                   & 0.0730                   & 0.1330        & 0.0960                    & 0.0090                    \\
                       &                       & 200  & 0.4500           & 0.0850           & 0.0850           & 0.0850           & 0.0860      & 0.0740                   & 0.0710                   & 0.0910        & 0.0600                    & 0.0010                    \\
                       &                       & 500  & 0.6820           & 0.0560           & 0.0560           & 0.0560           & 0.0520      & 0.0440                   & 0.0570                   & 0.0570        & 0.0320                    & 0.0020                    \\
                       &                       & 1000 & 0.9090           & 0.0640           & 0.0640           & 0.0640           & 0.0600      & 0.0560                   & 0.0580                   & 0.0660        & 0.0360                    & 0.0000                    \\ \cline{2-13} 
                       & \multirow[t]{4}{*}{DGP4} & 100  & 0.4080           & 0.0990           & 0.0990           & 0.0990           & 0.1060      & 0.1000                   & 0.0840                   & 0.1090        & 0.0980                    & 0.0030                    \\
                       &                       & 200  & 0.4530           & 0.0780           & 0.0780           & 0.0780           & 0.0720      & 0.0630                   & 0.0590                   & 0.0770        & 0.0700                    & 0.0020                    \\
                       &                       & 500  & 0.6900           & 0.0630           & 0.0630           & 0.0630           & 0.0550      & 0.0370                   & 0.0340                   & 0.0630        & 0.0590                    & 0.0010                    \\
                       &                       & 1000 & 0.9070           & 0.0580           & 0.0580           & 0.0580           & 0.0560      & 0.0550                   & 0.0520                   & 0.0600        & 0.0520                    & 0.0000                    \\ \hline
\multirow[t]{16}{*}{S2.4} & \multirow[t]{4}{*}{DGP1} & 100  & 0.8970           & 0.4760           & 0.4740           & 0.4740           & 0.4330      & 0.3410                   & 0.1560                   & 0.4490        & 0.3710                    & 0.0080                    \\
                       &                       & 200  & 0.9460           & 0.3040           & 0.2990           & 0.2990           & 0.2810      & 0.2200                   & 0.1450                   & 0.2950        & 0.2460                    & 0.0000                    \\
                       &                       & 500  & 0.9960           & 0.1690           & 0.1680           & 0.1680           & 0.1630      & 0.1420                   & 0.0920                   & 0.1700        & 0.1690                    & 0.0010                    \\
                       &                       & 1000 & 1.0000           & 0.1250           & 0.1230           & 0.1230           & 0.1180      & 0.1080                   & 0.0920                   & 0.1250        & 0.0990                    & 0.0010                    \\ \cline{2-13} 
                       & \multirow[t]{4}{*}{DGP2} & 100  & 0.8950           & 0.5070           & 0.5010           & 0.5010           & 0.4770      & 0.4610                   & 0.2140                   & 0.4900        & 0.4800                    & 0.0040                    \\
                       &                       & 200  & 0.9610           & 0.3170           & 0.3110           & 0.3110           & 0.2990      & 0.2860                   & 0.1570                   & 0.3150        & 0.3080                    & 0.0010                    \\
                       &                       & 500  & 0.9970           & 0.2010           & 0.1990           & 0.1990           & 0.1890      & 0.1750                   & 0.0970                   & 0.1960        & 0.1820                    & 0.0020                    \\
                       &                       & 1000 & 1.0000           & 0.1250           & 0.1240           & 0.1240           & 0.1170      & 0.1110                   & 0.0830                   & 0.1250        & 0.1370                    & 0.0000                    \\ \cline{2-13} 
                       & \multirow[t]{4}{*}{DGP3} & 100  & 0.8840           & 0.4650           & 0.4630           & 0.4630           & 0.5000      & 0.4520                   & 0.2840                   & 0.5170        & 0.4390                    & 0.0050                    \\
                       &                       & 200  & 0.9520           & 0.3050           & 0.3040           & 0.3040           & 0.3220      & 0.2980                   & 0.2000                   & 0.3390        & 0.2990                    & 0.0020                    \\
                       &                       & 500  & 0.9950           & 0.1820           & 0.1790           & 0.1790           & 0.1900      & 0.1800                   & 0.1140                   & 0.2010        & 0.1290                    & 0.0010                    \\
                       &                       & 1000 & 1.0000           & 0.1280           & 0.1260           & 0.1260           & 0.1370      & 0.1230                   & 0.0960                   & 0.1420        & 0.0880                    & 0.0010                    \\ \cline{2-13} 
                       & \multirow[t]{4}{*}{DGP4} & 100  & 0.8910           & 0.4990           & 0.4920           & 0.4920           & 0.5070      & 0.4000                   & 0.2280                   & 0.5140        & 0.4350                    & 0.0050                    \\
                       &                       & 200  & 0.9590           & 0.3140           & 0.3110           & 0.3110           & 0.3280      & 0.2590                   & 0.1620                   & 0.3310        & 0.2730                    & 0.0040                    \\
                       &                       & 500  & 0.9990           & 0.1840           & 0.1820           & 0.1820           & 0.1810      & 0.1650                   & 0.1100                   & 0.1870        & 0.1330                    & 0.0010                    \\
                       &                       & 1000 & 1.0000           & 0.1480           & 0.1470           & 0.1470           & 0.1440      & 0.1140                   & 0.0960                   & 0.1550        & 0.1100                    & 0.0000                    \\ \hline
\end{tabular}}
\label{S2.3--S2.4}
\end{table}

\section{Empirical analysis}
\label{Empirical analysis}
In this section, we apply the proposed tests to examine the correlation between DNA methylation and gene expression in cancer. DNA methylation is an important epigenetic modification that regulates gene expression and has been shown to play an important role in cancer development and progression \citep{noushmehr2010identification,kim2017dna,su2018homeobox} as well as a variety of other diseases \citep{chen2014correlation}. The role of DNA methylation in cancer is particularly crucial, with numerous studies analyzing cancer-specific CpG hypermethylation that suppresses the expression of tumor suppressors and hypomethylation that activates oncogene expression \citep{sproul2012tissue,jones2012functions,schubeler2015function}. Consequently, assessing the dependence between DNA methylation and gene expression provides a natural and scientifically meaningful application for the proposed methods.

The data we use are obtained from The Cancer Genome Atlas (TCGA), a comprehensive cancer genomics initiative that provides rich information for analyzing various molecular aspects of cancer genetics. In particular, TCGA provides matched DNA methylation and gene expression data from the same patients, making it convenient to investigate the correlation between epigenetic regulation and transcriptional activity. Following \citet{rauschenberger2016testing}, we use the prostate cancer dataset, which is publicly available at \url{https://linkedomics.org/data_download/TCGA-PRAD/}. The dataset originally contains gene-level DNA methylation and gene expression measurements for 17,854 genes. To reduce noise and improve measurement reliability, we focus on genes with sufficient variability and expression levels. To this end, we exclude genes with either variance below the 20th percentile or mean expression level below the 10th percentile. Note that such preprocessing is standard in high-dimensional genomic data analysis; see also \citet{wang2026testing}. Finally, the dimensionality is reduced to $p=13685$, and the sample size is $n=497$. 

We first apply our proposed test $\hat{T}_n$ together with existing tests SR, ZZYS, GFLS, and WLF to the full sample to examine the correlation between DNA methylation and gene expression. The tuning parameters for all test statistics are set identically to those described in \autoref{Simulations}. All tests considered reject the null hypothesis of uncorrelatedness even at the 1\% significance level, providing strong evidence of a significant association between DNA methylation and gene expression. 

Nevertheless, because both the sample size and the dimensionality of the full dataset are extremely large, even relatively weak dependence signals may be detected with high probability. It is therefore of interest to examine whether the observed dependence structure persists in smaller subsamples, thereby assessing the robustness of the underlying correlation. For this purpose, we randomly draw subsamples with sample size $n=30$ and dimensions $p\in\{50,75,100,150\}$, and apply the same testing procedures to each subsample. For each $p$, the subsampling procedure is repeated 300 times, with both observations and variables randomly selected in each replication. \autoref{empirical p values} reports the maximum $p$-value among the 300 replications for each test. For most test statistics, the largest $p$-values across different choices of $p$ remain below the conventional 5\% significance level, suggesting that the underlying dependence structure persists even in relatively small subsamples. The only exception is $\text{WLF}_{\max}$, for which the maximum $p$-value fails to reject the null hypothesis. This finding is consistent with the numerical evidence reported in \autoref{Simulations}, where $\text{WLF}_{\max}$ is observed to be more conservative under dense alternatives. Overall, our analysis provides another statistical perspective on the mechanism of DNA methylation and offers valuable insights for future biological research.

\begin{table}[htbp]
\vspace{1em}
\renewcommand\arraystretch{0.85}
\centering
\caption{Testing $p$-values for subsamples of the prostate cancer dataset.} 
\resizebox{\linewidth}{!}{
\begin{tabular}{ccccccccccc}
\hline
$p$ & $\hat{T}_n(W_1)$ & $\hat{T}_n(W_2)$ & $\hat{T}_n(W_3)$ & $\hat{T}_n(W_4)$ & $\text{SR}$ & $\text{ZZYS}_{\text{d}}$ & $\text{ZZYS}_{\text{h}}$ & $\text{GFLS}$ & $\text{WLF}_{\text{sum}}$ & $\text{WLF}_{\text{max}}$ \\ \hline
50  & $0.0208$         & $0.0227$         & $0.0228$         & $0.0228$         & $0.0107$    & $0.0017$                 & $0.0146$                 & $0.0087$      & $0.0023$                  & $0.6130$                  \\
75  & $0.0022$         & $0.0012$         & $0.0012$         & $0.0012$         & $0.0004$    & $<0.0001$                & $0.0004$                 & $0.0003$      & $0.0001$                  & $0.4555$                  \\
100 & $0.0095$         & $0.0008$         & $0.0007$         & $0.0007$         & $0.0010$    & $<0.0001$                & $<0.0001$                & $0.0007$      & $<0.0001$                 & $0.7605$                  \\
150 & $0.0001$         & $<0.0001$        & $<0.0001$        & $<0.0001$        & $<0.0001$   & $<0.0001$                & $0.0028$                 & $<0.0001$     & $<0.0001$                 & $0.7341$                  \\ \hline
\end{tabular}}
\label{empirical p values}
\end{table}

\section{Conclusion}
\label{Concluding comments}
This paper proposes a novel paired-sample test for high-dimensional covariance matrices. We demonstrate that under the null hypothesis of uncorrelatedness, the proposed test statistic $\hat{T}_n$ is asymptotically standard normal as $n, p \rightarrow \infty$. Consequently, our framework circumvents the need for computationally intensive resampling or bootstrapping procedures to determine critical values. Furthermore, we establish general theoretical results showing that a properly specified weighting matrix $W$ can substantially enhance empirical power. Extensive numerical simulations validate the robust finite-sample performance of our methods. Compared with existing tests, the proposed test, equipped with a properly specified $W$, exhibits comparable size control and power improvement against ``weak but pervasive'' alternatives.

Note that the testing procedure of this paper can be further extended in several ways. For example, testing the equality of mean vectors or variance matrices for paired data. Testing the parametric specification or block-diagonal structure of the variance matrices may also be considered. Given the flexibility in designing the weighting matrix $W$, the resulting RID-based test statistics can adapt to different high-dimensional structures and may therefore exhibit improved sensitivity to certain alternatives. A data-driven selection procedure for the weighting matrix $W$ is also worth studying. A detailed investigation of these issues is beyond the scope of this article and is left for future research.

\setlength{\bibsep}{0.5em} 
\spacingset{1}
\bibliographystyle{chicago}
\bibliography{reference}

@book{hall2014martingale,
  title={Martingale limit theory and its application},
  author={Hall, Peter and Heyde, Christopher C},
  year={1980},
  publisher={Academic press}
}

@article{jiang2024nonparametric,
  title={Nonparametric two-sample tests of high dimensional mean vectors via random integration},
  author={Jiang, Yunlu and Wang, Xueqin and Wen, Canhong and Jiang, Yukang and Zhang, Heping},
  journal={Journal of the American Statistical Association},
  volume={119},
  number={545},
  pages={701--714},
  year={2024},
  publisher={Taylor \& Francis}
}

@article{chen2010two,
  title={A two-sample test for high-dimensional data with applications to gene-set testing},
  author={Chen, Song Xi and Qin, Ying-Li},
  year={2010}
}

@article{li2012two,
  title={Two sample tests for high-dimensional covariance matrices},
  author={Li, Jun and Chen, Song Xi},
  year={2012}
}

@article{bai1996effect,
  title={Effect of high dimension: by an example of a two sample problem},
  author={Bai, Zhidong and Saranadasa, Hewa},
  journal={Statistica Sinica},
  pages={311--329},
  year={1996},
  publisher={JSTOR}
}

@article{chen2010tests,
  title={Tests for high-dimensional covariance matrices},
  author={Chen, Song Xi and Zhang, Li-Xin and Zhong, Ping-Shou},
  journal={Journal of the American Statistical Association},
  volume={105},
  number={490},
  pages={810--819},
  year={2010},
  publisher={Taylor \& Francis}
}

@article{wang2015high,
  title={A high-dimensional nonparametric multivariate test for mean vector},
  author={Wang, Lan and Peng, Bo and Li, Runze},
  journal={Journal of the American Statistical Association},
  volume={110},
  number={512},
  pages={1658--1669},
  year={2015},
  publisher={Taylor \& Francis}
}

@article{jiang2023use,
  title={Use of random integration to test equality of high dimensional covariance matrices},
  author={Jiang, Yunlu and Wen, Canhong and Jiang, Yukang and Wang, Xueqin and Zhang, Heping},
  journal={Statistica Sinica},
  volume={33},
  number={4},
  pages={2359--2380},
  year={2023},
  publisher={JSTOR}
}

@article{gorecki2020independence,
  title={Independence test and canonical correlation analysis based on the alignment between kernel matrices for multivariate functional data},
  author={G{\'o}recki, Tomasz and Krzy{\'s}ko, Miros{\l}aw and Wo{\l}y{\'n}ski, Waldemar},
  journal={Artificial Intelligence Review},
  volume={53},
  number={1},
  pages={475--499},
  year={2020},
  publisher={Springer}
}

@article{zhang2024fast,
  title={A fast and accurate kernel-based independence test with applications to high-dimensional and functional data},
  author={Zhang, Jin-Ting and Zhu, Tianming},
  journal={Journal of Multivariate Analysis},
  volume={202},
  pages={105320},
  year={2024},
  publisher={Elsevier}
}

@article{li2023deep,
  title={Deep feature screening: Feature selection for ultra high-dimensional data via deep neural networks},
  author={Li, Kexuan and Wang, Fangfang and Yang, Lingli and Liu, Ruiqi},
  journal={Neurocomputing},
  volume={538},
  pages={126186},
  year={2023},
  publisher={Elsevier}
}

@article{pearson1895vii,
  title={VII. Note on regression and inheritance in the case of two parents},
  author={Pearson, Karl},
  journal={proceedings of the royal society of London},
  volume={58},
  number={347-352},
  pages={240--242},
  year={1895},
  publisher={The Royal Society London}
}

@article{kendall1938new,
  title={A new measure of rank correlation},
  author={Kendall, Maurice G},
  journal={Biometrika},
  volume={30},
  number={1-2},
  pages={81--93},
  year={1938},
  publisher={Oxford University Press}
}

@article{spearman1961proof,
  title={The proof and measurement of association between two things.},
  author={Spearman, Charles},
  year={1961},
  publisher={Appleton-Century-Crofts}
}

@incollection{hoeffding1994non,
  title={A non-parametric test of independence},
  author={Hoeffding, Wassily},
  booktitle={The Collected Works of Wassily Hoeffding},
  pages={214--226},
  year={1994},
  publisher={Springer}
}

@book{blum1961distribution,
  title={Distribution free tests of independence based on the sample distribution function},
  author={Blum, Julius R and Kiefer, Jack and Rosenblatt, Murray},
  year={1961},
  publisher={Sandia Corporation}
}

@article{szekely2007measuring,
  title={Measuring and testing dependence by correlation of distances},
  author={Sz{\'e}kely, G{\'a}bor J and Rizzo, Maria L and Bakirov, Nail K},
  year={2007}
}

@article{gretton2005kernel,
  title={Kernel methods for measuring independence},
  author={Gretton, Arthur and Herbrich, Ralf and Smola, Alexander and Bousquet, Olivier and Sch{\"o}lkopf, Bernhard},
  year={2005},
  publisher={MIT Press}
}

@article{bergsma2014consistent,
  title={A consistent test of independence based on a sign covariance related to Kendall's tau},
  author={Bergsma, Wicher and Dassios, Angelos},
  journal={Bernoulli},
  pages={1006--1028},
  year={2014},
  publisher={JSTOR}
}

@article{berrett2020conditional,
  title={The conditional permutation test for independence while controlling for confounders},
  author={Berrett, Thomas B and Wang, Yi and Barber, Rina Foygel and Samworth, Richard J},
  journal={Journal of the Royal Statistical Society Series B: Statistical Methodology},
  volume={82},
  number={1},
  pages={175--197},
  year={2020},
  publisher={Oxford University Press}
}

@article{shah2020hardness,
  title={The hardness of conditional independence testing and the generalised covariance measure},
  author={Shah, Rajen D and Peters, Jonas},
  year={2020}
}

@article{szekely2009brownian,
  title={Brownian distance covariance},
  author={Sz{\'e}kely, G{\'a}bor J and Rizzo, Maria L},
  year={2009}
}

@article{pan2020ball,
  title={Ball covariance: A generic measure of dependence in Banach space},
  author={Pan, Wenliang and Wang, Xueqin and Zhang, Heping and Zhu, Hongtu and Zhu, Jin},
  journal={Journal of the American Statistical Association},
  year={2020},
  publisher={Taylor \& Francis}
}

@article{jiang2026robust,
  title={Robust independence test of functional variables via Gupta angle covariance},
  author={Jiang, Yunlu and Yang, Yukun and Jiang, Ruizhe},
  journal={Statistics},
  pages={1--29},
  year={2026},
  publisher={Taylor \& Francis}
}

@article{weihs2018symmetric,
  title={Symmetric rank covariances: a generalized framework for nonparametric measures of dependence},
  author={Weihs, Luca and Drton, Mathias and Meinshausen, Nicolai},
  journal={Biometrika},
  volume={105},
  number={3},
  pages={547--562},
  year={2018},
  publisher={Oxford University Press}
}

@inproceedings{gretton2005measuring,
  title={Measuring statistical dependence with Hilbert-Schmidt norms},
  author={Gretton, Arthur and Bousquet, Olivier and Smola, Alex and Sch{\"o}lkopf, Bernhard},
  booktitle={International conference on algorithmic learning theory},
  pages={63--77},
  year={2005},
  organization={Springer}
}

@article{gao2025studentized,
  title={Studentized tests of independence: Random-lifter approach},
  author={Gao, Zhe and Wang, Roulin and Wang, Xueqin and Zhang, Heping},
  journal={The Annals of Statistics},
  volume={53},
  number={2},
  pages={703--723},
  year={2025},
  publisher={Institute of Mathematical Statistics}
}

@article{deb2023multivariate,
  title={Multivariate rank-based distribution-free nonparametric testing using measure transportation},
  author={Deb, Nabarun and Sen, Bodhisattva},
  journal={Journal of the American Statistical Association},
  volume={118},
  number={541},
  pages={192--207},
  year={2023},
  publisher={Taylor \& Francis}
}

@article{kojadinovic2009tests,
  title={Tests of independence among continuous random vectors based on Cram{\'e}r--von Mises functionals of the empirical copula process},
  author={Kojadinovic, Ivan and Holmes, Mark},
  journal={Journal of Multivariate Analysis},
  volume={100},
  number={6},
  pages={1137--1154},
  year={2009},
  publisher={Elsevier}
}

@article{genest2019testing,
  title={Testing for independence in arbitrary distributions},
  author={Genest, Christian and Ne{\v{s}}lehov{\'a}, Johanna G and R{\'e}millard, Bruno and Murphy, Orla A},
  journal={Biometrika},
  volume={106},
  number={1},
  pages={47--68},
  year={2019},
  publisher={Oxford University Press}
}

@article{szekely2013distance,
  title={The distance correlation t-test of independence in high dimension},
  author={Sz{\'e}kely, G{\'a}bor J and Rizzo, Maria L},
  journal={Journal of Multivariate Analysis},
  volume={117},
  pages={193--213},
  year={2013},
  publisher={Elsevier}
}

@article{gao2021asymptotic,
  title={Asymptotic distributions of high-dimensional distance correlation inference},
  author={Gao, Lan and Fan, Yingying and Lv, Jinchi and Shao, Qi-Man},
  journal={Annals of statistics},
  volume={49},
  number={4},
  pages={1999},
  year={2021}
}

@article{zhu2020distance,
  title={Distance-based and RKHS-based dependence metrics in high dimension},
  author={Zhu, Changbo and Zhang, Xianyang and Yao, Shun and Shao, Xiaofeng},
  journal={The Annals of Statistics},
  volume={48},
  number={6},
  pages={3366--3394},
  year={2020},
  publisher={JSTOR}
}

@article{zhou2024rank,
  title={Rank-based indices for testing independence between two high-dimensional vectors},
  author={Zhou, Yeqing and Xu, Kai and Zhu, Liping and Li, Runze},
  journal={Annals of statistics},
  volume={52},
  number={1},
  pages={184},
  year={2024}
}

@article{wang2026testing,
  title={Testing independence between high-dimensional random vectors using rank-based max-sum tests},
  author={Wang, Hongfei and Liu, Binghui and Feng, Long},
  journal={Scandinavian Journal of Statistics},
  year={2026},
  publisher={Wiley Online Library}
}

@article{jiang2026distribution,
  title={Distribution-free robust independence test for high-dimensional data via semi-Grothendieck’s covariance},
  author={Jiang, Ruizhe and Huang, Xiaowen and Gao, Zhe and Jiang, Yunlu},
  journal={Journal of Multivariate Analysis},
  pages={105655},
  year={2026},
  publisher={Elsevier}
}

@article{yang2015independence,
  title={Independence test for high dimensional data based on regularized canonical correlation coefficients},
  author={Yang, Yanrong and Pan, Guangming},
  year={2015}
}

@article{anderson1952asymptotic,
  title={Asymptotic theory of certain goodness of fit criteria based on stochastic processes},
  author={Anderson, Theodore W and Darling, Donald A},
  journal={The annals of mathematical statistics},
  pages={193--212},
  year={1952},
  publisher={JSTOR}
}

@article{delgado2006consistent,
  title={Consistent tests of conditional moment restrictions},
  author={Delgado, Miguel A and Dom{\'\i}nguez, Manuel A and Lavergne, Pascal},
  journal={Annales d'{\'E}conomie et de Statistique},
  pages={33--67},
  year={2006},
  publisher={JSTOR}
}

@article{yang2010common,
  title={Common SNPs explain a large proportion of the heritability for human height},
  author={Yang, Jian and Benyamin, Beben and McEvoy, Brian P and Gordon, Scott and Henders, Anjali K and Nyholt, Dale R and Madden, Pamela A and Heath, Andrew C and Martin, Nicholas G and Montgomery, Grant W and others},
  journal={Nature genetics},
  volume={42},
  number={7},
  pages={565--569},
  year={2010},
  publisher={Nature Publishing Group US New York}
}

@article{noushmehr2010identification,
  title={Identification of a CpG island methylator phenotype that defines a distinct subgroup of glioma},
  author={Noushmehr, Houtan and Weisenberger, Daniel J and Diefes, Kristin and Phillips, Heidi S and Pujara, Kanan and Berman, Benjamin P and Pan, Fei and Pelloski, Christopher E and Sulman, Erik P and Bhat, Krishna P and others},
  journal={Cancer cell},
  volume={17},
  number={5},
  pages={510--522},
  year={2010},
  publisher={Elsevier}
}

@article{kim2017dna,
  title={DNA methylation: an epigenetic mark of cellular memory},
  author={Kim, Mirang and Costello, Joseph},
  journal={Experimental \& molecular medicine},
  volume={49},
  number={4},
  pages={e322--e322},
  year={2017},
  publisher={Nature Publishing Group}
}

@article{su2018homeobox,
  title={Homeobox oncogene activation by pan-cancer DNA hypermethylation},
  author={Su, Jianzhong and Huang, Yung-Hsin and Cui, Xiaodong and Wang, Xinyu and Zhang, Xiaotian and Lei, Yong and Xu, Jianfeng and Lin, Xueqiu and Chen, Kaifu and Lv, Jie and others},
  journal={Genome biology},
  volume={19},
  number={1},
  pages={108},
  year={2018},
  publisher={Springer}
}

@article{chen2014correlation,
  title={Correlation between DNA methylation and gene expression in the brains of patients with bipolar disorder and schizophrenia},
  author={Chen, Chao and Zhang, Chunling and Cheng, Lijun and Reilly, James L and Bishop, Jeffrey R and Sweeney, John A and Chen, Hua-Yun and Gershon, Elliot S and Liu, Chunyu},
  journal={Bipolar disorders},
  volume={16},
  number={8},
  pages={790--799},
  year={2014},
  publisher={Wiley Online Library}
}

@article{sproul2012tissue,
  title={Tissue of origin determines cancer-associated CpG island promoter hypermethylation patterns},
  author={Sproul, Duncan and Kitchen, Robert R and Nestor, Colm E and Dixon, J Michael and Sims, Andrew H and Harrison, David J and Ramsahoye, Bernard H and Meehan, Richard R},
  journal={Genome biology},
  volume={13},
  number={10},
  pages={R84},
  year={2012},
  publisher={Springer}
}

@article{schubeler2015function,
  title={Function and information content of DNA methylation},
  author={Sch{\"u}beler, Dirk},
  journal={Nature},
  volume={517},
  number={7534},
  pages={321--326},
  year={2015},
  publisher={Nature Publishing Group UK London}
}

@article{jones2012functions,
  title={Functions of DNA methylation: islands, start sites, gene bodies and beyond},
  author={Jones, Peter A},
  journal={Nature reviews genetics},
  volume={13},
  number={7},
  pages={484--492},
  year={2012},
  publisher={Nature Publishing Group UK London}
}

@article{rauschenberger2016testing,
  title={Testing for association between RNA-Seq and high-dimensional data},
  author={Rauschenberger, Armin and Jonker, Marianne A and van de Wiel, Mark A and Menezes, Ren{\'e}e X},
  journal={BMC bioinformatics},
  volume={17},
  number={1},
  pages={118},
  year={2016},
  publisher={Springer}
}
\spacingset{1.8}

\newpage
\appendix
 \begin{center}
    {\LARGE\bf Paired Sample Tests for High-dimensional Uncorrelatedness via Random Integration}

    {\Large\bf Shiyao Huang and Xiaojun Song\\
Peking University}

    {\large\bf Online Supplementary Material}
\end{center}

\section{Preliminary lemmas}
\label{Preliminary lemmas}
For notational brevity, for each $i,j,t=1,2,\ldots,n$ and $i\neq j$, define $\xi_{i}={Z}_i^\prime
\Lambda
_3{Z}_i$, $\xi^\star_i=\xi_i-Q$, $\eta_{ij}={Z}_i^\prime{\Lambda}_1{Z}_j{Z}_j^\prime{\Lambda}_2{Z}_i$, and $\eta_{ij}^\star=\eta_{ij}-\xi_i-\xi_j+Q$. Then, define the conditional expectation $h_i=\mathbb{E}(\eta_{ij}\xi_{j}|{Z}_i)$, $d_{it}=\mathbb{E}(\eta_{ij}\eta_{tj}|{Z}_i,{Z}_t)$ and $\Psi_{it}=\mathbb{E}(\eta_{ij}^\star\eta_{tj}^\star|{Z}_i,{Z}_t)$.

\begin{lemma}
\label{Lemma1}
Let ${M}_1$ and ${M}_2$ denote any $d\times d$ matrices and ${M}$ denote a $d$-dimensional nonnegative and symmetric matrix. Then, under Assumption~3.1, the following statements hold.
\begin{align*}
&(1) \ \mathbb{E}({Z}_i^\prime{M}_1{Z}_i{Z}_i^\prime{M}_2{Z}_i)=\textnormal{tr}({M}_1)\textnormal{tr}({M}_2)+\textnormal{tr}({M}_1{M}_2)+\textnormal{tr}({M}_1{M}_2^\prime)+\Delta \textnormal{tr}({M}_1\circ {M}_2).\\
&(2) \ \mathbb{E}\big[\big({Z}_i^\prime{M}{Z}_i-\textnormal{tr}({M})\big)^4\big]\leqslant K\textnormal{tr}^2({M}^2).\\
&(3) \ \mathbb{E}[({Z}_i^\prime{M}{Z}_i)^2]\leqslant K\textnormal{tr}^2({M}),
\end{align*}
where the constant $K$ only depends on ${Z}$.
\end{lemma}

\begin{proof}[\textbf{Proof}] 
 For (1), Let ${M}_1=(m_{1,ij})_{d\times d}$ and ${M}_2=(m_{2,ij})_{d\times d}$, it then follows that
\begin{align}
\mathbb{E}({Z}_i^\prime{M}_1{Z}_i{Z}_i^\prime{M}_2{Z}_i)=\sum_{j_1,l_1=1}^{d}\sum_{j_2,l_2=1}^{d}m_{1,j_1l_1}m_{2,j_2l_2}\mathbb{E}(Z_{ij_1}Z_{il_1}Z_{ij_2}Z_{il_2}).
\label{Lemma1 1}
\end{align}
Note that under Assumption~3.1, 
\begin{align}
\begin{split}
\mathbb{E}(Z_{ij_1}Z_{il_1}Z_{ij_2}Z_{il_2})=
\begin{cases}
1, \renewcommand{\arraystretch}{0.7}\begin{array}{l l}
    &\text{if $(j_1=l_1)\neq (j_2=l_2)$} \\
    &\text{or $(j_1=j_2)\neq (l_1=l_2)$}\\
    &\text{or $(j_1=l_2)\neq (j_2=l_1)$}
\end{array} ,\\
3+\Delta, \textnormal{if } j_1=l_1=j_2=l_2,\\
0,\textnormal{otherwise}.
\end{cases}.
\end{split}
\label{Lemma1 2}
\end{align}
Gathering (\ref{Lemma1 1}) and (\ref{Lemma1 2}), the results directly follow:
\begin{align*}
\mathbb{E}({Z}_i^\prime{M}_1{Z}_i{Z}_i^\prime{M}_2{Z}_i)=&\sum_{j_1,l_1=1}^{d}\sum_{j_2,l_2=1}^{d}m_{1,j_1l_1}m_{2,j_2l_2}\mathbb{E}(Z_{ij_1}Z_{il_1}Z_{ij_2}Z_{il_2})\\
=&\sum_{j,l=1}^{d}m_{1,jj}m_{2,ll}+\sum_{j,l=1}^{d}m_{1,jl}m_{2,jl}\sum_{j,l=1}^{d}m_{1,jl}m_{2,lj}+\Delta\sum_{j=1}^{d}m_{1,jj}m_{2,jj}\\
=&\textnormal{tr}({M}_1)\textnormal{tr}({M}_2)+\textnormal{tr}({M}_1{M}_2)+\textnormal{tr}({M}_1{M}_2^\prime)+\Delta \textnormal{tr}({M}_1\circ {M}_2).
\end{align*}

(2) can be found in Proposition A.1 (iii) in Chen et al. (2010), so we omit the detailed proof. For (3), by Jensen's inequality and (2), we have
\begin{align*}
\Big(\mathbb{E}\big[\big({Z}_i^\prime{M}{Z}_i-\textnormal{tr}({M})\big)^2\big]\Big)^2\leqslant \mathbb{E}\big[\big({Z}_i^\prime{M}{Z}_i-\textnormal{tr}({M})\big)^4\big]\leqslant K\textnormal{tr}^2({M}^2).
\end{align*}
Thus, 
\begin{align*}
\sqrt{K}\textnormal{tr}({M}^2)\geqslant \mathbb{E}\big[\big({Z}_i^\prime{M}{Z}_i-\textnormal{tr}({M})\big)^2\big] =\mathbb{E}[({Z}_i^\prime{M}{Z}_i)^2]-\textnormal{tr}^2({M}),
\end{align*}
which implies that
\begin{align}
\mathbb{E}[({Z}_i^\prime{M}{Z}_i)^2]\leqslant  \textnormal{tr}^2({M})+\sqrt{K}\textnormal{tr}({M}^2).
\label{111}
\end{align}
Let $\lambda^{(m)}_i$, $i=1,2,\ldots,d$ denote the eigenvalues of ${M}$. Since ${M}$ is symmetric and nonnegative, we get $\lambda^{(m)}_i\geqslant 0$, $i=1,2,\ldots,d$. Then, using the singular value decomposition for nonnegative and symmetric matrices, we have
\begin{align}
\begin{split}
\textnormal{tr}({M}^2)=\sum_{i=1}^{d}(\lambda^{(m)}_i)^2\leqslant\left(\sum_{i=1}^{d}\lambda^{(m)}_i\right)^2=\textnormal{tr}^2({M}).
\end{split}
\label{222}
\end{align}
Gathering (\ref{111})--(\ref{222}) yields that
\begin{align*}
\mathbb{E}[({Z}_i^\prime{M}{Z}_i)^2]\leqslant  \textnormal{tr}^2({M})+\sqrt{K}\textnormal{tr}({M}^2)\leqslant (\sqrt{K}+1)\textnormal{tr}^2({M}),
\end{align*}
which completes our proof.
\end{proof}

\begin{lemma}
\label{Lemma2}
The following statements hold.
\begin{flalign*}
&(1)\ \max\{\textnormal{tr}({\Lambda}_3^2),\textnormal{tr}({\Lambda}_3\circ{\Lambda}_3)\}\leqslant \textnormal{tr}({\Lambda}_3{\Lambda}_3^\prime).&&\\
&(2)\ \textnormal{tr}(({\Lambda}_1\circ {\Lambda}_1)({\Lambda}_2\circ {\Lambda}_2))\leqslant\textnormal{tr}({\Lambda}_1^2\circ {\Lambda}_2^2)\leqslant \sqrt{\textnormal{tr}({\Lambda}_1^4)\textnormal{tr}({\Lambda}_2^4)}.\\
&(3)\ |\textnormal{tr}({\Lambda}_3^4)|\leqslant \textnormal{tr}({\Lambda}_3^2{\Lambda}_3^{\prime2})\leqslant \textnormal{tr}(({\Lambda}_3{\Lambda}_3^\prime)^2)\leqslant \textnormal{tr}^2({\Lambda}_3{\Lambda}_3^\prime).
\end{flalign*}
\end{lemma}

\begin{proof}[\textbf{Proof}]
Let ${M}_1=(m_{1,ij})_{k\times k}$ and ${M}_2=(m_{2,ij})_{k\times k}$ denote any $k\times k$ matrices. Note that 
\begin{align*}
&|\textnormal{tr}({M}_1\circ {M}_2)|=\left|\sum_{i=1}^km_{1,ii}m_{2,ii}\right|\leqslant \sum_{i=1}^k|m_{1,ii}m_{2,ii}|\leqslant \sum_{i,j=1}^k|m_{1,ij}m_{2,ji}|,\\
&|\textnormal{tr}({M}_1{M}_2)|=\left|\sum_{i,j=1}^km_{1,ij}m_{2,ji}\right|\leqslant \sum_{i,j=1}^k|m_{1,ij}m_{2,ji}|.
\end{align*}
Then, using the Cauchy--Schwarz inequality, we have
\begin{align}
\begin{split}
\max\{|\textnormal{tr}({M}_1\circ {M}_2)|,|\textnormal{tr}({M}_1{M}_2)|\}\leqslant &\sum_{i,j=1}^k|m_{1,ij}m_{2,ji}|\\
\leqslant &\sqrt{\left(\sum_{i,j=1}^km_{1,ij}^2\right)\left(\sum_{i,j=1}^km_{2,ij}^2\right)}\\
=&\sqrt{\textnormal{tr}({M}_1{M}_1^\prime)\textnormal{tr}({M}_2{M}_2^\prime)}.
\end{split}
\label{Cauchy inequality}
\end{align}
By (\ref{Cauchy inequality}), (1) directly follows.

For (2), let $\lambda_{1,ij}$ and $\lambda_{2,ij}$ denote the elements of ${\Lambda}_1$ and ${\Lambda}_2$, respectively. By (\ref{Cauchy inequality}), we have
\begin{align*}
\textnormal{tr}(({\Lambda}_1\circ {\Lambda}_1)({\Lambda}_2\circ {\Lambda}_2))=\sum_{j,k=1}^{d}\lambda_{1,jk}^2\lambda_{2,jk}^2\leqslant \sum_{j,k,l=1}^{d}\lambda_{1,jl}^2\lambda_{2,jk}^2=\textnormal{tr}({\Lambda}_1^2\circ {\Lambda}_2^2)\leqslant \sqrt{\textnormal{tr}({\Lambda}_1^4)\textnormal{tr}({\Lambda}_2^4)}.
\end{align*}

For (3), again, using (\ref{Cauchy inequality}), we obtain
\begin{align*}
|\textnormal{tr}({\Lambda}_3^4)|\leqslant\textnormal{tr}({\Lambda}_3^2{\Lambda}_3^{\prime 2})=\textnormal{tr}({\Lambda}_3{\Lambda}_3^\prime{\Lambda}_3^\prime{\Lambda}_3)\leqslant\sqrt{\textnormal{tr}(({\Lambda}_3{\Lambda}_3^\prime)^2)\textnormal{tr}(({\Lambda}_3^\prime{\Lambda}_3)^2)}=\textnormal{tr}(({\Lambda}_3{\Lambda}_3^\prime)^2).
\end{align*}
Note that ${\Lambda}_3{\Lambda}_3^\prime$ is symmetric and nonnegative. Therefore, applying the inequality (\ref{222}) to ${\Lambda}_3{\Lambda}_3^\prime$, we get $\textnormal{tr}(({\Lambda}_3{\Lambda}_3^\prime)^2)\leqslant \textnormal{tr}^2({\Lambda}_3{\Lambda}_3^\prime)$, which then completes our proof.
\end{proof}

\begin{lemma}
\label{Lemma3}
Under Assumption~3.1, for $i,t=1,2,\ldots,n$, the following statements hold.
\begin{align*}
&(1)\ d_{it}={Z}_i^\prime{\Lambda}_3{Z}_i{Z}_t^\prime{\Lambda}_3{Z}_t+{Z}_i^\prime{\Lambda}_3{Z}_t{Z}_t^\prime{\Lambda}_3{Z}_i+{Z}_i^\prime{\Lambda}_1^2{Z}_t{Z}_t^\prime{\Lambda}_2^2{Z}_i+\Delta\textnormal{tr}({\Lambda}_1{Z}_i{Z}_i^\prime{\Lambda}_2\circ {\Lambda}_1{Z}_t{Z}_t^\prime{\Lambda}_2).\\
&(2)\ h_i=\textnormal{tr}({\Lambda}_3){Z}_i^\prime{\Lambda}_3{Z}_i+{Z}_i^\prime{\Lambda}_3^2{Z}_i+{Z}_i^\prime{\Lambda}_1^2{\Lambda}_2^2{Z}_i+\Delta\textnormal{tr}({\Lambda}_1{Z}_i{Z}_i^\prime{\Lambda}_2\circ{\Lambda}_3).\\
&(3)\ \Psi_{it}=d_{it}-h_i-h_t-\xi_i^\star\xi_{t}^\star+\mathbb{E}(\xi_j^2).
\end{align*}
\end{lemma}

\begin{proof}[\textbf{Proof}]
By applying \autoref{Lemma1}(1) to $d_{it}$ and $h_i$, we get
\begin{align*}
d_{it}=&\mathbb{E}(\eta_{ij}\eta_{tj}|{Z}_i,{Z}_t)\\
=&\mathbb{E}({Z}_j^\prime{\Lambda}_1{Z}_i{Z}_i^\prime{\Lambda}_2{Z}_j{Z}_j^\prime{\Lambda}_1{Z}_t{Z}_t^\prime{\Lambda}_2{Z}_j|{Z}_i,{Z}_t)\\
=&\textnormal{tr}({\Lambda}_1{Z}_i{Z}_i^\prime{\Lambda}_2)\textnormal{tr}({\Lambda}_1{Z}_t{Z}_t^\prime{\Lambda}_2)+\textnormal{tr}({\Lambda}_1{Z}_i{Z}_i^\prime{\Lambda}_2{\Lambda}_1{Z}_t{Z}_t^\prime{\Lambda}_2)+\textnormal{tr}({\Lambda}_1{Z}_i{Z}_i^\prime{\Lambda}_2{\Lambda}_2{Z}_t{Z}_t^\prime{\Lambda}_1)\\
&+\Delta\textnormal{tr}({\Lambda}_1{Z}_i{Z}_i^\prime{\Lambda}_2\circ {\Lambda}_1{Z}_t{Z}_t^\prime{\Lambda}_2)\\
=&{Z}_i^\prime{\Lambda}_3{Z}_i{Z}_t^\prime{\Lambda}_3{Z}_t+{Z}_i^\prime{\Lambda}_3{Z}_t{Z}_t^\prime{\Lambda}_3{Z}_i+{Z}_i^\prime{\Lambda}_1^2{Z}_t{Z}_t^\prime{\Lambda}_2^2{Z}_i+\Delta\textnormal{tr}({\Lambda}_1{Z}_i{Z}_i^\prime{\Lambda}_2\circ {\Lambda}_1{Z}_t{Z}_t^\prime{\Lambda}_2)
\end{align*}
and 
\begin{align*}
h_i=&\mathbb{E}(\eta_{ij}\xi_j|{Z}_i)\\
=&\mathbb{E}({Z}_j^\prime{\Lambda}_1{Z}_i{Z}_i^\prime{\Lambda}_2{Z}_j{Z}_j^\prime{\Lambda}_3{Z}_j|{Z}_i)\\
=&\textnormal{tr}({\Lambda}_1{Z}_i{Z}_i^\prime{\Lambda}_2)\textnormal{tr}({\Lambda}_3)+\textnormal{tr}({\Lambda}_1{Z}_i{Z}_i^\prime{\Lambda}_2{\Lambda}_3)+\textnormal{tr}({\Lambda}_1{Z}_i{Z}_i^\prime{\Lambda}_2{\Lambda}_3^\prime)+\Delta\textnormal{tr}({\Lambda}_1{Z}_i{Z}_i^\prime{\Lambda}_2\circ{\Lambda}_3)\\
=&\textnormal{tr}({\Lambda}_3){Z}_i^\prime{\Lambda}_3{Z}_i+{Z}_i^\prime{\Lambda}_3^2{Z}_i+{Z}_i^\prime{\Lambda}_1^2{\Lambda}_2^2{Z}_i+\Delta\textnormal{tr}({\Lambda}_1{Z}_i{Z}_i^\prime{\Lambda}_2\circ{\Lambda}_3).
\end{align*}   
Finally, by straightforward calculations, we obtain
\begin{align*}
\Psi_{it}=&\mathbb{E}(\eta_{ij}^\star\eta_{tj}^\star|{Z}_i,{Z}_t)\\
=&\mathbb{E}\big[\big((\eta_{ij}-\xi_i)-(\xi_j-Q)\big)\big((\eta_{tj}-\xi_t)-(\xi_j-Q)\big)|{Z}_i,{Z}_t\big]\\
=&\mathbb{E}\big[(\eta_{ij}-\xi_i)(\eta_{tj}-\xi_t)|{Z}_i,{Z}_t\big]-\mathbb{E}\big[(\eta_{ij}-\xi_i)(\xi_{j}-Q)|{Z}_i\big]\\
&-\mathbb{E}\big[(\eta_{tj}-\xi_t)(\xi_{j}-Q)|{Z}_t\big]+\textnormal{Var}(\xi_j)\\
=&d_{it}-h_i-h_t-\xi_i^\star\xi_{t}^\star+\mathbb{E}(\xi_j^2),
\end{align*}
which completes the proof.
\end{proof}

\begin{lemma}
\label{Lemma4}
Under Assumption~3.1, for $i,j=1,2,\ldots,n$ and $i\neq j$, the following statements hold.
\begin{align*}
(1)\ \mathbb{E}(\xi_i^{\star 2})=&\textnormal{tr}(({W}{\Sigma}{W}{\Sigma}^\prime)^2)+\textnormal{tr}({W}{\Sigma}_1{W}{\Sigma}{W}{\Sigma}_2{W}{\Sigma}^\prime)+\Delta\textnormal{tr}({\Lambda}_3\circ {\Lambda}_3).\\
(2)\ \mathbb{E}(\eta_{ij}^{\star 2})=&\textnormal{tr}(({W}{\Sigma}_1)^2)\textnormal{tr}(({W}{\Sigma}_2)^2)+\textnormal{tr}^2({W}{\Sigma}{W}{\Sigma}^\prime)+2\textnormal{tr}({W}{\Sigma}_1{W}{\Sigma}{W}{\Sigma}_2{W}{\Sigma}^\prime)\\
&+2\Delta\textnormal{tr}({\Lambda}_3\circ{\Lambda}_3)+2\Delta\textnormal{tr}({\Lambda}_1^2\circ{\Lambda}_2^2)+\Delta^2\textnormal{tr}\big(({\Lambda}_1\circ {\Lambda}_1)({\Lambda}_2\circ {\Lambda}_2)\big).
\end{align*}
\end{lemma}

\begin{proof}[\textbf{Proof}]
By \autoref{Lemma1}(1), we have
\begin{align}
\begin{split}
\mathbb{E}(\xi_i^2)=&\mathbb{E}({Z}_i^\prime{\Lambda}_3{Z}_i{Z}_i^\prime{\Lambda}_3{Z}_i)\\
=&\textnormal{tr}^2({\Lambda}_3)+\textnormal{tr}({\Lambda}_3{\Lambda}_3^\prime)+\textnormal{tr}({\Lambda}_3^2)+\Delta\textnormal{tr}({\Lambda}_3\circ {\Lambda}_3)\\
=&Q^2+\textnormal{tr}(({W}{\Sigma}{W}{\Sigma}^\prime)^2)+\textnormal{tr}({W}{\Sigma}_1{W}{\Sigma}{W}{\Sigma}_2{W}{\Sigma}^\prime)+\Delta\textnormal{tr}({\Lambda}_3\circ {\Lambda}_3).
\end{split}
\label{Lemma4 1}
\end{align}
Then, (1) immediately yields by gathering (\ref{Lemma4 1}) and $\mathbb{E}(\xi_i^{\star 2})=\mathbb{E}(\xi_i^2)-Q^2$.

For (2), by the Law of Iterated Expectations, we get $\mathbb{E}(\eta_{ij}^{\star 2})=\mathbb{E}(\Psi_{ii})$. Then, following \autoref{Lemma1}(1) and \autoref{Lemma3}, we obtain
\begin{align}
\begin{split}
\mathbb{E}(d_{ii})=&2\mathbb{E}({Z}_i^\prime{\Lambda}_3{Z}_i{Z}_i^\prime{\Lambda}_3{Z}_i)+\mathbb{E}({Z}_i^\prime{\Lambda}_1^2{Z}_i{Z}_i^\prime{\Lambda}_2^2{Z}_i)+\Delta\mathbb{E}\big(\textnormal{tr}({\Lambda}_1{Z}_i{Z}_i^\prime{\Lambda}_2\circ {\Lambda}_1{Z}_i{Z}_i^\prime{\Lambda}_2)\big)\\
=&2\mathbb{E}(\xi_i^2)+\textnormal{tr}({\Lambda}_1^2)\textnormal{tr}({\Lambda}_2^2)+2\textnormal{tr}({\Lambda}_1^2{\Lambda}_2^2)+\Delta\textnormal{tr}({\Lambda}_1^2\circ {\Lambda}_2^2)+\Delta\mathbb{E}\big(\textnormal{tr}({\Lambda}_1{Z}_i{Z}_i^\prime{\Lambda}_2\circ {\Lambda}_1{Z}_i{Z}_i^\prime{\Lambda}_2)\big).
\end{split}
\label{Lemma4 2}
\end{align}
For the last term in (\ref{Lemma4 2}), note that the $(k,k)$th element of ${\Lambda}_1{Z}_i{Z}_i^\prime{\Lambda}_2$ is $\sum_{l,t=1}^{d}\lambda_{1,kl}\lambda_{2,kt}Z_{il}Z_{it}$, where $\lambda_{i,jk}$ denotes the ${j,k}$th element of ${\Lambda}_i$, $i=1,2$. Then, it follows that 
\begin{align*}
&\mathbb{E}\big(\textnormal{tr}({\Lambda}_1{Z}_i{Z}_i^\prime{\Lambda}_2\circ {\Lambda}_1{Z}_i{Z}_i^\prime{\Lambda}_2)\big)\\
=&\sum_{k=1}^{d}\mathbb{E}\left(\sum_{l,t=1}^{d}\lambda_{1,kl}\lambda_{2,kt}Z_{il}Z_{it}\right)^2\\
=&\sum_{k=1}^{d}\sum_{l_1,t_1=1}^{d}\sum_{l_2,t_2=1}^{d}\lambda_{1,kl_1}\lambda_{2,kt_1}\lambda_{1,kl_2}\lambda_{2,kt_2}\mathbb{E}(Z_{il_1}Z_{it_1}Z_{il_2}Z_{it_2}).
\end{align*}
Again, using (\ref{Lemma1 2}), the expression is further simplified as
\begin{align}
\begin{split}
&\mathbb{E}\big(\textnormal{tr}({\Lambda}_1{Z}_i{Z}_i^\prime{\Lambda}_2\circ {\Lambda}_1{Z}_i{Z}_i^\prime{\Lambda}_2)\big)\\
=&\sum_{k=1}^{d}\sum_{l_1,t_1=1}^{d}\sum_{l_2,t_2=1}^{d}\lambda_{1,kl_1}\lambda_{2,kt_1}\lambda_{1,kl_2}\lambda_{2,kt_2}\mathbb{E}(Z_{il_1}Z_{it_1}Z_{il_2}Z_{it_2})\\
=&\sum_{k=1}^{d}\left(2\sum_{l,t=1}^{d}\lambda_{1,kl}\lambda_{2,kl}\lambda_{1,kt}\lambda_{2,kt}+\sum_{l,t=1}^{d}\lambda_{1,kl}^2\lambda_{2,kt}^2+\Delta\sum_{l=1}^{d}\lambda_{1,kl}^2\lambda_{2,kl}^2\right)\\
=&2\textnormal{tr}({\Lambda}_3\circ {\Lambda}_3)+\textnormal{tr}({\Lambda}_1^2\circ {\Lambda}_2^2)+\Delta\textnormal{tr}(({\Lambda}_1\circ {\Lambda}_1)({\Lambda}_2\circ {\Lambda}_2))
\end{split}
\label{Lemma4 3}
\end{align}
Gathering (\ref{Lemma4 2})--(\ref{Lemma4 3}), we obtain
\begin{align}
\begin{split}
\mathbb{E}(d_{ii})=&2\mathbb{E}(\xi_i^2)+\textnormal{tr}({\Lambda}_1^2)\textnormal{tr}({\Lambda}_2^2)+2\textnormal{tr}({\Lambda}_1^2{\Lambda}_2^2)+\Delta\textnormal{tr}({\Lambda}_1^2\circ {\Lambda}_2^2)+\Delta\Big(2\textnormal{tr}({\Lambda}_3\circ {\Lambda}_3)+\textnormal{tr}({\Lambda}_1^2\circ {\Lambda}_2^2)\Big)\\
&+\Delta^2\textnormal{tr}(({\Lambda}_1\circ {\Lambda}_1)({\Lambda}_2\circ {\Lambda}_2))\\
=&2\mathbb{E}(\xi_i^2)+\textnormal{tr}(({W}{\Sigma}_1)^2)\textnormal{tr}(({W}{\Sigma}_2)^2)+2\textnormal{tr}({W}{\Sigma}_1{W}{\Sigma}{W}{\Sigma}_2{W}{\Sigma}^\prime)+2\Delta\textnormal{tr}({\Lambda}_3\circ{\Lambda}_3)\\
&+2\Delta\textnormal{tr}({\Lambda}_1^2\circ{\Lambda}_2^2)+\Delta^2\textnormal{tr}\big(({\Lambda}_1\circ {\Lambda}_1)({\Lambda}_2\circ {\Lambda}_2)\big).
\end{split}
\label{Lemma4 4}
\end{align}
Note that $\mathbb{E}(h_i)=\mathbb{E}(\xi_i^2)$. Therefore, by \autoref{Lemma3}, (\ref{Lemma4 1}), and (\ref{Lemma4 4}), the desired results for (2) yield 
\begin{align}
\begin{split}
\mathbb{E}(\eta_{ij}^{\star 2})=&\mathbb{E}(\Psi_{ii})\\
=&\mathbb{E}(d_{ii})-2\mathbb{E}(h_i)-\mathbb{E}(\xi_i^{\star 2})+\mathbb{E}(\xi_i^{2})\\
=&\mathbb{E}(d_{ii})-2\mathbb{E}(\xi_i^{2})+Q^2\\
=&\textnormal{tr}(({W}{\Sigma}_1)^2)\textnormal{tr}(({W}{\Sigma}_2)^2)+\textnormal{tr}^2({W}{\Sigma}{W}{\Sigma}^\prime)+2\textnormal{tr}({W}{\Sigma}_1{W}{\Sigma}{W}{\Sigma}_2{W}{\Sigma}^\prime)+2\Delta\textnormal{tr}({\Lambda}_3\circ{\Lambda}_3)\\
&+2\Delta\textnormal{tr}({\Lambda}_1^2\circ{\Lambda}_2^2)+\Delta^2\textnormal{tr}\big(({\Lambda}_1\circ {\Lambda}_1)({\Lambda}_2\circ {\Lambda}_2)\big),
\end{split}
\label{Lemma4 5}
\end{align}
which ends the proof.
\end{proof}

\begin{lemma}
\label{Lemma5}
Under Assumptions 3.1--3.3, for $i,j=1,2,\ldots,n$ and $i\neq j$, the following statements hold.
\begin{align*}
&(1) \ \mathbb{E}(\Psi_{ii})=\textnormal{tr}(({W}{\Sigma}_1)^2)\textnormal{tr}(({W}{\Sigma}_2)^2)+o\Big(\textnormal{tr}(({W}{\Sigma}_1)^2)\textnormal{tr}(({W}{\Sigma}_2)^2)\Big).\\
&(2) \ \mathbb{E}(\Psi_{ij}^2)=o\Big(\textnormal{tr}^2(({W}{\Sigma}_1)^2)\textnormal{tr}^2(({W}{\Sigma}_2)^2)\Big).\\
&(3) \ \mathbb{E}(\Psi_{ii}^2)=O\Big(\textnormal{tr}^2(({W}{\Sigma}_1)^2)\textnormal{tr}^2(({W}{\Sigma}_2)^2)\Big).\\
&(4)\ \mathbb{E}(\eta_{ij}^{\star 4})=O\Big(\textnormal{tr}^2(({W}{\Sigma}_1)^2)\textnormal{tr}^2(({W}{\Sigma}_2)^2)\Big).
\end{align*}
\end{lemma}

\begin{proof}[\textbf{Proof}]
Note that $\textnormal{tr}({\Lambda}_3{\Lambda}_3^\prime)=\textnormal{tr}({W}{\Sigma}_1{W}{\Sigma}{W}{\Sigma}_2{\Sigma}^\prime)$, $\textnormal{tr}({\Lambda}_i^4)=\textnormal{tr}(({W}{\Sigma}_i)^4)$, $i=1,2$. Then, under Assumptions 3.2--3.3, (1) directly follows by gathering (\ref{Lemma4 5}) and \autoref{Lemma2}.

To prove (2), note that $\mathbb{E}(\Psi_{ij}^2)$ is bounded by
\begin{align}
\begin{split}
\mathbb{E}(\Psi_{ij}^2)\leqslant &5\mathbb{E}(d_{ij}^2)+10\mathbb{E}(h_i^2)+5\mathbb{E}^2(\xi_i^{\star 2})+5\mathbb{E}^2(\xi_i^{2})\\
\leqslant &5\mathbb{E}(d_{ij}^2)+10\mathbb{E}(h_i^2)+10\mathbb{E}^2(\xi_i^{2}).
\end{split}
\label{Lemma5 0}
\end{align}
Note that $\textnormal{tr}(({W}{\Sigma}{W}{\Sigma}^\prime)^2)=\textnormal{tr}({\Lambda}_3^2)\leqslant \textnormal{tr}({\Lambda}_3{\Lambda}_3^\prime)=\textnormal{tr}({W}{\Sigma}_1{W}{\Sigma}{W}{\Sigma}_2{W}{\Sigma}^\prime)$ by (\ref{Cauchy inequality}). Therefore, under Assumption~3.3, gathering (\ref{Lemma4 1}) and \autoref{Lemma2} yields that 
\begin{align}
\begin{split}
\mathbb{E}^2(\xi_i^{2})\leqslant &\Big(\textnormal{tr}^2({W}{\Sigma}{W}{\Sigma}^\prime)+(2+\Delta)\textnormal{tr}({W}{\Sigma}_1{W}{\Sigma}{W}{\Sigma}_2{W}{\Sigma}^\prime)\Big)^2\\
=&o\Big(\textnormal{tr}^2(({W}{\Sigma}_1)^2)\textnormal{tr}^2(({W}{\Sigma}_2)^2)\Big).
\end{split}
\label{Lemma5 1}
\end{align}

By \autoref{Lemma3}, $\mathbb{E}(d_{ij}^2)$ can be bounded by
\begin{align}
\begin{split}
\mathbb{E}(d_{ij}^2)\leqslant& 4\mathbb{E}^2\big[({Z}_i^\prime{\Lambda}_3{Z}_i)^2\big]+4\mathbb{E}\big[({Z}_i^\prime{\Lambda}_3{Z}_j{Z}_j^\prime{\Lambda}_3{Z}_i)^2\big]+4\mathbb{E}\big[({Z}_i^\prime{\Lambda}_1^2{Z}_j{Z}_j^\prime{\Lambda}_2^2{Z}_i)^2\big]\\
&+4\Delta^2\mathbb{E}\big[\textnormal{tr}^2({\Lambda}_1{Z}_i{Z}_i^\prime{\Lambda}_2\circ {\Lambda}_1{Z}_j{Z}_j^\prime{\Lambda}_2)\big]\\
=&4\mathbb{E}^2(\xi_i^2)+4\mathbb{E}\big[({Z}_i^\prime{\Lambda}_3{Z}_j{Z}_j^\prime{\Lambda}_3{Z}_i)^2\big]+4\mathbb{E}\big[({Z}_i^\prime{\Lambda}_1^2{Z}_j{Z}_j^\prime{\Lambda}_2^2{Z}_i)^2\big]\\
&+4\Delta^2\mathbb{E}\big[\textnormal{tr}^2({\Lambda}_1{Z}_i{Z}_i^\prime{\Lambda}_2\circ {\Lambda}_1{Z}_j{Z}_j^\prime{\Lambda}_2)\big].
\end{split}
\label{Lemma5 2}
\end{align}
To calculate the bound of $\mathbb{E}(d_{ij}^2)$, we first show that for any $d\times d$ matrix ${M}$ and $i\neq j$, 
\begin{align}
\mathbb{E}\left[({Z}_i{M}{Z}_j)^4\right]\leqslant K^2\textnormal{tr}^2({M}{M}^\prime),
\label{bound power 4}
\end{align}
where the constant $K$ is given in \autoref{Lemma1}. In fact, using the Law of Iterated Expectations and \autoref{Lemma1}(3), (\ref{bound power 4}) can be proved straightforward:
\begin{align*}
\mathbb{E}\left[({Z}_i{M}{Z}_j)^4\right]=\mathbb{E}\left[({Z}_i^\prime({M}{Z}_j{Z}_j^\prime{M}^\prime){Z}_i)^2\right]\leqslant K\mathbb{E}\left[({Z}_j^\prime{M}^\prime{M}{Z}_j)^2\right]\leqslant K^2\textnormal{tr}^2({M}{M}^\prime).
\end{align*}
Then, by Cauchy--Schwarz inequality and (\ref{bound power 4}), it holds that
\begin{align}
\begin{split}
\mathbb{E}\big[({Z}_i^\prime{\Lambda}_3{Z}_j{Z}_j^\prime{\Lambda}_3{Z}_i)^2\big]\leqslant &\mathbb{E}\big[({Z}_i^\prime{\Lambda}_3{Z}_j)^4\big]\leqslant K^2\textnormal{tr}^2({\Lambda}_3{\Lambda}_3^\prime),\\
\mathbb{E}\left[({Z}_i^\prime{\Lambda}_1^2{Z}_j{Z}_j^\prime{\Lambda}_2^2{Z}_i)^2\right]\leqslant &\sqrt{\mathbb{E}\left[({Z}_i^\prime{\Lambda}_1^2{Z}_j)^4\right]\mathbb{E}\left[({Z}_i^\prime{\Lambda}_2^2{Z}_j)^4\right]}\leqslant K^2\textnormal{tr}({\Lambda}_1^4)\textnormal{tr}({\Lambda}_2^4).
\end{split}
\label{bound dij 23_1}
\end{align}
Since $\textnormal{tr}({\Lambda}_3{\Lambda}_3^\prime)=\textnormal{tr}({W}{\Sigma}_1{W}{\Sigma}{W}{\Sigma}_2{\Sigma}^\prime)$, $\textnormal{tr}({\Lambda}_i^4)=\textnormal{tr}(({W}{\Sigma}_i)^4)$, $i=1,2$, (\ref{bound dij 23_1}) implies that
\begin{align}
\begin{split}
&\mathbb{E}\big[({Z}_i^\prime{\Lambda}_3{Z}_j{Z}_j^\prime{\Lambda}_3{Z}_i)^2\big]=o\Big(\textnormal{tr}^2(({W}{\Sigma}_1)^2)\textnormal{tr}^2(({W}{\Sigma}_2)^2)\Big),\\
&\mathbb{E}\left[({Z}_i^\prime{\Lambda}_1^2{Z}_j{Z}_j^\prime{\Lambda}_2^2{Z}_i)^2\right]=o\Big(\textnormal{tr}^2(({W}{\Sigma}_1)^2)\textnormal{tr}^2(({W}{\Sigma}_2)^2)\Big).
\end{split}
\label{bound dij 23_2}
\end{align}
For the last term in (\ref{Lemma5 2}), under Assumptions 3.2--3.3, using the Cauchy--Schwarz inequality, (\ref{Lemma4 3}), and \autoref{Lemma2}, we obtain that
\begin{align}
\begin{split}
&\mathbb{E}\big[\textnormal{tr}^2({\Lambda}_1{Z}_i{Z}_i^\prime{\Lambda}_2\circ {\Lambda}_1{Z}_j{Z}_j^\prime{\Lambda}_2)\big]\\\leqslant & \mathbb{E}^2\big[\textnormal{tr}({\Lambda}_1{Z}_i{Z}_i^\prime{\Lambda}_2\circ {\Lambda}_1{Z}_i{Z}_i^\prime{\Lambda}_2)\big]\\
=&\Big(2\textnormal{tr}({\Lambda}_3\circ {\Lambda}_3)+\textnormal{tr}({\Lambda}_1^2\circ {\Lambda}_2^2)+\Delta\textnormal{tr}(({\Lambda}_1\circ {\Lambda}_1)({\Lambda}_2\circ {\Lambda}_2))\Big)^2\\
\leqslant &\Big(2\textnormal{tr}({\Lambda}_3{\Lambda}_3^\prime)+(1+\Delta)\sqrt{\textnormal{tr}({\Lambda}_1^4)\textnormal{tr}({\Lambda}_2^4)}\Big)^2\\
=&o\Big(\textnormal{tr}^2(({W}{\Sigma}_1)^2)\textnormal{tr}^2(({W}{\Sigma}_2)^2)\Big).
\end{split}
\label{bound dij 4}
\end{align}
Gathering (\ref{Lemma5 1}), (\ref{Lemma5 2}), (\ref{bound dij 23_2}), and (\ref{bound dij 4}) implies that
\begin{align}
\begin{split}
\mathbb{E}(d_{ij}^2)\leqslant& 4\mathbb{E}^2\big[(\xi_i^2)\big]+4\mathbb{E}\big[({Z}_i^\prime{\Lambda}_3{Z}_j{Z}_j^\prime{\Lambda}_3{Z}_i)^2\big]+4\mathbb{E}\big[({Z}_i^\prime{\Lambda}_1^2{Z}_j{Z}_j^\prime{\Lambda}_2^2{Z}_i)^2\big]\\
&+4\Delta^2\mathbb{E}\big[\textnormal{tr}^2({\Lambda}_1{Z}_i{Z}_i^\prime{\Lambda}_2\circ {\Lambda}_1{Z}_j{Z}_j^\prime{\Lambda}_2)\big]\\
=&o\Big(\textnormal{tr}^2(({W}{\Sigma}_1)^2)\textnormal{tr}^2(({W}{\Sigma}_2)^2)\Big).
\end{split}
\label{bound dij}
\end{align}

For $\mathbb{E}(h_i^2)$, note that it can be bounded by
\begin{align}
\begin{split}
\mathbb{E}(h_i^2)\leqslant &4\Big(\textnormal{tr}^2({\Lambda}_3)\mathbb{E}(\xi_i^2)+\mathbb{E}\big[({Z}_i^\prime{\Lambda}_3^2{Z}_i)^2\big]+\mathbb{E}\big[({Z}_i^\prime{\Lambda}_1^2{\Lambda}_2^2{Z}_i)^2\big]+\Delta^2\mathbb{E}\big[\textnormal{tr}^2({\Lambda}_1{Z}_i{Z}_i^\prime{\Lambda}_2\circ {\Lambda}_3)\big]\Big).
\end{split}
\label{Lemma5 21}
\end{align}
Since $\textnormal{tr}({\Lambda}_3)=\textnormal{tr}({W}{\Sigma}{W}{\Sigma}^\prime)$, by (\ref{Lemma5 1}), we get
\begin{align}
\textnormal{tr}^2({\Lambda}_3)\mathbb{E}(\xi_i^2)=o\Big(\textnormal{tr}^2(({W}{\Sigma}_1)^2)\textnormal{tr}^2(({W}{\Sigma}_2)^2)\Big).
\label{Lemma5 22}
\end{align}

Next, we establish that, for any $d\times d$ matrices ${M}_1$ and ${M}_2$, 
\begin{align}
\mathbb{E}({Z}_i^\prime{M}_1{M}_2{Z}_i)^2\leqslant K^2\textnormal{tr}({M}_1{M}_1^\prime)\textnormal{tr}({M}_2{M}_2^\prime),
\label{bound M1M2}
\end{align}
where the constant $K$ is given in \autoref{Lemma1}. Following similar techniques when dealing with (\ref{bound power 4}), (\ref{bound M1M2}) can be proved by using Cauchy--Schwarz inequality and \autoref{Lemma1}(3):
\begin{align*}
\mathbb{E}({Z}_i^\prime{M}_1{M}_2{Z}_i)^2\leqslant &\mathbb{E}({Z}_i^\prime{M}_1{M}_1^\prime{Z}_i{Z}_i^\prime{M}_2^\prime {M}_2{Z}_i)\\
\leqslant &\sqrt{\mathbb{E}\left[({Z}_i^\prime{M}_1{M}_1^\prime{Z}_i)^2\right]\mathbb{E}\left[({Z}_i^\prime{M}_2^\prime{M}_2{Z}_i)^2\right]}\\
\leqslant &K^2\textnormal{tr}({M}_1{M}_1^\prime)\textnormal{tr}({M}_2{M}_2^\prime).
\end{align*}
Note that $\textnormal{tr}({\Lambda}_3{\Lambda}_3^\prime)=\textnormal{tr}({W}{\Sigma}_1{W}{\Sigma}{W}{\Sigma}_2{\Sigma}^\prime)$, $\textnormal{tr}({\Lambda}_i^4)=\textnormal{tr}(({W}{\Sigma}_i)^4)$, $i=1,2$. Under Assumptions 3.2--3.3, it follows that
\begin{align}
\begin{split}
&\mathbb{E}\big[({Z}_i^\prime{\Lambda}_3^2{Z}_i)^2\big]\leqslant K^2\textnormal{tr}^2({\Lambda}_3{\Lambda}_3^\prime)=o\Big(\textnormal{tr}^2(({W}{\Sigma}_1)^2)\textnormal{tr}^2(({W}{\Sigma}_2)^2)\Big),\\
&\mathbb{E}\big[({Z}_i^\prime{\Lambda}_1^2{\Lambda}_2^2{Z}_i)^2\big]\leqslant K^2\textnormal{tr}({\Lambda}_1^4)\textnormal{tr}({\Lambda}_2^4)=o\Big(\textnormal{tr}^2(({W}{\Sigma}_1)^2)\textnormal{tr}^2(({W}{\Sigma}_2)^2)\Big).
\end{split}
\label{bound hi 23}
\end{align}

By Cauchy--Schwarz inequality, (\ref{Lemma4 3}), and \autoref{Lemma2}, we obtain
\begin{align}
\begin{split}
&\mathbb{E}\big[\textnormal{tr}^2({\Lambda}_1{Z}_i{Z}_i^\prime{\Lambda}_2\circ {\Lambda}_3)\big]\\
\leqslant&\textnormal{tr}(\Lambda_3\circ \Lambda_3)\mathbb{E}\big[\textnormal{tr}({\Lambda}_1{Z}_i{Z}_i^\prime{\Lambda}_2\circ {\Lambda}_1{Z}_i{Z}_i^\prime{\Lambda}_2) \big]\\
=&\textnormal{tr}(\Lambda_3\circ \Lambda_3)\Big(2\textnormal{tr}({\Lambda}_3\circ {\Lambda}_3)+\textnormal{tr}({\Lambda}_1^2\circ {\Lambda}_2^2)+\Delta\textnormal{tr}(({\Lambda}_1\circ {\Lambda}_1)({\Lambda}_2\circ {\Lambda}_2))\Big)\\
\leqslant &\textnormal{tr}({\Lambda}_3{\Lambda}_3^\prime)\Big(2\textnormal{tr}({\Lambda}_3{\Lambda}_3^\prime)+(1+\Delta)\sqrt{\textnormal{tr}({\Lambda}_1^4)\textnormal{tr}({\Lambda}_2^4)}\Big)\\
=&o\Big(\textnormal{tr}^2(({W}{\Sigma}_1)^2)\textnormal{tr}^2(({W}{\Sigma}_2)^2)\Big).
\end{split}
\label{Lemma5 25}
\end{align}
Combining (\ref{Lemma5 21}), (\ref{Lemma5 22}), (\ref{bound hi 23}), and (\ref{Lemma5 25}), we obtain
\begin{align}
\begin{split}
\mathbb{E}(h_i^2)\leqslant& 4\Big(\textnormal{tr}^2({\Lambda}_3)\mathbb{E}(\xi_i^2)+\mathbb{E}\big[({Z}_i^\prime{\Lambda}_3^2{Z}_i)^2\big]+\mathbb{E}\big[({Z}_i^\prime{\Lambda}_1^2{\Lambda}_2^2{Z}_i)^2\big]+\Delta^2\mathbb{E}\big[\textnormal{tr}^2({\Lambda}_1{Z}_i{Z}_i^\prime{\Lambda}_2\circ {\Lambda}_3)\big]\Big)\\
=&o\Big(\textnormal{tr}^2(({W}{\Sigma}_1)^2)\textnormal{tr}^2(({W}{\Sigma}_2)^2)\Big).
\end{split}
\label{Lemma5 26}
\end{align}
Then, we prove (2) by gathering (\ref{Lemma5 0}), (\ref{Lemma5 1}), (\ref{bound dij}), and (\ref{Lemma5 26}).

Before showing (3), we first establish that, under Assumptions 3.1--3.2, 
\begin{align}
\mathbb{E}({Z}_j^\prime{\Lambda}_i^2{Z}_j)^4=O(\textnormal{tr}^4({\Lambda}_i^2)), \ \text{for $i=1,2$}.
\label{haodingxi}
\end{align}
Note that $\mathbb{E}({Z}_j^\prime{\Lambda}_i^2{Z}_j)=\textnormal{tr}({\Lambda}_i^2)$, we have 
\begin{align}
\begin{split}
\mathbb{E}({Z}_j^\prime{\Lambda}_i^2{Z}_j)^4=&\mathbb{E}\Big[\Big({Z}_j^\prime{\Lambda}_i^2{Z}_j-\textnormal{tr}({\Lambda}_i^2)\Big)+\textnormal{tr}({\Lambda}_i^2)\Big]^4\\
=&\mathbb{E}\Big[\Big({Z}_j^\prime{\Lambda}_i^2{Z}_j-\textnormal{tr}({\Lambda}_i^2)\Big)\Big]^4+4\textnormal{tr}({\Lambda}_i^2)\mathbb{E}\Big[\Big({Z}_j^\prime{\Lambda}_i^2{Z}_j-\textnormal{tr}({\Lambda}_i^2)\Big)\Big]^3\\
&+6\textnormal{tr}^2({\Lambda}_i^2)\mathbb{E}\Big[\Big({Z}_j^\prime{\Lambda}_i^2{Z}_j-\textnormal{tr}({\Lambda}_i^2)\Big)\Big]^2+\textnormal{tr}^4({\Lambda}_i^2).
\end{split}
\label{Lemma5 31}
\end{align}
Under Assumption~3.2, directly applying \autoref{Lemma1}(2), we get 
\begin{align}
\mathbb{E}\Big[\Big({Z}_j^\prime{\Lambda}_i^2{Z}_j-\textnormal{tr}({\Lambda}_i^2)\Big)\Big]^4=O\Big(\textnormal{tr}^2({\Lambda}_i^4)\Big)=o\Big(\textnormal{tr}^4({\Lambda}_i^2)\Big).
\label{Lemma5 32}
\end{align}
Next, by \autoref{Lemma1}(3), we obtain
\begin{align}
\begin{split}
\mathbb{E}\Big[\Big({Z}_j^\prime{\Lambda}_i^2{Z}_j-\textnormal{tr}({\Lambda}_i^2)\Big)\Big]^2\leqslant\mathbb{E}\big[({Z}_j^\prime{\Lambda}_i^2{Z}_j)^2\big]=O\big(\textnormal{tr}^2({\Lambda}_i^2)\big).
\end{split}
\label{Lemma5 33}
\end{align}
Then, gathering (\ref{Lemma5 32}), (\ref{Lemma5 33}), and Cauchy--Schwarz inequality, we get
\begin{align}
\begin{split}
\Bigg|\mathbb{E}\Big[\Big({Z}_j^\prime{\Lambda}_i^2{Z}_j-\textnormal{tr}({\Lambda}_i^2)\Big)\Big]^3\Bigg|\leqslant &\sqrt{\mathbb{E}\Big[\Big({Z}_j^\prime{\Lambda}_i^2{Z}_j-\textnormal{tr}({\Lambda}_i^2)\Big)\Big]^2\mathbb{E}\Big[\Big({Z}_j^\prime{\Lambda}_i^2{Z}_j-\textnormal{tr}({\Lambda}_i^2)\Big)\Big]^4}\\
=&O\big(\textnormal{tr}^3({\Lambda}_i^2)\big).
\end{split}
\label{Lemma5 34}
\end{align}
Therefore, gathering (\ref{Lemma5 31})--(\ref{Lemma5 34}), the results of (\ref{haodingxi}) yield. 

Note that $\mathbb{E}(\Psi_{ii}^2)$ can be bounded by 
\begin{align}
\mathbb{E}(\Psi_{ii}^2)\leqslant 5\mathbb{E}(d_{ii}^2)+10\mathbb{E}(h_i^2)+5\mathbb{E}(\xi_i^{\star 4})+5\mathbb{E}^2(\xi_i^2).
\label{Lemma5 51}
\end{align}
For $\mathbb{E}(d_{ii}^2)$, note that
\begin{align*}
\mathbb{E}(d_{ii}^2)\leqslant 6\Big(\mathbb{E}[({Z}_i^\prime{\Lambda}_3{Z}_i)^4]+\mathbb{E}[({Z}_i^\prime{\Lambda}_1^2{Z}_i)^2({Z}_i^\prime{\Lambda}_2^2{Z}_i)^2]+\Delta^2\mathbb{E}[\textnormal{tr}^2({\Lambda}_1{Z}_i{Z}_i^\prime{\Lambda}_2\circ {\Lambda}_1{Z}_i{Z}_i^\prime{\Lambda}_2)]\Big).
\end{align*}
Using Cauchy--Schwarz inequality and (\ref{Cauchy inequality}), we get
\begin{align}
\begin{split}
({Z}_i^\prime{\Lambda}_3{Z}_i)^4=&\big[\big(({{\Lambda}_1{Z}_i})^\prime({\Lambda}_2{Z}_i)\big)^2\big]^2\leqslant \big[({{\Lambda}_1{Z}_i})^\prime({\Lambda}_1{Z}_i)({{\Lambda}_2{Z}_i})^\prime({\Lambda}_2{Z}_i)\big]^2=({Z}_i^\prime{\Lambda}_1^2{Z}_i)^2({Z}_i^\prime{\Lambda}_2^2{Z}_i)^2
\end{split}
\label{Lemma5 55}
\end{align}
and 
\begin{align*}
\begin{split}
\textnormal{tr}^2({\Lambda}_1{Z}_i{Z}_i^\prime{\Lambda}_2\circ {\Lambda}_1{Z}_i{Z}_i^\prime{\Lambda}_2)\leqslant \textnormal{tr}^2({\Lambda}_1{Z}_i{Z}_i^\prime{\Lambda}_2^2{Z}_i{Z}_i^\prime{\Lambda}_1)=({Z}_i^\prime{\Lambda}_1^2{Z}_i)^2({Z}_i^\prime{\Lambda}_2^2{Z}_i)^2.
\end{split}
\end{align*}
Thus, the order of $\mathbb{E}(d_{ii}^2)$ is 
\begin{align}
\mathbb{E}(d_{ii}^2)=O\Big(\mathbb{E}[({Z}_i^\prime{\Lambda}_1^2{Z}_i)^2({Z}_i^\prime{\Lambda}_2^2{Z}_i)^2]\Big).
\label{Lemma5 52}
\end{align}
By Cauchy--Schwarz inequality and (\ref{haodingxi}), we have
\begin{align}
\begin{split}
\mathbb{E}[({Z}_i^\prime{\Lambda}_1^2{Z}_i)^2({Z}_i^\prime{\Lambda}_2^2{Z}_i)^2]\leqslant& \sqrt{\mathbb{E}[({Z}_i^\prime{\Lambda}_1^2{Z}_i)^4]\mathbb{E}[({Z}_i^\prime{\Lambda}_2^2{Z}_i)^4]}\\
=&O\big(\textnormal{tr}^2({\Lambda}_1^2)\textnormal{tr}^2({\Lambda}_2^2)\big)\\
=&O\Big(\textnormal{tr}^2(({W}{\Sigma}_1)^2)\textnormal{tr}^2(({W}{\Sigma}_2)^2)\Big).
\end{split}
\label{Lemma5 53}
\end{align}
Gathering (\ref{Lemma5 52})--(\ref{Lemma5 53}), we have
\begin{align}
\begin{split}
\mathbb{E}(d_{ii}^2)=&O\Big(\mathbb{E}[({Z}_i^\prime{\Lambda}_1^2{Z}_i)^2({Z}_i^\prime{\Lambda}_2^2{Z}_i)^2]\Big)=O\Big(\textnormal{tr}^2(({W}{\Sigma}_1)^2)\textnormal{tr}^2(({W}{\Sigma}_2)^2)\Big).
\end{split}
\label{Lemma5 54}
\end{align}
For $\mathbb{E}(\xi_i^{\star 4})$, notice that
\begin{align}
\begin{split}
\mathbb{E}(\xi_i^{\star 4})=\mathbb{E}(\xi_i^4)-4Q\mathbb{E}(\xi_i^3)+6Q^2\mathbb{E}(\xi_i^2)-3Q^4.
\label{Lemma5 61}
\end{split}
\end{align}
Note that
\begin{align}
\begin{split}
&\mathbb{E}(\xi_i^{4})=O\Big(\textnormal{tr}^2(({W}{\Sigma}_1)^2)\textnormal{tr}^2(({W}{\Sigma}_2)^2)\Big),\\
&\mathbb{E}(\xi_i^{2})=o\Big(\textnormal{tr}(({W}{\Sigma}_1)^2)\textnormal{tr}(({W}{\Sigma}_2)^2)\Big)
\end{split}
\label{Lemma5 62}
\end{align}
by (\ref{Lemma5 55}), (\ref{Lemma5 53}), and (\ref{Lemma5 1}). In addition, using Cauchy--Schwarz inequality, we get 
\begin{align}
|\mathbb{E}(\xi_i^{3})|\leqslant \sqrt{\mathbb{E}(\xi_i^{2})\mathbb{E}(\xi_i^{4})}=O\Big(\textnormal{tr}^{\frac{3}{2}}(({W}{\Sigma}_1)^2)\textnormal{tr}^{\frac{3}{2}}(({W}{\Sigma}_2)^2)\Big).
\label{Lemma5 63}
\end{align}
Under Assumption~3.3, $Q^2=o\Big(\textnormal{tr}(({W}{\Sigma}_1)^2)\textnormal{tr}(({W}{\Sigma}_2)^2)\Big)$. Then, gathering (\ref{Lemma5 61})--(\ref{Lemma5 63}), we derive that
\begin{align}
\begin{split}
\mathbb{E}(\xi_i^{\star 4})=&\mathbb{E}(\xi_i^4)-4Q\mathbb{E}(\xi_i^3)+6Q^2\mathbb{E}(\xi_i^2)-3Q^4=O\Big(\textnormal{tr}^2(({W}{\Sigma}_1)^2)\textnormal{tr}^2(({W}{\Sigma}_2)^2)\Big).
\end{split}
\label{Lemma5 64}
\end{align}
Finally, gathering (\ref{Lemma5 1}), (\ref{Lemma5 26}), (\ref{Lemma5 51}), (\ref{Lemma5 54}), and (\ref{Lemma5 64}) completes our proof of (3).

For (4), note that $\eta_{ij}^{\star 2}\leqslant 4(\eta_{ij}^2+\xi_i^2+\xi_j^2+Q^2)$. Thus, we have
\begin{align*}
\begin{split}
\eta_{ij}^{\star 4}\leqslant& 16(\eta_{ij}^2+\xi_i^2+\xi_j^2+Q^2)^2\leqslant 64(\eta_{ij}^4+\xi_i^4+\xi_j^4+Q^4),
\end{split}
\end{align*}
which implies that
\begin{align}
\mathbb{E}(\eta_{ij}^{\star 4})\leqslant 64(\mathbb{E}(\eta_{ij}^4)+2\mathbb{E}(\xi_i^4)+Q^4).
\label{Lemma5 71}
\end{align}

Let $C_{1ij}={Z}_i^\prime{\Lambda}_1{Z}_j{Z}_j^\prime{\Lambda}_1{Z}_i-{Z}_j^\prime{\Lambda}_1^2{Z}_j$ and $C_{2ij}={Z}_i^\prime{\Lambda}_2{Z}_j{Z}_j^\prime{\Lambda}_2{Z}_i-{Z}_j^\prime{\Lambda}_2^2{Z}_j$. For $\mathbb{E}(\eta_{ij}^4)$, note that 
\begin{align*}
\eta_{ij}^2=&({Z}_i^\prime{\Lambda}_1{Z}_j{Z}_j^\prime{\Lambda}_2{Z}_i)^2\\
=&({Z}_i^\prime{\Lambda}_1{Z}_j{Z}_j^\prime{\Lambda}_1{Z}_i)({Z}_i^\prime{\Lambda}_2{Z}_j{Z}_j^\prime{\Lambda}_2{Z}_i)\\
=&(C_{1ij}+{Z}_j^\prime{\Lambda}_1^2{Z}_j)(C_{2ij}+{Z}_j^\prime{\Lambda}_2^2{Z}_j)\\
=&C_{1ij}C_{2ij}+C_{1ij}{Z}_j^\prime{\Lambda}_2^2{Z}_j+C_{2ij}{Z}_j^\prime{\Lambda}_1^2{Z}_j+{Z}_j^\prime{\Lambda}_1^2{Z}_j{Z}_j^\prime{\Lambda}_2^2{Z}_j.
\end{align*}
Therefore, $\mathbb{E}(\eta_{ij}^4)$ is bounded by
\begin{align}
\begin{split}
\mathbb{E}(\eta_{ij}^4)\leqslant& 4\Big(\mathbb{E}(C_{1ij}^2C_{2ij}^2)+\mathbb{E}\big(C_{1ij}^2({Z}_j^\prime{\Lambda}_2^2{Z}_j)^2\big)+\mathbb{E}\big(C_{2ij}^2({Z}_j^\prime{\Lambda}_1^2{Z}_j)^2\big)\\
&+\mathbb{E}\big(({Z}_j^\prime{\Lambda}_1^2{Z}_j)^2({Z}_j^\prime{\Lambda}_2^2{Z}_j)^2\big)\Big).
\end{split}
\label{Lemma5 72}
\end{align}
Using Cauchy--Schwarz inequality and \autoref{Lemma1}(2), we have
\begin{align}
\begin{split}
\mathbb{E}(C_{1ij}^2C_{2ij}^2)=&\mathbb{E}[\mathbb{E}(C_{1ij}^2C_{2ij}^2|{Z}_j)]\\
\leqslant& \mathbb{E}\left[\sqrt{\mathbb{E}(C_{1ij}^4|{Z}_j)\mathbb{E}(C_{2ij}^4|{Z}_j)}\right]\\
=&O\Big(\mathbb{E}\big[\textnormal{tr}(({\Lambda}_1{Z}_j{Z}_j^\prime{\Lambda}_1)^2)\textnormal{tr}(({\Lambda}_2{Z}_j{Z}_j^\prime{\Lambda}_2)^2)\big]\Big)\\
=&O\Big(\mathbb{E}[({Z}_j^\prime{\Lambda}_1^2{Z}_j)^2({Z}_j^\prime{\Lambda}_2^2{Z}_j)^2]\Big).
\end{split}
\label{Lemma5 73}
\end{align}
By Lyapunov inequality and \autoref{Lemma1}(2), we get
\begin{align}
\begin{split}
\mathbb{E}\big(C_{1ij}^2({Z}_j^\prime{\Lambda}_2^2{Z}_j)^2\big)=&\mathbb{E}\big[({Z}_j^\prime{\Lambda}_2^2{Z}_j)^2\mathbb{E}\big(C_{1ij}^2|{Z}_j\big)\big]\\
\leqslant& \mathbb{E}\left[({Z}_j^\prime{\Lambda}_2^2{Z}_j)^2\sqrt{\mathbb{E}\big(C_{1ij}^4|{Z}_j\big)}\right]\\
=&O\Big(\mathbb{E}\big[({Z}_j^\prime{\Lambda}_2^2{Z}_j)^2\textnormal{tr}(({\Lambda}_1{Z}_j{Z}_j^\prime{\Lambda}_1)^2))\big]\Big)\\
=&O\Big(\mathbb{E}[({Z}_j^\prime{\Lambda}_1^2{Z}_j)^2({Z}_j^\prime{\Lambda}_2^2{Z}_j)^2]\Big)
\end{split}
\label{Lemma5 74}
\end{align}
and similarly, 
\begin{align}
\mathbb{E}\big(C_{2ij}^2({Z}_j^\prime{\Lambda}_1^2{Z}_j)^2\big)=&O\Big(\mathbb{E}[({Z}_j^\prime{\Lambda}_1^2{Z}_j)^2({Z}_j^\prime{\Lambda}_2^2{Z}_j)^2]\Big).
\label{Lemma5 75}
\end{align}
Gathering (\ref{Lemma5 72})--(\ref{Lemma5 75}) and (\ref{Lemma5 53}), we get
\begin{align}
\begin{split}
\mathbb{E}(\eta_{ij}^4)=&O\Big(\mathbb{E}[({Z}_j^\prime{\Lambda}_1^2{Z}_j)^2({Z}_j^\prime{\Lambda}_2^2{Z}_j)^2]\Big)=O\Big(\textnormal{tr}^2(({W}{\Sigma}_1)^2)\textnormal{tr}^2(({W}{\Sigma}_2)^2)\Big).
\end{split}
\label{Lemma5 76}
\end{align}
Finally, under Assumption~3.3, we prove (4) by gathering (\ref{Lemma5 62}), (\ref{Lemma5 71}), and (\ref{Lemma5 76}).
\end{proof}

\begin{lemma}
\label{Lemma6}
Under Assumptions 3.1--3.3, as $n,p\rightarrow \infty$, 
\begin{align*}
\frac{1}{n(n-1)}\sum_{i\neq j}^n\frac{\eta_{ij}^{\star}}{\sigma_n}\stackrel{d}\longrightarrow \mathcal{N}(0,1).
\end{align*}
\end{lemma}

\begin{proof}[\textbf{Proof}]
for notational brevity, let ${\omega}_i=({X}_i,{Y}_i)$, $i=1,2,\ldots,n$. For $i\neq j$, define $\phi_{ij}=\eta_{ij}^{\star}/(n(n-1))$.
Next, let $W_j$ denote $\sum_{i=1}^{j-1}\phi_{ij}$ and $\mathcal{S}_l=\sum_{j=2}^lW_j$. Let $\mathcal{F}_l$ be the $\sigma$-algebra generated by $\{{\omega}_1,{\omega}_2,\ldots,{\omega}_l\}$. Then, it is obvious to note that $\mathcal{F}_{l}\subset\mathcal{F}_{l+1}$ for $2\leqslant l \leqslant n$, $\mathbb{E}(\mathcal{S}_l)=0$, and $\mathcal{S}_l$ is square-integrable. In addition, for any $l_1>l$,
\begin{align*}
\mathbb{E}(\mathcal{S}_{l_1}|\mathcal{F}_l)=&\mathcal{S}_l+\mathbb{E}\left(\sum_{j=l+1}^{l_1}W_j|\mathcal{F}_l\right)\\
=&\mathcal{S}_l+\mathbb{E}\left(\sum_{j=l+1}^{l_1}\sum_{i=1}^{j-1}\phi_{ij}|\mathcal{F}_l\right)\\
=&\mathcal{S}_l+\sum_{j=l+1}^{l_1}\left[\sum_{i=1}^{l}\mathbb{E}(\phi_{ij}|\mathcal{F}_l)+\sum_{i=l+1}^{j-1}\mathbb{E}(\phi_{ij}|\mathcal{F}_l)\right].
\end{align*}
Since $\mathbb{E}(\phi_{ij})=0, \ \mathbb{E}(\phi_{ij}|{\omega}_i)=0$, we obtain
\begin{align*}
\mathbb{E}(\mathcal{S}_{l_1}|\mathcal{F}_l)=\mathcal{S}_l+\sum_{j=l+1}^{l_1}\left[\sum_{i=1}^{l}\mathbb{E}(\phi_{ij}|\mathcal{F}_l)+\sum_{i=l+1}^{j-1}\mathbb{E}(\phi_{ij}|\mathcal{F}_l)\right]=\mathcal{S}_l,
\end{align*}
which implies that $\{\mathcal{S}_l,\mathcal{F}_l\}_{r=2}^{n}$ is a zero-mean and square-integrable martingale sequence. In order to apply the Central Limit Theorem for martingale difference sequences, we first establish that under Assumptions 3.1--3.3, as $n,p\rightarrow \infty$,
\begin{align}
\frac{1}{\sigma_{n}^2}\sum_{j=2}^{n}\mathbb{E}\left(\left.W_j^2\right|\mathcal{F}_{j-1}\right)\stackrel{p}\longrightarrow\frac{1}{4}.
\label{condition1}
\end{align}

By the Law of Iterated Expectations, we have
\begin{align*}
\mathbb{E}\left[\sum_{j=2}^{n}\mathbb{E}\left(\left.W_j^2\right|\mathcal{F}_{j-1}\right)\right]=&\sum_{j=2}^{n}\mathbb{E}\left(W_j^2\right)=\sum_{j=2}^{n}\sum_{i,t=1}^{j-1}\mathbb{E}(\phi_{ij}\phi_{tj}).
\end{align*}
Note that if $i\neq t$, $\mathbb{E}(\phi_{ij}\phi_{tj})=0$. Therefore, by \autoref{Lemma5}, we obtain 
\begin{equation*}
\begin{split}
\mathbb{E}\left[\sum_{j=2}^{n}\mathbb{E}\left(\left.W_j^2\right|\mathcal{F}_{j-1}\right)\right]=&\frac{1}{n^2(n-1)^2}\sum_{j=2}^{n}\sum_{i=1}^{j-1}\mathbb{E}(\eta_{ij}^{\star 2})\\
=&\frac{1}{n^2(n-1)^2}\sum_{j=2}^{n}\sum_{i=1}^{j-1}\mathbb{E}(\Psi_{ii})\\
=&\frac{\textnormal{tr}(({W}{\Sigma}_1)^2)\textnormal{tr}(({W}{\Sigma}_2)^2)}{2n(n-1)}(1+o(1))\\
=&\frac{\sigma_n^2}{4}(1+o(1)),
\end{split}
\end{equation*}
which further implies 
\begin{align}
\mathbb{E}\left[\frac{1}{\sigma_{n}^2}\sum_{j=2}^{n}\mathbb{E}\left(\left.W_j^2\right|\mathcal{F}_{j-1}\right)\right]=\frac{1}{4}(1+o(1))\rightarrow \frac{1}{4}.
\label{Lemma6 11}
\end{align}

Next, note that the variance of $\sigma_n^{-2}\sum_{j=2}^{n}\mathbb{E}\left(\left.W_j^2\right|\mathcal{F}_{j-1}\right)$ is 
\begin{align}
\begin{split}
&\textnormal{Var}\left[\frac{1}{\sigma_{n}^2}\sum_{j=2}^{n}\mathbb{E}\left(\left.W_j^2\right|\mathcal{F}_{j-1}\right)\right]\\
=&\frac{1}{\sigma_n^4n^4(n-1)^4}\textnormal{Var}\left(\sum_{j=2}^n\sum_{i,t=1}^{j-1}\Psi_{it}\right)\\
=&\frac{1}{\sigma_n^4n^4(n-1)^4}\left[\sum_{j=2}^n\textnormal{Var}\left(\sum_{i,t=1}^{j-1}\Psi_{it}\right)\right.\left.+2\sum_{2\leqslant j_1<j\leqslant n}\textnormal{Cov}\left(\sum_{i,t=1}^{j-1}\Psi_{it},\sum_{i_1,t_1=1}^{j_1-1}\Psi_{i_1t_1}\right)\right].
\end{split}
\label{Lemma6 12}
\end{align}
By direct calculations, we obtain
\begin{align*}
&\sum_{2\leqslant j_1<j\leqslant n}\textnormal{Cov}\left(\sum_{i,t=1}^{j-1}\Psi_{it},\sum_{i_1,t_1=1}^{j_1-1}\Psi_{i_1t_1}\right)\\
&=\sum_{2\leqslant j_1<j\leqslant n} \mathbb{E}\left[\left(\sum_{i\neq t}^{j-1}\Psi_{it}+\sum_{i=1}^{j-1}\big(\Psi_{ii}-\mathbb{E}(\Psi_{ii})\big)\right)\right.\left.\cdot\left(\sum_{i_1\neq t_1}^{j_1-1}\Psi_{i_1t_1}+\sum_{i_1=1}^{j_1-1}\big(\Psi_{i_1i_1}-\mathbb{E}(\Psi_{i_1i_1})\big)\right)\right]\\
=&\sum_{2\leqslant j_1<j\leqslant n}\left(\sum_{i\neq t}^{j-1}\sum_{i_1\neq t_1}^{j_1-1}\mathbb{E}(\Psi_{it}\Psi_{i_1t_1})+\right.+\sum_{i=1}^{j-1}\sum_{i_1\neq t_1}^{j_1-1}\mathbb{E}\big[\big(\Psi_{ii}-\mathbb{E}(\Psi_{ii})\big)\Psi_{i_1t_1}\big]\\
&+\sum_{i\neq t}^{j-1}\sum_{i_1=1}^{j_1-1}\mathbb{E}\big[\big(\Psi_{i_1i_1}-\mathbb{E}(\Psi_{i_1i_1})\big)\Psi_{it}\big]\left.+\sum_{i=1}^{j-1}\sum_{i_1=1}^{j_1-1}\mathbb{E}\big[\big(\Psi_{ii}-\mathbb{E}(\Psi_{ii})\big)\big(\Psi_{i_1i_1}-\mathbb{E}(\Psi_{i_1i_1})\big)\big]\right).
\end{align*}
Recall that $\Psi_{it}=\Psi_{ti}$, and for $i\neq t$, $\mathbb{E}(\Psi_{it}|{Z}_i)=0$. Thus, by \autoref{Lemma5}, the above expression can be further simplified as 
\begin{align}
\begin{split}
\sum_{2\leqslant j_1<j\leqslant n}\textnormal{Cov}\left(\sum_{i,t=1}^{j-1}\Psi_{it},\sum_{i_1,t_1=1}^{j_1-1}\Psi_{i_1t_1}\right)=&
\sum_{2\leqslant j_1<j\leqslant n}\left(2\sum_{i\neq t}^{j_1-1}\mathbb{E}(\Psi_{it}^2)+\sum_{i=1}^{j_1-1}\big(\mathbb{E}(\Psi_{ii}^2)-\mathbb{E}^2(\Psi_{ii})\big)\right)\\
=&O\big(n^4\mathbb{E}(\Psi_{it}^2)+n^{3}\big(\mathbb{E}(\Psi_{ii}^2)-\mathbb{E}^2(\Psi_{ii})\big)\big)\\
=&o(n^8\sigma_n^4)+O(n^7\sigma_n^4).
\end{split}
\label{Lemma6 13}
\end{align}
Similarly, we have
\begin{align}
\begin{split}
\sum_{j=2}^n\textnormal{Var}\left(\sum_{i,t=1}^{j-1}\Psi_{it}\right)=&\sum_{j=2}^n \mathbb{E}\left[\left(\sum_{i\neq t}^{j-1}\Psi_{it}+\sum_{i=1}^{j-1}\big(\Psi_{ii}-\mathbb{E}(\Psi_{ii})\big)\right)^2\right]\\
=&\sum_{j=2}^n\left(\sum_{i\neq t}^{j-1}\sum_{i_1\neq t_1}^{j-1}\mathbb{E}(\Psi_{it}\Psi_{i_1t_1})\right.+2\sum_{i\neq t}^{j-1}\sum_{i_1=1}^{j-1}\mathbb{E}\big[\big(\Psi_{i_1i_1}-\mathbb{E}(\Psi_{i_1i_1})\big)\Psi_{it}\big]\\
&\left.+\sum_{i=1}^{j-1}\sum_{i_1=1}^{j-1}\mathbb{E}\big[\big(\Psi_{ii}-\mathbb{E}(\Psi_{ii})\big)\big(\Psi_{i_1i_1}-\mathbb{E}(\Psi_{i_1i_1})\big)\big]\right)\\
=&\sum_{j=2}^n\left(2\sum_{i\neq t}^{j-1}\mathbb{E}(\Psi_{it}^2)+\sum_{i=1}^{j-1}\big(\mathbb{E}(\Psi_{ii}^2)-\mathbb{E}^2(\Psi_{ii})\big)\right)\\
=&O\big(n^3\mathbb{E}(\Psi_{it}^2)+n^{2}\big(\mathbb{E}(\Psi_{ii}^2)-\mathbb{E}^2(\Psi_{ii})\big)\big)\\
=&O(n^7\sigma_n^4+n^6\sigma_n^4).
\end{split}
\label{Lemma6 14}
\end{align}
Gathering (\ref{Lemma6 12})--(\ref{Lemma6 14}), we get
\begin{align}
\begin{split}
\textnormal{Var}\left[\frac{1}{\sigma_{n}^2}\sum_{j=2}^{n}\mathbb{E}\left(\left.W_j^2\right|\mathcal{F}_{j-1}\right)\right]=&\frac{1}{\sigma_n^4n^4(n-1)^4}\left[\sum_{j=2}^n\textnormal{Var}\left(\sum_{i,t=1}^{j-1}\Psi_{it}\right)\right.\\
&\left.+2\sum_{2\leqslant j_1<j\leqslant n}\textnormal{Cov}\left(\sum_{i,t=1}^{j-1}\Psi_{it},\sum_{i_1,t_1=1}^{j_1-1}\Psi_{i_1t_1}\right)\right]\\
=&\frac{1}{\sigma_n^4n^4(n-1)^4}\Big(o(n^8\sigma_n^4)+O(n^7\sigma_n^4)\Big)\\
\rightarrow& 0.
\end{split}
\label{Lemma6 15}
\end{align}
By (\ref{Lemma6 11}) and (\ref{Lemma6 15}), the results of (\ref{condition1}) yield. 

Next, we establish that under Assumptions 3.1--3.3, $\forall \varepsilon>0$, as $n,p\rightarrow\infty$,
\begin{align}
\frac{1}{\sigma_{n}^2}\sum_{j=2}^{n}\mathbb{E}\left[\left.W_j^2\mathbf{1}(\left|W_j\right|>\varepsilon\sigma_{n})\right|\mathcal{F}_{j-1}\right]\stackrel{p}\longrightarrow0.
\label{Lemm5 21}
\end{align}

By Markov's inequality, it holds that
\begin{align}
\frac{1}{\sigma_{n}^2}\sum_{j=2}^{n}\mathbb{E}\left[\left.W_j^2\mathbf{1}(\left|W_j\right|>\varepsilon\sigma_{n})\right|\mathcal{F}_{j-1}\right]\leqslant \frac{1}{\varepsilon^2\sigma_{n}^4}\sum_{j=2}^{n}\mathbb{E}\left(\left.W_j^4\right|\mathcal{F}_{j-1}\right).
\label{Lemm5 22}
\end{align}
Note that the sequence $\{\sigma_n^{-4}\sum_{j=2}^{n}\mathbb{E}(\left.W_j^4\right|\mathcal{F}_{j-1})\}_{n=2}^\infty$ is nonnegative. Therefore, to show that $$\sigma_n^{-4}\sum_{j=2}^{n}\mathbb{E}(\left.W_j^4\right|\mathcal{F}_{j-1})=o_p(1),$$
we only need to prove that its expectation converges to 0 as $n,p\rightarrow\infty$.

By direct calculations and \autoref{Lemma5}, we obtain
\begin{align}
\begin{split}
\mathbb{E}\left[\frac{1}{\sigma_n^4}\sum_{j=2}^{n}\mathbb{E}\left(\left.W_j^4\right|\mathcal{F}_{j-1}\right)\right]=&\frac{1}{\sigma_n^4}\sum_{j=2}^{n}\mathbb{E}(W_j^4)\\
=&\frac{1}{\sigma_n^4n^4(n-1)^4}\sum_{j=2}^{n}\sum_{i_1,i_2,i_3,i_4=1}^{j-1}\mathbb{E}(\eta_{i_1j}^{\star}\eta_{i_2j}^{\star}\eta_{i_3j}^{\star}\eta_{i_4j}^{\star})\\
=&\frac{1}{\sigma_n^4n^4(n-1)^4}\left(3\sum_{j=2}^{n}\sum_{i\neq t}^{j-1}\mathbb{E}(\eta_{ij}^{\star 2}\eta_{tj}^{\star 2})+\sum_{j=2}^{n}\sum_{i=1}^{j-1}\mathbb{E}(\eta_{ij}^{\star 4})\right)\\
=&\frac{1}{\sigma_n^4n^4(n-1)^4}\left(3\sum_{j=2}^{n}\sum_{i\neq t}^{j-1}\mathbb{E}(\Psi_{jj}^2)+\sum_{j=2}^{n}\sum_{i=1}^{j-1}\mathbb{E}(\eta_{ij}^{\star 4})\right)\\
=&\frac{1}{\sigma_n^4n^4(n-1)^4}\Big(O(n^7\sigma_n^4)+O(n^6\sigma_n^4)\Big)\\
\rightarrow& 0.
\end{split}
\label{Lemma6 23}
\end{align}
Then, (\ref{Lemm5 21}) yields by gathering (\ref{Lemm5 22})--(\ref{Lemma6 23}).

With  (\ref{condition1}) and (\ref{Lemm5 21}) in hand, by Corollary 3.1 of Hall and Heyde (1980), we have
\begin{align*}
\frac{1}{n(n-1)}\sum_{i\neq j}^n\frac{\eta_{ij}^{\star}}{\sigma_n}=\frac{2}{\sigma_{n}}\sum_{j=2}^{n}W_j\stackrel{d}\longrightarrow\mathcal{N}(0,1),
\end{align*}
which completes our proof.
\end{proof}

\newpage

\section{Proof of main results}
\label{Proof of main results}
\setcounter{equation}{0}

\begin{proof}[\textbf{Proof of Theorem~2.1}]
It can be easily verified that
\begin{align*}
\int_{\mathbb{R}^p}{\xi}{\xi}^\prime w({\xi})d{\xi}={B}+{a}{a}^\prime={W}.
\end{align*}
Then, with direct calculations, we obtain
\begin{align*}
Q=&\int_{\mathbb{R}^p}\int_{\mathbb{R}^p} ({\xi}^\prime {\Sigma}{\eta})^2 w({\xi})w({\eta})d{\xi}d{\eta}\\
=&\int_{\mathbb{R}^p}\int_{\mathbb{R}^p} ({\xi}^\prime {\Sigma}{\eta}{\eta}^\prime{\Sigma}^\prime{\xi}) w({\xi})w({\eta})d{\xi}d{\eta}\\
=&\textnormal{tr}\left[\left(\int_{\mathbb{R}^p}{\xi}{\xi}^\prime {\Sigma}w({\xi})d{\xi}\right)\left(\int_{\mathbb{R}^p}{\eta}{\eta}^\prime {\Sigma}^\prime w({\eta})d{\eta}\right)\right]\\
=&\textnormal{tr}({W}{\Sigma}{W}{\Sigma}^\prime),
\end{align*}
which completes our proof.  
\end{proof}

\begin{proof}[\textbf{Proof of Lemma~3.1}]
By the standard Hoeffding decomposition (Lee, 1990, p.26), we have
\begin{align}
\begin{split}
\tilde{Q}_{n1}=&Q+\frac{2}{n}\sum_{i=1}^n\xi_i^\star+\frac{1}{n(n-1)}\sum_{i\neq j}^n\eta_{ij}^\star.
\end{split}
\label{decomposition}
\end{align}
Then, using the variance formula for Hoeffding decomposition (Lee, 1990, p.30), we get
\begin{align}
\begin{split}
\textnormal{Var}(\tilde{Q}_{n1})=&\frac{4}{n}\mathbb{E}(\xi_i^{\star 2})+\frac{2}{n(n-1)}\mathbb{E}(\eta_{ij}^{\star 2}).
\end{split}
\label{variance_Q_n1}
\end{align}
Gathering (\ref{variance_Q_n1}) and \autoref{Lemma4}, the results of $\textnormal{Var}(\tilde{Q}_{n1})$ directly follows.

Similarly, for $\tilde{Q}_{n2}$, the Hoeffding decomposition (Lee, 1990, p.30) implies that
\begin{align*}
\tilde{Q}_{n2}=\frac{1}{n(n-1)}\sum_{i\neq j}^n{Z}_i{\Lambda}_3{Z}_j+\frac{1}{n(n-1)(n-2)}\sum_{i\neq j\neq k}^n({Z}_i{\Lambda}_1{Z}_k{Z}_k^\prime{\Lambda}_2{Z}_j-{Z}_i{\Lambda}_3{Z}_j).
\end{align*}
Using the variance formula for Hoeffding decomposition (Lee, 1990, p.30), it yields that
\begin{align}
\textnormal{Var}(\tilde{Q}_{n2})=O\big(n^{-2}\mathbb{E}[({Z}_i{\Lambda}_3{Z}_j)^2]+n^{-3}\mathbb{E}[({Z}_i{\Lambda}_1{Z}_k{Z}_k^\prime{\Lambda}_2{Z}_j)^2]\big).
\label{var Q n2}
\end{align}
With straightforward calculations, we get
\begin{align}
\mathbb{E}[({Z}_i{\Lambda}_3{Z}_j)^2]=\textnormal{tr}({\Lambda}_3{\Lambda}_3^\prime)=\textnormal{tr}({W}{\Sigma}_1{W}{\Sigma}{W}{\Sigma}_2{\Sigma}^\prime).
\label{var Q n2 1}
\end{align}
By \autoref{Lemma1}(1), we have
\begin{align}
\begin{split}
\mathbb{E}[({Z}_i{\Lambda}_1{Z}_k{Z}_k^\prime{\Lambda}_2{Z}_j)^2]=&\mathbb{E}[\mathbb{E}({Z}_k^\prime{\Lambda}_1{Z}_i{Z}_i^\prime{\Lambda}_1{Z}_k{Z}_k^\prime{\Lambda}_2{Z}_j{Z}_j^\prime{\Lambda}_2{Z}_k|{Z}_k)]\\
=&\mathbb{E}({Z}_k^\prime{\Lambda}_1^2{Z}_k{Z}_k^\prime{\Lambda}_2^2{Z}_k)\\
=&\textnormal{tr}({\Lambda}_1^2)\textnormal{tr}({\Lambda}_2^2)+2\textnormal{tr}({\Lambda}_1^2{\Lambda}_2^2)+\Delta\textnormal{tr}({\Lambda}_1^2\circ {\Lambda}_2^2).
\end{split}
\label{var Q n2 2}
\end{align}
Note that $\textnormal{tr}({\Lambda}_1^2{\Lambda}_2^2)=\textnormal{tr}({\Lambda}_3{\Lambda}_3^\prime)$. Then, gathering (\ref{var Q n2})--(\ref{var Q n2 2}) yields the variance expression of $\tilde{Q}_{n2}$. 

Finally, note that $\tilde{Q}_{n3}$ is a $U$-statistic with third-order degeneracy. Therefore, using the variance formula for Hoeffding decomposition (Lee, 1990, p.30), we get
\begin{align}
\textnormal{Var}(\tilde{Q}_{n3})=O\big(n^{-4}\mathbb{E}({Z}_i^\prime{\Lambda}_1{Z}_j{Z}_k^\prime{\Lambda}_2{Z}_l)^2\big).
\label{var Q n3}
\end{align}
With direct calculations, we obtain
\begin{align}
\begin{split}
\mathbb{E}({Z}_i^\prime{\Lambda}_1{Z}_j{Z}_k^\prime{\Lambda}_2{Z}_l)^2=&\mathbb{E}({Z}_i^\prime{\Lambda}_1{Z}_j{Z}_j^\prime{\Lambda}_1{Z}_i)\mathbb{E}({Z}_k^\prime{\Lambda}_2{Z}_l{Z}_l^\prime{\Lambda}_2{Z}_k)\\
=&\textnormal{tr}({\Lambda}_1^2)\textnormal{tr}({\Lambda}_2^2)\\
=&\textnormal{tr}(({W}{\Sigma}_1)^2)\textnormal{tr}(({W}{\Sigma}_2)^2).
\end{split}
\label{var Q n3 1}
\end{align}
Then, the variance expression of $\tilde{Q}_{n3}$ yields by gathering (\ref{var Q n3})--(\ref{var Q n3 1}).
\end{proof}

\begin{proof}[\textbf{Proof of Theorem~3.1}]
Note that under $\mathcal{H}_0$, $Q=\xi_i^\star=0$, therefore (\ref{decomposition}) implies that
\begin{align}
\tilde{Q}_{n1}=\frac{1}{n(n-1)}\sum_{i\neq j}^n\eta_{ij}^\star.
\end{align}
In addition, for $i=2,3$, straightforward calculations show that $\mathbb{E}(\tilde{Q}_{ni})=0$, and $\textnormal{Var}(\tilde{Q}_{ni})=o(\sigma_n^2)$ can be derived by (3.3). Therefore, we have $\tilde{Q}_{ni}/\sigma_n=o_p(1)$ for $i=2,3$. Note that Assumption~3.3 is trivially satisfied under $\mathcal{H}_0$. Then, gathering \autoref{Lemma6} and the analysis above, the desired results yield:
\begin{align*}
\frac{Q}{\sigma_n}=\frac{\tilde{Q}_{n1}}{\sigma_n}+\frac{-2\tilde{Q}_{n2}+\tilde{Q}_{n3}}{\sigma_n}=\frac{1}{n(n-1)}\sum_{i\neq j}^n\eta_{ij}^\star+o_p(1)\stackrel{d}\longrightarrow \mathcal{N}(0,1), 
\end{align*}
where the last step follows using Slutsky's theorem.    
\end{proof}

\begin{proof}[\textbf{Proof of Lemma~3.2}]
To show the ratio consistency, by Slutsky's theorem, we only need to verify $\hat{V}_x/V_x\stackrel{p}\longrightarrow 1, \ \hat{V}_y/V_y\stackrel{p}\longrightarrow 1.$

Following $\hat{Q}_n$, note that $\hat{V}_x$ can be centralized as $\hat{V}_x=\tilde{V}_{x1}-2\tilde{V}_{x2}+\tilde{V}_{x3}$,
where 
\begin{align*}
&\tilde{V}_{x1}=\frac{1}{n(n-1)}\sum_{i\neq j} ({A}_i^\prime{W}{A}_j)^2,\ \tilde{V}_{x2}=\frac{1}{n(n-1)(n-2)}\sum_{i\neq j\neq k} {A}_i^\prime{W}{A}_j{A}_j^\prime{W}{A}_k,\\
&\tilde{V}_{x3}=\frac{1}{n(n-1)(n-2)(n-3)}\sum_{i\neq j\neq k\neq l} {A}_i^\prime{W}{A}_j{A}_k^\prime{W}{A}_l.
\end{align*}
It is obvious to note that 
\begin{align*}
&\mathbb{E}(\tilde{V}_{x1})=\mathbb{E}({Z}_i{\Lambda}_1{Z}_j{Z}_j^\prime{\Lambda}_1{Z}_i)=\textnormal{tr}({\Lambda}_1^2)=V_x,\,\,\,\mathbb{E}(\tilde{V}_{x2})=\mathbb{E}(\tilde{V}_{x3})=0.
\end{align*}
Therefore, $\hat{V}_x$ is an unbiased estimator for $V_x$. Note that each $\tilde{V}_{xi}$ can be derived simply by replacing ${\Lambda}_2$ with ${\Lambda}_1$ in $\tilde{Q}_{ni}$, $i=1,2,3$. Thus, the variance of each $\tilde{V}_{xi}$ can be calculated by replacing ${\Sigma}$ and ${\Sigma}_2$ with ${\Sigma}_1$ in (Lemma~3.1). As a result, we get
\begin{flalign}
\begin{split}
\textnormal{Var}(\tilde{V}_{x1})=&O\Big(n^{-1}\textnormal{tr}\big(({W}{\Sigma}_1)^4+{\Lambda}_1^2\circ {\Lambda}_1^2\big)+n^{-2}\big[\textnormal{tr}^2(({W}{\Sigma}_1)^2)+\textnormal{tr}({\Lambda}_1^2\circ {\Lambda}_1^2+({\Lambda}_1\circ {\Lambda}_1)^2)\big]\Big),\\
\textnormal{Var}(\tilde{V}_{x2})=&O\Big(n^{-2}\textnormal{tr}\big(({W}{\Sigma}_1)^4\big)+n^{-3}\big[\textnormal{tr}^2\big(({W}{\Sigma}_1)^2\big)+\textnormal{tr}({\Lambda}_1^2\circ {\Lambda}_1^2)\big]\Big),\\
\textnormal{Var}(\tilde{V}_{x3})=&O\Big(n^{-4}\textnormal{tr}^2\big(({W}{\Sigma}_1)^2\big)\Big).
\end{split}
\label{var Vx}
\end{flalign}
By (\ref{Cauchy inequality}), we derive that
\begin{align*}
&\textnormal{tr}({\Lambda}_1^2\circ {\Lambda}_1^2)\leqslant \textnormal{tr}({\Lambda}_1^4)=\textnormal{tr}\big(({W}{\Sigma}_1)^4\big),\\
&\textnormal{tr}\big(({\Lambda}_1\circ {\Lambda}_1)^2\big)=\sum_{j,k=1}^{d}\lambda_{1,kj}^2\lambda_{1,jk}^2\leqslant \sum_{j,k,l=1}^{d}\lambda_{1,kj}^2\lambda_{1,lk}^2=\textnormal{tr}({\Lambda}_1^2\circ {\Lambda}_1^2)\leqslant \textnormal{tr}({\Lambda}_1^4)=\textnormal{tr}\big(({W}{\Sigma}_1)^4\big).
\end{align*}
Therefore, with Assumption~3.2, (\ref{var Vx}) can be further simplified as 
\begin{align}
\textnormal{Var}(\tilde{V}_{x1})=O(n^{-1}V_x^2), \ \textnormal{Var}(\tilde{V}_{x2})=O(n^{-2}V_x^2),\ \textnormal{Var}(\tilde{V}_{x3})=O(n^{-4}V_x^2).
\label{var Vx s}
\end{align}
By Cauchy--Schwarz inequality and (\ref{var Vx s}), we obtain
\begin{align*}
\begin{split}
\textnormal{Var}(\hat{V}_x/V_x)=O\Big(V_x^{-2}\big[\textnormal{Var}(\tilde{V}_{x1})+\textnormal{Var}(\tilde{V}_{x2})+\textnormal{Var}(\tilde{V}_{x3})\big]\Big)\rightarrow 0.
\end{split}
\end{align*}
Since $\mathbb{E}(\hat{V}_x/V_x)=1$, we get $\hat{V}_x/V_x\stackrel{p}\longrightarrow 1$. Using similar reasoning, we obtain $\hat{V}_y/V_y\stackrel{p}\longrightarrow 1$. Gathering these two arguments completes our proof.     
\end{proof}

\begin{proof}[\textbf{Proof of Theorem~3.2}]
Taking Theorem~3.1, Lemma~3.2, and Slutsky's theorem, the results of Theorem~3.2 follow directly. 
\end{proof}

\begin{proof}[\textbf{Proof of Theorem~3.3}]
For the Hoeffding decomposition of $\tilde{Q}_{n1}$ in (\ref{decomposition}), note that $\textnormal{tr}({\Lambda}_3\circ {\Lambda}_3)\leqslant \textnormal{tr}({\Lambda}_3{\Lambda}_3^\prime)=\textnormal{tr}({W}{\Sigma}_1{W}{\Sigma}{W}{\Sigma}_2{W}{\Sigma}^\prime)$ by \autoref{Lemma2} and  $\textnormal{tr}(({W}{\Sigma}{W}{\Sigma}^\prime)^2)=\textnormal{tr}({\Lambda}_3^2)\leqslant \textnormal{tr}({\Lambda}_3{\Lambda}_3^\prime)=\textnormal{tr}({W}{\Sigma}_1{W}{\Sigma}{W}{\Sigma}_2{W}{\Sigma}^\prime)$ by (\ref{Cauchy inequality}). Therefore, under Assumption~3.3, we have
\begin{align*}
\textnormal{Var}\left(\frac{2}{n}\sum_{i=1}^n\xi_i^\star\right)=\frac{4}{n}\mathbb{E}(\xi_i^{\star 2})\leqslant\frac{4(2+\Delta)}{n}\textnormal{tr}({W}{\Sigma}_1{W}{\Sigma}{W}{\Sigma}_2{W}{\Sigma}^\prime)=o(\sigma_n^2).
\end{align*}
Since $\mathbb{E}(\xi_i^\star)=0$, we have $2/n\sum_{i=1}^n(\xi_i^\star/\sigma_n)=o_p(1)$. for $i=2,3$, note that $\mathbb{E}(\tilde{Q}_{ni})=0$ and $\textnormal{Var}(\tilde{Q}_{ni})=o(\sigma^2_n)$ remains true by Assumption~3.3 and (3.3). Therefore, by (\ref{decomposition}), it holds that
\begin{align}
\begin{split}
\frac{\hat{Q}_n-Q}{\sigma_n}=&\frac{\tilde{Q}_{n1}-Q}{\sigma_n}+\frac{-2\tilde{Q}_{n2}+\tilde{Q}_{n3}}{\sigma_n}\\
=&\frac{1}{n(n-1)}\sum_{i\neq j}^n\frac{\eta_{ij}^\star}{\sigma_n}+\frac{2}{n}\sum_{i=1}^n\frac{\xi_i^\star}{\sigma_n}+\frac{-2\tilde{Q}_{n2}+\tilde{Q}_{n3}}{\sigma_n}\\
=&\frac{1}{n(n-1)}\sum_{i\neq j}^n\frac{\eta_{ij}^\star}{\sigma_n}+o_p(1)\\
\stackrel{d}\longrightarrow& \mathcal{N}(0,1).
\end{split}
\label{normality}
\end{align}

Next, note that
\begin{align*}
\mathbb{P}\left\{\hat{T}_n>z_{1-\alpha}\right\} = \mathbb{P}\left\{\frac{\hat{Q}_n - Q}{\sigma_n} - \left(\frac{\hat{\sigma}_n}{\sigma_n}-1\right)z_{1-\alpha}>z_{1-\alpha} - G_n(W;\Gamma_1,\Gamma_2)\right\}.
\end{align*}
By (\ref{normality}) and Slutsky's theorem, it follows that
\begin{align*}
\frac{\hat{Q}_n - Q}{\sigma_n} - \left(\frac{\hat{\sigma}_n}{\sigma_n}-1\right)z_{1-\alpha}\stackrel{d}\longrightarrow \mathcal{N}(0,1).
\end{align*}
Thus, by P{\'o}lya's Theorem (P{\'o}lya, 1930), we have
\begin{align*}
\sup_{x\in\mathbb{R}}\left|\mathbb{P}\left\{\frac{\hat{Q}_n - Q}{\sigma_n} - \left(\frac{\hat{\sigma}_n}{\sigma_n}-1\right)z_{1-\alpha}>x\right\} - \Phi(-x)\right|\rightarrow 0,\,\,\,n,p\rightarrow\infty,
\end{align*}
which implies (3.8) in Theorem 3.3.
\end{proof}

\begin{proof}[\textbf{Proof of Theorem~3.4}]
By direct calculations, we obtain
\begin{align*}
\begin{split}
&\textnormal{tr}(\bm{W}\bm{\Sigma}\bm{W}\bm{\Sigma}^\prime)=b^4\textnormal{tr}(\bm{\Sigma}\bm{\Sigma}^\prime)+a^2b^2\mathbf{1}_p^\prime(\bm{\Sigma}\bm{\Sigma}^\prime+\bm{\Sigma}^\prime\bm{\Sigma})\mathbf{1}_p+a^4(\mathbf{1}_p^\prime\bm{\Sigma}\mathbf{1}_p)^2,\\
&\textnormal{tr}((\bm{W}\bm{\Sigma}_i)^2)=b^4\textnormal{tr}(\bm{\Sigma}_i^2)+2a^2b^2\mathbf{1}_p^\prime\bm{\Sigma}_i^2\mathbf{1}_p+a^4(\mathbf{1}_p^\prime\bm{\Sigma}_i\mathbf{1}_p)^2,
\end{split}
\end{align*}
Then, it follows that
\begin{align*}
\frac{\textnormal{tr}^2({W}{\Sigma}{W}{\Sigma}^\prime)}{\textnormal{tr}(({W}{\Sigma}_1)^2)\textnormal{tr}(({W}{\Sigma}_2)^2)}=\frac{\textnormal{tr}^2({\Sigma}{\Sigma}^\prime)}{\textnormal{tr}({\Sigma}_1^2)\textnormal{tr}({\Sigma}_2^2)}\frac{\left(1+r^2\frac{\mathbf{1}_p^\prime({\Sigma}{\Sigma}^\prime+{\Sigma}^\prime{\Sigma})\mathbf{1}_p}{\textnormal{tr}({\Sigma}{\Sigma}^\prime)}+r^4\frac{(\mathbf{1}_p^\prime{\Sigma}\mathbf{1}_p)^2}{\textnormal{tr}({\Sigma}{\Sigma}^\prime)}\right)^2}{\prod_{i=1}^2\left(1+2r^2\frac{\mathbf{1}_p^\prime{\Sigma}_i^2\mathbf{1}_p}{\textnormal{tr}({\Sigma}_i^2)}+r^4\frac{(\mathbf{1}_p^\prime{\Sigma}_i\mathbf{1}_p)^2}{\textnormal{tr}({\Sigma}_i^2)}\right)},
\end{align*}
implying that
\begin{align}
\textnormal{ARE}({W},{I}_p)=&\frac{\left(1+r^2\frac{\mathbf{1}_p^\prime({\Sigma}{\Sigma}^\prime+{\Sigma}^\prime{\Sigma})\mathbf{1}_p}{\textnormal{tr}({\Sigma}{\Sigma}^\prime)}+r^4\frac{(\mathbf{1}_p^\prime{\Sigma}\mathbf{1}_p)^2}{\textnormal{tr}({\Sigma}{\Sigma}^\prime)}\right)^2}{\prod_{i=1}^2\left(1+2r^2\frac{\mathbf{1}_p^\prime{\Sigma}_i^2\mathbf{1}_p}{\textnormal{tr}({\Sigma}_i^2)}+r^4\frac{(\mathbf{1}_p^\prime{\Sigma}_i\mathbf{1}_p)^2}{\textnormal{tr}({\Sigma}_i^2)}\right)}.
\label{ARE}
\end{align}

By the Cauchy--Schwarz inequality, we obtain
\begin{align}
(1_p'\Sigma1_p)^2\leqslant p 1_p'\Sigma'\Sigma1_p,\,\,\, (1_p'\Sigma_i1_p)^2\leqslant p 1_p'\Sigma_i^21_p,\,\,\,i=1,2.
\label{inequality 3.4}
\end{align}
Furthermore, note that $1_p'\Sigma'\Sigma1_p\leqslant p\textnormal{tr}(\Sigma\Sigma')$, $1_p'\Sigma\Sigma'1_p\leqslant p\textnormal{tr}(\Sigma\Sigma')$, and $1_p'\Sigma_i^21_p\leqslant p\textnormal{tr}(\Sigma_i^2)$ for $i=1,2$. Thus, if $pr^2\rightarrow 0$, we get
\begin{align}
\begin{split}
&r^2\frac{\mathbf{1}_p^\prime({\Sigma}{\Sigma}^\prime+{\Sigma}^\prime{\Sigma})\mathbf{1}_p}{\textnormal{tr}({\Sigma}{\Sigma}^\prime)}+r^4\frac{(\mathbf{1}_p^\prime{\Sigma}\mathbf{1}_p)^2}{\textnormal{tr}({\Sigma}{\Sigma}^\prime)}\leqslant 2pr^2 + p^2r^4\rightarrow 0,\\
&2r^2\frac{\mathbf{1}_p^\prime{\Sigma}_i^2\mathbf{1}_p}{\textnormal{tr}({\Sigma}_i^2)}+r^4\frac{(\mathbf{1}_p^\prime{\Sigma}_i\mathbf{1}_p)^2}{\textnormal{tr}({\Sigma}_i^2)}\leqslant 2pr^2 + p^2r^4\rightarrow 0.
\end{split}
\label{dominate}
\end{align}
Gathering (\ref{ARE}) and (\ref{dominate}), we have
\begin{align*}
\lim_{n,p\rightarrow\infty}\textnormal{ARE}({W},{I}_p)=1
\end{align*}
provided that $pr^2\rightarrow 0$.

Next, note that
\begin{align*}
\textnormal{ARE}({W},{I}_p)\geqslant&\frac{\left[1+p^{-1}\left(2r^2+r^4\right)\frac{(\mathbf{1}_p^\prime{\Sigma}\mathbf{1}_p)^2}{\textnormal{tr}({\Sigma}{\Sigma}^\prime)}\right]^2}{\prod_{i=1}^2\left(1+(2r^2+pr^4)\frac{\mathbf{1}_p^\prime{\Sigma}_i^2\mathbf{1}_p}{\textnormal{tr}({\Sigma}_i^2)}\right)}
\end{align*}
by (\ref{ARE}) and (\ref{inequality 3.4}). Thus, under the conditions given in Theorem 3.4 (2), we have 
\begin{align*}
\lim_{n,p\rightarrow\infty}\textnormal{ARE}({W},{I}_p)>1,
\end{align*}
which ends the proof of Theorem 3.4.
\end{proof}

\section*{References}

\noindent Chen, S. X., L.-X. Zhang, and P.-S. Zhong (2010). Tests for high-dimensional covariance matrices.\textit{Journal of the American Statistical Association 105} (490), 810–819.

\noindent Hall, P. and C. C. Heyde (1980). \textit{Martingale limit theory and its application}. Academic Press.

\noindent Lee, A. J. (1990). \textit{U-statistics: Theory and Practice}. Routledge.

\noindent P{\'o}lya, G. (1930). Sur quelques points de la th{\'e}orie des probabilit{\'e}s. \textit{Annales de l'Institut Henri Poincar{\'e}}, \textbf{1}, 117--161.

\end{document}